%% file: main.tex
\documentclass[11pt,letterpaper]{article}

\input{header.tex}
\usepackage{abstract}
\newcommand{\FormatAuthor}[3]{
\begin{tabular}{c}
#1 \\ {\small\texttt{#2}} \\ {\small #3}
\end{tabular}
}

\begin{document}
\pagenumbering{roman}

\title{Efficient Recovery of Latent Coordinate Structure\texorpdfstring{\\}{ }from Sparse Observations of the Hypercube}
\author{
\begin{tabular}[h!]{cc}
  \FormatAuthor{Rares-Darius Buhai}{rares-darius.buhai@epfl.ch}{EPFL}&
  \FormatAuthor{Davide Mazzali}{davide.mazzali@epfl.ch}{EPFL}
\end{tabular}\\\\
\begin{tabular}[h!]{c}
  \FormatAuthor{Weronika Wrzos-Kaminska}{weronika.wrzos-kaminska@epfl.ch}{EPFL}
\end{tabular}
}
\date{}
\hypersetup{pdfauthor={Rares-Darius Buhai, Davide Mazzali, Weronika Wrzos-Kaminska}}
\maketitle

\begin{abstract}
   Recovering latent geometric structure from graph observations is a well-studied problem in statistical inference. Kapralov, Trevisan, and Wrzos-Kaminska (2026) introduced the problem of recovering the coordinate structure of the Boolean hypercube from a small random sample of its edges. More specifically, there are $n=2^d$ vertices, each corresponding to a distinct ``feature'' vector in $\{\pm 1\}^d$. Between each pair whose feature vectors are at Hamming distance one, an edge is observed independently with probability $p$. As long as the expected degree $pd$ is $\gtrsim \log d = \log \log n$, we give a polynomial-time algorithm that, given only the graph of observed edges, correctly recovers the entire feature vector of all but a vanishing fraction of vertices. This matches the information-theoretic guarantee of Kapralov, Trevisan, and Wrzos-Kaminska in polynomial rather than exponential time, resolving the algorithmic question left open by their work.

   Our algorithm combines a degree-$4$ sum-of-squares certificate for the structure of balanced near-minimum cuts with a rounding scheme originally developed for tensor decomposition by Ma, Shi, and Steurer (2016). The analysis relies on two novel ingredients: a sum-of-squares version of the Friedgut-Kalai-Naor theorem in Boolean Fourier analysis and a spectral concentration result for the observed subgraph of the hypercube. Finally, we provide a justification for why higher-degree sum-of-squares might be needed by showing a limitation of the basic SDP relaxation of this problem.
\end{abstract}
\thispagestyle{empty}

\clearpage
\setcounter{tocdepth}{2}

\makeatletter
\@starttoc{toc}
\makeatother
\thispagestyle{empty}

\clearpage

\pagenumbering{arabic}

\input{introduction.tex}

\input{technical-overview.tex}

\input{preliminaries.tex}

\input{sos-fkn.tex}

\input{spectral-concentration.tex}

\input{coordinate-recovery.tex}

\input{degree-2-lower-bound.tex}

\section*{Acknowledgements}
\addcontentsline{toc}{section}{Acknowledgements}
We thank Michael Kapralov and Pravesh K. Kothari for helpful discussions on this project.

The authors used AI systems as research assistants during this project, specifically several versions of OpenAI's GPT Pro models.
GPT 5.4 and 5.5 suggested the main ideas behind two components of the paper: the spectral approximation in \Cref{lem:overview-spectral-concentration} and the degree-$2$ lower bound construction in \Cref{thm:overview-degree-2-lower-bound}.
The authors selected, simplified, wrote formally, and independently verified these arguments, and integrated them into the overall framework of the paper.
Elsewhere, models from GPT 5.4 through 6 were used to stress-test proofs, suggest local improvements, identify and explain related work, edit exposition, adjust formatting, and clean the LaTeX source files.
The authors assume responsibility for all content.

\bibliographystyle{alphaurl}
\bibliography{references}

\end{document}

%% file: header.tex
\usepackage[utf8]{inputenc}
\usepackage[T1]{fontenc}
\usepackage[full]{textcomp}
\usepackage[american]{babel}
\usepackage{amssymb,mathtools,amsthm,bbm,turnstile}
\usepackage[letterpaper,top=1.2in,bottom=1.2in,left=1in,right=1in]{geometry}
\usepackage{newpxtext}
\usepackage[varg,bigdelims]{newpxmath}

\usepackage{bm}
\let\mathbb\varmathbb
\usepackage{microtype}
\usepackage{xcolor}
\usepackage[hyperindex=true]{hyperref}
\definecolor{linkborder}{rgb}{0.8,0.8,1}
\hypersetup{
  pdfpagemode=UseOutlines,
  bookmarksnumbered=true,
  bookmarksopen=true,
  bookmarksopenlevel=2,
  pdfstartview=FitH,
  pdfborder={0 0 1},
  linkbordercolor=linkborder,
  citebordercolor=linkborder,
  urlbordercolor=linkborder
}
\makeatletter
\AddToHook{cmd/maketitle/before}{\hypersetup{pdftitle={\@title}}\pdfbookmark[1]{Abstract}{abstract}}
\makeatother
\usepackage{etoolbox}
\apptocmd{\thebibliography}{\addcontentsline{toc}{section}{\refname}}{}{}
\usepackage[capitalise,nameinlink]{cleveref}
\usepackage{thmtools}

\allowdisplaybreaks
\numberwithin{equation}{section}

\newtheorem{theorem}{Theorem}[section]
\newtheorem{definition}[theorem]{Definition}
\newtheorem{lemma}[theorem]{Lemma}
\newtheorem{claim}[theorem]{Claim}
\newtheorem{remark}[theorem]{Remark}
\newtheorem{fact}[theorem]{Fact}
\newtheorem{proposition}[theorem]{Proposition}

\newcommand{\ind}{\mathbbm{1}}
\newcommand{\cA}{\mathcal{A}}
\newcommand{\cP}{\mathcal{P}}
\newcommand{\N}{\mathbb{N}}
\newcommand{\R}{\mathbb{R}}
\DeclareMathOperator*{\E}{\ensuremath{\mathbb{E}}}
\DeclareMathOperator*{\Var}{\ensuremath{Var}}
\newcommand{\e}{\ensuremath{\mathrm{e}}}
\newcommand{\diag}{\ensuremath{\mathrm{diag}}}
\newcommand{\spn}{\ensuremath{\mathrm{span}}}
\newcommand{\image}{\ensuremath{\mathrm{Im}}}
\newcommand{\poly}{\mathrm{poly}}
\newcommand{\polylog}{\mathrm{polylog}}
\newcommand{\bl}{\ensuremath{\mathrm{BL}}}
\newcommand{\Paren}[1]{\left(#1\right)}
\newcommand{\Abs}[1]{\left\lvert#1\right\rvert}
\newcommand{\Set}[1]{\left\{#1\right\}}
\newcommand{\norm}[1]{\lVert#1\rVert}
\newcommand{\Norm}[1]{\left\lVert#1\right\rVert}
\newcommand{\Iprod}[1]{\left\langle#1\right\rangle}
\newcommand{\proves}[2]{\sststile{#2}{#1}}

%% file: introduction.tex
\section{Introduction}
\label{sec:introduction}

A central question in the study of structured random graph models is whether their latent structure can be recovered from an observed sample.
For example, in stochastic block models this structure is a latent partition of the vertices~\cite{MR718088-Holland83, MR3827065-Abbe17}; in random geometric graphs it is a latent geometry on the vertices~\cite{MR132566-Gilbert61, MR1986198-Penrose03}.
The goal is to recover the underlying partition or geometry, respectively.

In geometric models, recovery often takes the form of estimating latent vertex positions from adjacency information.
Latent-space and random dot-product models represent these positions by vectors whose pairwise relations determine edge probabilities~\cite{latentspace, randomdotproductgraphsurvey}.
Such vectors can also be interpreted as latent features when their coordinates encode attributes of the vertices.
Recovering the underlying positions, up to symmetries of the model, is the subject of work on latent position estimation~\cite{latentpositionestimation, LS26, FZ26} and geometric reconstruction~\cite{latentdistanceestimation, DDHM22, HJM24}.

In this paper, we study recovery of the latent coordinate structure in random subgraphs of the Boolean hypercube.
The $d$-dimensional hypercube $Q_d$ has vertex set $\{\pm1\}^d$, of size $n=2^d$, and edges between pairs at Hamming distance one.
For $p\in[0,1]$, we consider the distribution on graphs obtained by retaining each edge of $Q_d$ independently with probability $p$. This distribution, denoted by $Q_d(p)$, is known as the \emph{percolated hypercube} and has been studied extensively by the combinatorics community~\cite{MR534014-Erdos79, MR671140-Ajtai82}.

Kapralov, Trevisan, and Wrzos-Kaminska~\cite{kapralov2026recovering} introduced the recovery problem on this distribution, which we now describe. After a graph is sampled from $Q_d(p)$, its vertices are arbitrarily renamed. The original Boolean vectors are now hidden, and serve as latent feature vectors specifying the vertices' positions in the underlying hypercube geometry. In particular, an observed edge reveals that its endpoints differ in exactly one coordinate, but does not identify that coordinate or the vectors themselves. We write $G=(V,E)$ for the observed graph, where $V$ denotes the $n$-element set of renamed vertices, and we call $z:V\to\{\pm1\}^d$ the bijective embedding that assigns to each vertex its latent feature vector in $Q_d$. Kapralov, Trevisan, and Wrzos-Kaminska raised the following question.

\begin{center}
\fbox{\parbox{0.9\linewidth}{\centering\itshape
Given only the observed graph $G=(V,E)$, can one recover the hidden Boolean vectors $z(v)$ associated with its vertices?
}}
\end{center}

We remark that, even when $p=1$, the observed graph cannot distinguish between embeddings related by an automorphism of the hypercube. These automorphisms consist precisely of coordinate permutations and sign flips applied uniformly to all vertices. Recovery is therefore understood up to a single common transformation of this form.

As in the work of~\cite{kapralov2026recovering}, we focus on the setting where $p\geq C\log d/d$ for a sufficiently large absolute constant $C$. The expected degree can then be as small as $\Theta(\log d)=\Theta(\log\log n)$, making this a very sparse recovery problem.
As observed in~\cite{kapralov2026recovering}, it is thus impossible to exactly recover the latent feature vector of every vertex: when $pd =\Theta(\log d)$, a $1/d^{O(1)}$ fraction of vertices become isolated with high probability. This means that even if the vectors $z(v)$ of every other vertex were known, any assignment of the remaining vectors to the isolated vertices would be consistent with the observed graph.
The main result of~\cite{kapralov2026recovering} shows that recovering the Boolean vector of all but a $1/d^{O(1)}$ fraction of vertices is nevertheless information-theoretically possible.

\begin{theorem}[Theorem 1.1 of~\cite{kapralov2026recovering}]
\label{thm:intro-ktw-recovery}
For a sufficiently large absolute constant $C$ and $p\geq C\log d/d$, there exists an $\exp(\tilde{O}(n))$-time algorithm that takes as input an arbitrarily renamed sample $G=(V,E)$ from $Q_d(p)$, and outputs an embedding $\widehat z:V\to\{\pm1\}^d$.
With probability $1-o_d(1)$, the recovered vectors $\widehat z(v)$ agree exactly with the hidden vectors $z(v)$ on all but a $1/d^{O(1)}$ fraction of vertices.
\end{theorem}

Again, agreement is meant up to a common permutation of the coordinates and a common choice of sign for each coordinate. The work of~\cite{kapralov2026recovering} left efficient recovery as an open problem. Specifically, they raised the following natural question.

\begin{center}
\fbox
{
\parbox{0.9\linewidth}
{\centering\itshape
Can one achieve the same recovery guarantees as \Cref{thm:intro-ktw-recovery} with an algorithm that runs in time polynomial in the size of the graph?
}
}
\end{center}

Here, we answer their question in the affirmative.

\begin{theorem}[Main result]
\label{thm:intro-main-recovery}
For a sufficiently large absolute constant $C$ and $p\geq C\log d/d$, there exists a $\poly(n)$-time algorithm that takes as input an arbitrarily renamed sample $G=(V,E)$ from $Q_d(p)$, and outputs an embedding $\widehat z:V\to\{\pm1\}^d$.
With probability $1-o_d(1)$, the recovered vectors $\widehat z(v)$ agree exactly with the hidden vectors $z(v)$ on all but a $1/d^{O(1)}$ fraction of vertices.
\end{theorem}

Our algorithm replaces the exponential-time search of~\cite{kapralov2026recovering} with a sum-of-squares approach combining a degree-$4$ certificate and a polynomial-time rounding procedure.
The main challenge is to show that the relaxation retains enough of the hidden coordinate structure for this rounding procedure to recover the Boolean vectors.
Our analysis relies on two new ingredients that may be of independent interest: a degree-$4$ sum-of-squares proof of the Friedgut-Kalai-Naor theorem~\cite{MR1942632-Friedgut02}, and a spectral concentration theorem for the percolated hypercube.

\subsection{Related work}

\paragraph{Structure recovery from geometric graphs.} Several works have studied random dot-product graphs, where each vertex has a latent feature vector in Euclidean space and the probability of an edge between two vertices is proportional to the inner product of their vectors. Under suitable assumptions, approximate recovery of the latent vectors up to an orthogonal transformation has been shown to be possible~\cite{latentpositionestimation, randomdotproductgraphsurvey}. Related work studies latent vector recovery in generalized random dot-product models under local differential privacy~\cite{pmlr-v258-vishwanath25a}. Other works consider graphs whose edge probabilities depend on the distances between their feature vectors. These works show how to use such observations to estimate pairwise distances~\cite{latentdistanceestimation}, to approximately recover the latent vectors~\cite{DDHM22}, and to reconstruct the underlying manifold geometry~\cite{HJM24}. The results above establish nontrivial guarantees for fixed-dimensional latent geometry in regimes with polylogarithmic or larger expected average degree in the number of vertices $n$, including dense regimes with degree $\Theta(n)$.

Li and Schramm~\cite{LS26} and Fernandez and Zhu~\cite{FZ26} study vector recovery in the Gaussian (mixture) block model, where each vertex is assigned an independent sample from a Gaussian (mixture) distribution and an edge is observed between two vertices if the inner product of their latent vectors exceeds a threshold. In their setting, the ambient dimension is allowed to grow with the number of vertices, as opposed to being fixed as in the works above. Their spectral algorithm recovers the feature vector of the vertices up to rotation, and their result gives vanishing average error when the expected degree exceeds the dimension by a polylogarithmic factor.

Our work shares some similarities with the results above, but also has important differences. In particular, our model works in ambient dimension $d=\log_2 n$ that grows with the number of vertices, while the expected average degree can be as small as $\Theta(\log d)=\Theta(\log\log n)$, well below the ambient dimension. Our algorithm recovers the latent feature vectors exactly, as opposed to approximately, on almost all vertices, up to a common hypercube symmetry. Our model contains one feature vector for every point in our vector space, i.e., the Boolean hypercube, and the randomness comes from the independent edge sampling. These differences in geometry, observations, and recovery criteria make the guarantees incomparable. Both works~\cite{LS26,FZ26} also study clustering, where mixture memberships partition the vertices into two groups.

As we shall discuss, the recovery problem that we study benefits from a community interpretation~\cite{kapralov2026recovering}: each coordinate in our model gives a bipartition, or \emph{coordinate cut}. The striking difference with the Gaussian mixture block model is that  communities arising in our model overlap heavily: almost every vertex belongs to $\Theta(d)$ of the sets $\{v\in V:z(v)_i=1\}$.

\paragraph{Percolated hypercube.}
Work on random subgraphs of the hypercube covers connectivity and component structure~\cite{MR534014-Erdos79, MR671140-Ajtai82}, the percolation phase transition and largest components~\cite{MR1139488-Bollobas92, MR2260845-Borgs06, MR3612867-Hofstad17, MR4060356-Hulshof20}, and expansion, isoperimetry, diameter, and mixing~\cite{MR4515692-Erde23, MR4777865-Diskin24, anastos2025diameter}.
Kapralov, Trevisan, and Wrzos-Kaminska~\cite{kapralov2026recovering} introduced recovery of the hidden hypercube coordinate structure from such graphs.

\paragraph{Sum-of-squares Boolean function analysis.}
Our degree-$4$ sum-of-squares proof of the Friedgut-Kalai-Naor theorem contributes to a line of sum-of-squares proofs in Fourier analysis of Boolean functions, including hypercontractivity~\cite{MR2961513-Barak12}, the Kahn-Kalai-Linial (KKL) theorem~\cite{OZ13}, Majority Is Stablest~\cite{DMN16}, and reverse hypercontractivity~\cite{MR3533303-Kauers16}.

\subsection{Organization}

In \Cref{sec:technical-overview}, we give a technical overview of our results.
\Cref{sec:preliminaries} collects notation and background.
We prove the degree-$4$ sum-of-squares version of the Friedgut-Kalai-Naor theorem in \Cref{sec:sos-fkn} and the spectral concentration theorem for the percolated hypercube in \Cref{sec:spectral-concentration}.
In \Cref{sec:coordinate-recovery}, we combine these ingredients with the rounding theorem to prove our main recovery theorem.
Finally, in \Cref{sec:degree-2-lower-bound}, we prove the lower bound for the degree-$2$ relaxation.

%% file: technical-overview.tex
\section{Technical overview}
\label{sec:technical-overview}

\input{overview-coordinate-recovery.tex}

\input{overview-quartic-proxy.tex}

\input{overview-degree-2-lower-bound.tex}

%% file: overview-coordinate-recovery.tex
\subsection{Recovery via Coordinate Cuts}
\label{sec:overview-coordinate-recovery}
We now reduce recovery of the latent Boolean vectors to simultaneous recovery of the \(d\) coordinate cuts.
The coordinate structure of $Q_d$ is naturally encoded by the dictator functions $\chi_1,\dots,\chi_d$, or equivalently by the coordinate cuts, which we define as follows: for each $i \in [d]$, the $i$-th coordinate cut is $(S_i,\{\pm1\}^d\setminus S_i)$, where $S_i=\{x \in  \{\pm 1\}^d:x_i=1\}$. We will interchangeably identify cuts with functions $f:\{\pm1\}^d\to\{\pm1\}$, with the two signs representing the two sides of the cut; the $i$-th coordinate cut is then represented by the $i$-th dictator $\chi_i(x)=x_i$.

Let $G=(V,E)$ be an arbitrarily renamed sample from $Q_d(p)$, and let $z : V \to \{\pm 1\}^d$ be the embedding of the vertices to their corresponding vectors. Under $z$, a cut $f$ in $Q_d$ naturally corresponds to a cut $\hat{f}$ in $G$ defined by $ \hat{f}(v) = f(z(v))$.
In particular, $\hat{\chi}_i(v)=\chi_i(z(v))=z_i(v)$ for each $i\in[d]$ and $v\in V$.
The key point is that the $d$ coordinate cuts in $G$, namely $\hat{\chi}_1, \dots, \hat{\chi}_d$, determine $z$ up to unavoidable symmetries (i.e. permuting the coordinates and flipping their signs, as discussed in \cref{sec:introduction}).
Thus, recovering the coordinate cuts from observing $G$ is equivalent to recovering the hidden Boolean vector of the vertices.
To measure accuracy, we define the agreement between two cuts $f,g$ as
\[
\max\left\{\Pr_{{x}}[f({x})=g({x})],\Pr_{{x}}[f({x})=-g({x})]\right\},
\]
where ${x}$ is uniform on $\{\pm1\}^d$, or equivalently we could have the same definition for $\hat{f},\hat{g}$ by averaging over $v$ uniform on $V$. We have the following elementary observation, showing that recovering the coordinate cuts with high agreement amounts to recovering the hidden Boolean vector of most vertices.

\begin{lemma}[Latent vector recovery from coordinate cuts]
\label[lemma]{lem:overview-latent-vector-recovery}
    Let $G=(V,E)$ be an arbitrarily renamed sample from $Q_d(p)$, and let $z : V \to \{\pm 1\}^d$ be the embedding of the vertices to their corresponding vectors. If $f_1,\dots,f_d :V\to \{\pm 1\}$ are cuts in $G$ such that there is a permutation $\pi$ of $[d]$ for which each $f_{\pi(i)}$ has agreement at least $1-\varepsilon$ with $\hat{\chi}_i$, then there is a signing $s \in \{\pm 1\}^d$ such that the embedding $\widehat{z}(v) = (s_1 f_{\pi(1)}(v), \dots, s_d f_{\pi(d)}(v))$ exactly equals $z(v)$ on all but at most a $d\varepsilon$ fraction of vertices.
\end{lemma}

\begin{proof}
For each $i\in[d]$, choose $s_i\in\{\pm1\}$ to be a sign attaining
the maximum in the definition of agreement between $f_{\pi(i)}$ and $\hat{\chi}_i$.
Then $s_i f_{\pi(i)}(v)=z_i(v)$ on at least a $1-\varepsilon$
fraction of vertices.
By a union bound, all $d$ equalities hold simultaneously
on at least a $1-d\varepsilon$ fraction of vertices.
\end{proof}

We remark that this is how~\cite{kapralov2026recovering} also approached the problem. Thus, in the remainder of the paper we can focus on recovering the coordinate cuts with high agreement $1-\varepsilon$ for $\varepsilon \le 1/d^c$ for some $c\ge 2$, giving recovery for all but a $1/d^{c-1}=1/\poly(d)$ fraction of vertices (as stated in \Cref{thm:intro-main-recovery}). Note that when $\varepsilon \ge 1/d$, no nontrivial conclusion can be drawn from the lemma above.

As anticipated, our algorithm is based on the sum-of-squares hierarchy of semidefinite programs.
We start by formulating a set of polynomial constraints over variables $\{f(v)\}_{v \in V}$, imposing that $f$ represents a balanced cut with small value in $G=(V,E)$:\footnote{Our formal algorithm does not require knowledge of $p$, but instead estimates $p$ based on the average degree of the observed graph. For the sake of clarity, we pretend $p$ is known throughout \Cref{sec:overview-coordinate-recovery}.}
\begin{equation*}
\cP = \begin{cases}
f(v)^2 = 1\,, \forall v \in V\,,\\
\sum_{v \in V} f(v) = 0\,,\\
\sum_{\{u, v\} \in E} \frac{1}{4} (f(u) - f(v))^2 \leq \left(1+\frac{1}{d^{O(1)}}\right) p 2^{d-1}.
\end{cases}
\end{equation*}
The first constraint imposes that $f(v) \in \{\pm 1\}$ for all $v \in V$, the second constraint imposes that $f$ is a balanced cut, and the third constraint imposes that $f$ has small cut value.
One can verify that this is indeed the natural constraint system to consider. Let $\hat \chi_1, \ldots, \hat \chi_d : V\to \Set{\pm 1}$ be the coordinate cuts of $G$ according to the latent embedding $z$. These are the ideal solutions that we would like to recover, which are unknown to the algorithm. It is easy to check that with high probability each $\hat \chi_i$ satisfies~$\cP$.

\begin{lemma}[Coordinate cut value in $G$ -- see \Cref{lem:coordinate-cut-concentration}]
    \label{lem:overview-coordinate-cut-concentration}
    With probability at least $1-o_d(1)$, for every $i \in [d]$ we have that $\hat{\chi}_i$ satisfies $\cP$, and in particular
    \begin{equation*}
         \sum_{\{u, v\} \in E(G)} \frac{1}{4} (\hat{\chi}_i(u) - \hat{\chi}_i(v))^2 = \left(1 \pm \frac{1}{d^{O(1)}}\right) p 2^{d-1}\,.
    \end{equation*}
\end{lemma}
\begin{proof}[Proof sketch]
We recall that coordinate cuts correspond to dictator functions, which are the eigenvectors of the Laplacian of $Q_d$ with the smallest positive eigenvalue. In particular, they have eigenvalue $2$, i.e.,
\begin{equation*}
    \langle \chi_i, L_{Q_d} \chi_i\rangle =2 \, ,
\end{equation*}
where $L_{Q_d}$ is the unnormalized Laplacian of $Q_d$.\footnote{For functions $g,h:\Omega \to \R$, we write $\langle g,h\rangle$ for their inner product with respect to the uniform measure on $\Omega$, i.e., $\langle g,h\rangle = \E_{x\sim \mathrm{Unif}(\Omega)} g(x)h(x)$, and define the associated norm by $\|g\|^2 = \langle g,g\rangle$. In our case, $\Omega$ is either $\{\pm1\}^d$ or its arbitrary renaming $V$.}
On the other hand,
\begin{equation*}
    \langle \chi_i, L_{Q_d} \chi_i\rangle =  \frac{4}{2^d}\sum_{\{x,y\} \in E(Q_d)} \frac{1}{4}({\chi}_i(x) - {\chi}_i(y))^2 \, .
\end{equation*}
By rearranging, we get that each ${\chi}_i$ has cut value $2^{d-1}$ in $Q_d$. Therefore, the expected cut value of each $\hat{\chi}_i$ in $G$ is
\begin{equation*}
    \E_G\left[
  \frac14\sum_{\{u,v\}\in E(G)}
  (\hat{\chi}_i(u)-\hat{\chi}_i(v))^2
\right]
= p2^{d-1}\,.
\end{equation*}
Because the edges are added to $G$ independently, a Chernoff bound gives the claim.
\end{proof}

Then, our goal is to prove, via a \emph{degree-$4$ sum-of-squares proof}, a quartic inequality that acts as a proxy to the statement that any $f$ that satisfies $\cP$ must be close to one of the ideal solutions $\hat \chi_1, \ldots, \hat \chi_d$.

\begin{lemma}[Quartic proxy -- see \Cref{claim:sos-quartic-proxy-certificate}]
\label{lem:overview-quartic-proxy-observed}
With probability $1-o_d(1)$, there is a degree-$4$ sum-of-squares proof that $\cP$ implies
\begin{equation*}
    \sum_{i\in[d]}\langle f,\hat{\chi}_i\rangle^4 \geq 1-\frac{1}{d^{O(1)}}\,.
\end{equation*}
\end{lemma}

Ordinarily, this inequality forces $f$ to be close to one of the $\hat \chi_i$'s up to a sign flip:
\[
1-\frac{1}{d^{O(1)}}
\leq \sum_{i\in[d]}\langle f,\hat\chi_i\rangle^4
\leq
  \left(\sum_{i\in[d]}\langle f,\hat\chi_i\rangle^2\right)
  \max_{j\in[d]}\langle f,\hat\chi_j\rangle^2
\leq \max_{i\in[d]}\langle f,\hat\chi_i\rangle^2 \, ,
\]
where the last inequality follows because $\sum_i \langle f,\hat\chi_i\rangle^2 \le 1$ by Parseval.
Thus, for some $i \in [d]$, we have $|\langle f,\hat\chi_i\rangle| \geq 1-1/d^{O(1)}$, which implies agreement $1-1/d^{O(1)}$ between $f$ and $\hat\chi_i$.
However, this reasoning is not captured by the sum-of-squares proof system, so at this point we invoke instead a rounding theorem of~\cite{MR3631006-Ma16}.

\begin{theorem}[Rounding from quartic proxy -- see \Cref{thm:sos-rounding}]
\label{thm:overview-sos-rounding}
    Let $a_1, \dots, a_d$ be orthonormal vectors in $\R^n$, and let $\cA$ be a system of polynomial constraints in variables $w_1, \dots, w_n$ that implies \smash{$\sum_{i=1}^d \langle w, a_i\rangle^4 \ge 1-\varepsilon$} via a degree-$r$ sum-of-squares proof. If each $a_i$ satisfies the system $\cA$, there is an $(n+d)^{O(r)}$-time algorithm that takes as input only $\cA$ and outputs with high probability $d$ vectors $a_1',\dots, a_d'$ such that there is a permutation $\pi$ of $[d]$ satisfying $\min\{\|a_i - a_{\pi(i)}'\|^2,\|a_i + a_{\pi(i)}'\|^2\} \le O(\varepsilon)$ for all $i\in[d]$.
\end{theorem}

Equipped with this rounding algorithm, if we can establish \Cref{lem:overview-quartic-proxy-observed}, then we can easily conclude~\Cref{thm:intro-main-recovery}.

\begin{proof}[Proof sketch of \Cref{thm:intro-main-recovery}]
We recall that the coordinate cuts $\hat{\chi}_i$ are orthonormal vectors in $\R^V$ with respect to the uniform measure on $V$, and they are feasible solutions for $\cP$ by \Cref{lem:overview-coordinate-cut-concentration}. Moreover, \Cref{lem:overview-quartic-proxy-observed} gives a degree-$4$ sum-of-squares derivation of $\sum_{i\in[d]}\langle f,\hat{\chi}_i\rangle^4 \geq 1-{1}/{d^{O(1)}}$ from $\cP$. We can then apply the rounding algorithm in \Cref{thm:overview-sos-rounding}: we have a polynomial-time algorithm that takes as input $\cP$ and outputs with high probability $d$ real-valued vectors $a'_1,\dots, a'_d$ such that $a'_i$ is $1/d^{O(1)}$ close to $\hat{\chi}_i$. A simple threshold rounding allows us to turn the vectors $a'_1,\dots, a'_d$ into cuts $f_1,\dots,f_d$ so that each $f_i$ has agreement at least $1-1/d^{O(1)}$ with $\hat{\chi}_i$. Since $\cP$ can be written from $G$ alone, without any knowledge of the latent $\{\pm 1\}^d$ vectors of the vertices $V$, we are done.
\end{proof}

%% file: overview-quartic-proxy.tex
\subsection{From balanced low-value cut to the quartic proxy}
\label{sec:overview-quartic-proxy}

The main challenge left is proving \Cref{lem:overview-quartic-proxy-observed}, which we now turn to. We note that the third constraint of $\cP$ can be rewritten as $ \langle f, L_G f\rangle \leq (1+1/d^{O(1)})2p$, where $L_G$ is the unnormalized Laplacian of $G$. Additionally, observe that the sum-of-squares proof in \Cref{lem:overview-quartic-proxy-observed} appears only in the analysis, and we can therefore work with $G$ expressed in its true hypercube coordinates, which we call $G^*$. By virtue of these two observations, we define the constraint system

\begin{equation}
\label{eq:overview-cut-constraints}
\cP^* = \begin{cases}
f(x)^2 = 1\,, \forall x \in \{\pm 1\}^d\,,\\
\sum_{x \in \{\pm 1\}^d} f(x) = 0\,,\\
\langle f, L_{G^*} f\rangle \leq \left(1+\frac{1}{d^{O(1)}}\right)2p \,,
\end{cases}
\end{equation}

and proving \Cref{lem:overview-quartic-proxy-observed} amounts to proving the following claim.

\begin{lemma}[Quartic proxy -- see \Cref{claim:sos-quartic-proxy-certificate}]
\label{lem:overview-quartic-proxy}
With probability $1-o_d(1)$, there is a degree-$4$ sum-of-squares proof that $\cP^*$ implies
\begin{equation}
\label{eq:overview-quartic-proxy}
    \sum_{i\in[d]}\langle f,{\chi}_i\rangle^4 \geq 1-\frac{1}{d^{O(1)}}\,.
\end{equation}
\end{lemma}

The above cannot be obtained by simply rewriting the argument of~\cite{kapralov2026recovering}. This is because
a low-degree sum-of-squares proof is an efficiently checkable certificate: it certifies that every balanced cut of near-minimum value is close to a coordinate cut.
Thus, to obtain such a proof, we need to replace the non-algorithmic parts of the earlier analysis by statements that can be verified in polynomial time. One concrete way to view our result is the following, which is a corollary of our proof.

\begin{proposition}[Efficient certificate]
\label[proposition]{prop:overview-efficient-certificate}
Given a subgraph $G$ of the $d$-dimensional hypercube (with vertices expressed in their true hypercube coordinates) and a parameter $p$, there is a $\poly(n)$-time algorithm that outputs either \textsc{yes} or \textsc{unsure} so that:
\begin{itemize}
    \item if the algorithm outputs \textsc{yes}, then every balanced cut with value $\leq (1+1/d^{O(1)})p2^{d-1}$ in $G$ has agreement at least $1-1/d^{O(1)}$ with a coordinate cut of the hypercube;
    \item if $G$ is sampled from $Q_d(p)$, then the algorithm outputs \textsc{yes} with probability $1-o_d(1)$ over the input.
\end{itemize}
\end{proposition}

The first item is a soundness guarantee for every input graph: the algorithm never certifies a graph unless the desired structural conclusion is true.
The second item is a completeness guarantee for the random model: with high probability, samples from $Q_d(p)$ admit such a certificate. The structural theorem of~\cite{kapralov2026recovering} also implicitly gives such a certification algorithm, but it runs in $\exp(\tilde{O}(n))$ time: the algorithm has to check every cut in $G$ in order to maintain the soundness guarantee.

\subsubsection{Proof outline of the quartic proxy lemma}
\label{sec:overview-quartic-proof-outline}
We start by rewriting the proof of \cite{kapralov2026recovering} in a way that mimics ours. In particular, we show that $\cP^*$ implies \Cref{eq:overview-quartic-proxy}, albeit not via a sum-of-squares proof. The remainder of this subsection can thus be thought of as an outline of the proof route that we take to eventually prove \Cref{lem:overview-quartic-proxy}. This will allow us to highlight the sources of ``inefficiency'' in the proof of \cite{kapralov2026recovering}, which we discuss and address below.

The crucial idea underlying the work of \cite{kapralov2026recovering} is to apply a theorem of Friedgut, Kalai, and Naor: a Boolean function with most of its Fourier mass on the degree-$1$ characters is very close to one of the dictators.

\begin{theorem}[FKN -- see \cite{MR1942632-Friedgut02}]
Let $\delta \ge 0$ and let $f : \{\pm 1\}^d \to \{\pm 1\}$ be a Boolean function such that \smash{$\sum_{i=1}^d \langle f, \chi_i\rangle^2 \ge 1-\delta$}. Then, there exists some $i \in [d]$ such that $f$ has agreement at least $1-O(\delta)$ with either $\chi_i$ or $-\chi_i$, i.e. $\langle f, \chi_i\rangle^2 \ge 1-O(\delta)$.
\end{theorem}

By the FKN theorem, it suffices to show that every $f$ satisfying $\cP^*$ has almost all its Fourier mass on the degree-$1$ characters. For a balanced Boolean function, this follows from having cut value close to the minimum in the full hypercube, as the following lemma shows.

\begin{lemma}[Near-minimum cuts are almost degree-$1$ -- see \Cref{lem:cut-energy-fourier-concentration}]
\label{lem:overview-cut-energy-fourier-concentration}
    Let $\eta \ge 0$, and let $f:\{\pm 1\}^d \to \{\pm 1\}$ be a balanced cut. If $\langle f, L_{Q_d} f\rangle \le 2(1+\eta)$, then \smash{$\sum_{i=1}^d \langle f, \chi_i\rangle^2 \ge 1-2\eta$}.
\end{lemma}
\begin{proof}[Proof sketch]
    Writing the eigendecomposition of $L_{Q_d}$, we have $\langle f, L_{Q_d} f\rangle = 2\sum_{S \subseteq [d]} |S| \langle f, \chi_S\rangle^2 \geq 2 + \sum_{|S| \neq 1} \langle f, \chi_S\rangle^2$. On the other hand, $\langle f, L_{Q_d} f\rangle \le (1+\eta)2$. Combining these facts gives $\sum_{|S|\neq 1} \langle f, \chi_S\rangle^2 \le 2 \eta$ and hence the claim.
\end{proof}

It therefore remains to show that every $f$ satisfying $\cP^*$ also satisfies
\(
\langle f,L_{Q_d}f\rangle \leq 2\left(1+\frac{1}{d^{O(1)}}\right).
\)
We begin with a general comparison between cut values in the sampled graph and the full hypercube. For each $i\in[d]$, we denote by \smash{$L_{Q_d}^{(i)}$} and \smash{$L_{G^*}^{(i)}$} the unnormalized Laplacians of the graphs obtained from $Q_d$ and $G^*$, respectively, by removing the edges across the coordinate cut $\chi_i$.

\begin{lemma}[Cut value concentration -- see proof of \Cref{claim:sos-cut-value-concentration}]
\label{lem:overview-cut-value-concentration}
With probability $1-o_d(1)$, for some small constant $\epsilon>0$ and simultaneously for every Boolean function $g:\{\pm1\}^d\to\{\pm1\}$, we have
\[
\begin{aligned}
\left|\langle g,L_{Q_d}g\rangle
  -\frac1p\langle g,L_{G^*}g\rangle\right| \leq \epsilon\langle g,L_{Q_d}g\rangle
  +\frac{1}{d^{O(1)}}.
\end{aligned}
\]
The same bound holds simultaneously for every $i\in[d]$ with $L_{Q_d}, L_{G^*}$ replaced by \smash{$L_{Q_d}^{(i)}, L_{G^*}^{(i)}$}, respectively.
\end{lemma}

Applied directly to a function $f$ satisfying $\cP^*$, this estimate gives only
\[
\langle f,L_{Q_d}f\rangle \leq 2+O(\epsilon)+\frac{1}{d^{O(1)}}.
\]
Since $\epsilon$ is a constant, this does not give the inverse-polynomial error we need. To obtain a sharper bound, we compare the cut value of $f$ with that of a coordinate cut. The following identity expresses this difference as a quadratic form in $f\circ\chi_i$.

\begin{lemma}[Cut value difference -- see \Cref{lem:cut-symmetric-difference-identity}]
    \label{lem:overview-cut-symmetric-difference}
    Let $f$ be a function that satisfies $\cP^*$, and let $i \in [d]$. Then, denoting by $f \circ \chi_i$ the entrywise product between $f$ and $\chi_i$, we have
    \begin{equation*}
        \langle f, L_{H} f\rangle - \langle \chi_i, L_{H} \chi_i\rangle = \left\langle f \circ \chi_i, \left(2 L_{H}^{(i)} - L_{H}\right) (f \circ \chi_i)\right\rangle
    \end{equation*}
    for both $H=Q_d$ and $H=G^*$.
\end{lemma}

We apply \Cref{lem:overview-cut-value-concentration} to the Boolean function $g=f\circ\chi_i$, using the bounds for both $L_{Q_d}$ and \smash{$L_{Q_d}^{(i)}$}. Together with \Cref{lem:overview-cut-symmetric-difference}, this compares the excess cut values in $Q_d$ and $G^*$ with error at most $O(\epsilon)\langle f\circ\chi_i,L_{Q_d}(f\circ\chi_i)\rangle+1/d^{O(1)}$. With some additional work, we obtain the following lemma.

\begin{lemma}[Cut value in $Q_d$ vs in $G^*$ -- see \Cref{lem:cut-comparison-coordinate-correlation}]
\label{lem:overview-refined-cut-comparison}
With probability $1-o_d(1)$, for some small constant $\epsilon>0$, simultaneously for every function $f$ that satisfies $\cP^*$, and every $i\in[d]$, we have
\[
(1-\epsilon)\langle f,L_{Q_d}f\rangle \leq \frac1p\langle f,L_{G^*}f\rangle-2\epsilon+4\epsilon(1-\langle f,\chi_i\rangle^2)+\frac{1}{d^{O(1)}}.
\]
\end{lemma}

We now combine this refined comparison with the FKN theorem to obtain the desired cut-value bound in $Q_d$, and hence the quartic proxy in \Cref{eq:overview-quartic-proxy}.

\begin{lemma}[Inefficient quartic proxy]\label{lem:overview-inefficient-quartic-proxy}
With probability $1-o_d(1)$, every function $f$ that satisfies $\cP^*$ must also satisfy \smash{$\sum_{i=1}^d \langle f, \chi_i\rangle^4 \ge 1-1/d^{O(1)}$}.
\end{lemma}
\begin{proof}[Proof sketch]
Let $\chi_i$ be the coordinate cut with the largest agreement with $f$. By \Cref{lem:overview-refined-cut-comparison} and the third constraint in $\cP^*$ (see \Cref{eq:overview-cut-constraints}), we have
\begin{align*}
(1-\epsilon)\langle f, L_{Q_d} f\rangle & \leq \frac{1}{p}\langle f, L_{G^*} f\rangle -2\epsilon + 4\epsilon (1-\langle f,\chi_i\rangle^2) + \frac{1}{d^{O(1)}} \\
 &\leq 2\left(1+\frac{1}{d^{O(1)}}\right)  -2\epsilon + 4\epsilon (1-\langle f,\chi_i\rangle^2) \,.
\end{align*}
We can apply \Cref{lem:overview-cut-energy-fourier-concentration} with $\eta=\frac{1}{2}(\langle f, L_{Q_d} f\rangle-2)$ to get the bound $\sum_{j=1}^d\langle f,\chi_j\rangle^2 \ge 1-2\eta$. Then, by the FKN theorem and our choice of $i$, we have $\langle f, \chi_i\rangle^2 \ge 1- O(\eta) = 1-O(\langle f, L_{Q_d} f\rangle-2)$. Hence, we continue from above and write
\begin{align*}
(1-\epsilon)\langle f, L_{Q_d} f\rangle 
\leq 2\left(1+\frac{1}{d^{O(1)}}\right)  -2\epsilon + \epsilon \cdot O((\langle f, L_{Q_d} f\rangle-2))\,.
\end{align*}
Subtracting $2(1-\epsilon)$ from both sides and choosing $\epsilon$ sufficiently small, we can absorb the term proportional to $\epsilon(\langle f,L_{Q_d}f\rangle-2)$ into the left-hand side. This gives the desired cut-value bound:
\begin{align*}
\langle f, L_{Q_d} f\rangle  \leq 2\left(1+\frac{1}{d^{O(1)}}\right) \,.
\end{align*}
Finally, a simple application of \Cref{lem:overview-cut-energy-fourier-concentration} with $\eta=1/d^{O(1)}$ gives that $\sum_{j =1}^d\langle f, \chi_j \rangle^2 \ge 1-1/d^{O(1)}$. Then, one final application of the FKN theorem gives $\langle f, \chi_i \rangle^2 \ge 1-1/d^{O(1)} $, which trivially implies $\sum_{j=1}^d \langle f, \chi_j\rangle^4 \ge 1-1/d^{O(1)}$.
\end{proof}

\paragraph{Argument SoS-ization and sources of inefficiency.} In the way we have written the above arguments, one can see that most of the steps are easily written as degree-$2$ or degree-$4$ sum-of-squares proofs. Even the basic Fourier-analytic statements that we need, such as Parseval's identity and the Fourier expansion of cut values, are captured by the sum-of-squares proof system quite straightforwardly. There are, however, two crucial steps for which it is not a priori clear that a low-degree sum-of-squares equivalent can be found: the proof of the cut value approximation in \Cref{lem:overview-cut-value-concentration}, and the application of the FKN theorem in the proof of \Cref{lem:overview-inefficient-quartic-proxy}. We discuss these issues and how to address them in \Cref{sec:overview-spectral-concentration} and \Cref{sec:overview-sos-fkn}, respectively.

\input{overview-spectral-concentration.tex}

\input{overview-sos-fkn.tex}

\input{overview-sos-certificate.tex}

%% file: overview-spectral-concentration.tex
\subsubsection{Cut value concentration}
\label{sec:overview-spectral-concentration}
As written, \Cref{lem:overview-cut-value-concentration} resembles the form of the intermediate lemmas used by~\cite{kapralov2026recovering}. Naively, a certifiable version of this lemma would require checking exponentially many cuts. We instead certify it using a \emph{spectral} concentration theorem for the percolated hypercube.

\paragraph{Obstruction: atypical-degree vertices. }We note that, in the sparse regime $p = \Theta(\log d/d)$, the sampled graph $G$ is not a spectral sparsifier of the full cube. In fact, it is not even a cut sparsifier. For example, for a fixed vertex $v$, the probability that $v$ is isolated in $G$ is
\[
\Pr[\deg_G(v)=0] = (1-p)^d \approx \frac{1}{d^{O(1)}},
\]
and the probability that $v$ has, say, $\sqrt{d}$ incident edges in $G$ is
\[
\Pr[\deg_G(v) \ge \sqrt{d}] \ge p^{\sqrt{d}} \approx \frac{1}{2^{o(d)}},
\]
so there are vertices of each type with probability $1-o_d(1)$. Both types of atypical vertices create local errors: these singleton cuts deviate by much more than a $(1+\epsilon)$-factor from the expected value $pd$. This challenge appears also in the approach of \cite{kapralov2026recovering}. On a high level, they address it by taking advantage of the structure of the cuts for which they need to show concentration: they are cuts $f \circ \chi_i$ where $f$ is a balanced function of small cut value in the sample $G$. In particular, such cuts are not significantly affected by the local errors introduced by the atypical-degree vertices. However, to obtain an \emph{efficient} certificate, we cannot leverage the combinatorial properties of the cuts for which we need concentration.

\paragraph{Workaround: allowing additive error.}  To sidestep the obstruction, we allow an additive error supported on a small set of atypical-degree vertices. We can then show that the Laplacian of the sampled graph is spectrally close to its expectation $pL_{Q_d}$, up to such additive error. This yields a polynomial-time checkable condition that certifies the needed cut estimates.

\begin{lemma}[Spectral concentration -- see \Cref{lem:spectral-concentration}]
\label{lem:overview-spectral-concentration}
For some small constant $\epsilon>0$, with probability $1-o_d(1)$ we have
\begin{equation*}
(1-\epsilon) p L_{Q_d} -  Z \preceq L_{G^*} \preceq  (1+\epsilon) p L_{Q_d}  + Z\, ,
\end{equation*}
where $Z =d^{2} P_{> 0} \Pi_R P_{> 0}$ is the additive error that we incur, $P_{> 0} = I - \frac{1}{2^d} \chi_\emptyset^{\otimes 2}$ is the projection orthogonal to the bottom eigenvector, and $\Pi_R$ is the matrix with $1$ on the diagonal for every vertex whose subsampled degree is $\Omega(\epsilon pd)$-far from its expected degree $p d$.
\end{lemma}

The proof of this fact is technical and we defer it to \Cref{sec:spectral-concentration}. We present here a very high-level overview.

\begin{proof}[Proof sketch]
By an application of Bernstein's inequality for symmetric matrices, one shows that $L_{G^*}$ is $(1\pm \epsilon)$-spectrally close to $p L_{Q_d}$ on the subspace of functions of Fourier degree at most $d/100$ (i.e. functions in the span of $\{\chi_S\}_{|S| \le d/100}$). Hence, we only need to focus on the subspace of functions of higher degree.
To do so, we show that the subgraph induced by well-behaved vertices is spectrally close to the corresponding subgraph in $Q_d$ up to additive error $\approx \log d$. More precisely, let $R$ be the set of vertices whose degree in $G^*$ deviates by more than $\Omega(\epsilon pd)$ from $pd$, and let $U = \{\pm 1\}^d \setminus R$. Then, one can show that the Laplacian of $G^*[U]$ is $\pm (\log d) I$ close to the Laplacian of $Q_d[U]$ rescaled by $p$ (where the induced subgraphs retain the original degrees). Now we observe that functions of degree higher than $d/100$ have quadratic form at least $\Omega(pd)$ with $p L_{Q_d}$. Hence, we can tolerate the $\log d$ error by choosing $C$ large enough in the condition $p \ge C \log d /d$. We note that the idea of ``removing'' vertices deviating from the mean is present in some spectral algorithms for stochastic block models in the constant-degree regime (e.g.,~\cite{sparse-sbm}), a setting that is morally similar to ours.
What is left is the contribution of the high-degree functions on the subgraph induced on $R$. We incur this as the additive error $Z$. The factor of $d^2$ can be thought of as coming from the maximum eigenvalue of the Laplacian matrices involved, which is $2d$ in our case. However, more nuances are involved in the formal argument, such as handling cross terms of mixed high- and low-degree functions.
\end{proof}

The spectral approximation in \Cref{lem:overview-spectral-concentration} yields
the following degree-$2$ sum-of-squares version of
\Cref{lem:overview-cut-value-concentration}.

\begin{lemma}[Efficient cut value concentration -- see the proof of \Cref{claim:sos-cut-value-concentration}]
\label{lem:overview-sos-cut-value-concentration}
With probability $1-o_d(1)$, there are degree-$2$ sum-of-squares proofs that the Booleanity constraints $\{g(x)^2=1, \,\forall x\in\{\pm1\}^d\}$ imply
\[
\begin{aligned}
\pm \left(\langle g,L_{Q_d}g\rangle
    -\frac{1}{p}\langle g,L_{G^*}g\rangle\right) \leq \epsilon\langle g,L_{Q_d}g\rangle
    +\frac{1}{d^{O(1)}}\,.
\end{aligned}
\]
Both choices of sign are included.
The same conclusion holds simultaneously for all $i\in[d]$ when $L_{Q_d},L_{G^*}$ are replaced by \smash{$L_{Q_d}^{(i)},L_{G^*}^{(i)}$}, respectively.
\end{lemma}
\begin{proof}[Proof sketch]
We simply apply \Cref{lem:overview-spectral-concentration}, observing that the quadratic form of $g$ with $Z$ is at most $1/\poly(d)$ with high probability. To get the same conclusion for $L_{Q_d}^{(i)}$, we note that it consists of two disjoint $(d-1)$-dimensional hypercubes. Hence, we separately apply \Cref{lem:overview-spectral-concentration} to these two copies of $Q_{d-1}$.
\end{proof}

%% file: overview-sos-fkn.tex
\subsubsection{FKN theorem}
\label{sec:overview-sos-fkn}
The FKN theorem plays a crucial role both in
\cite{kapralov2026recovering} and in the inefficient argument of
\Cref{lem:overview-inefficient-quartic-proxy}. Under the hypothesis that a Boolean-valued function
$f$ has almost all of its Fourier mass on the degree-$1$ characters,
the theorem guarantees that $f$ is close to one of the dictators
$\pm\chi_{i^*}$. This conclusion singles out an unknown coordinate
$i^*$ and is therefore not naturally expressible as a single
low-degree polynomial inequality in the sum-of-squares proof system.

Instead, we use an intermediate step in the proof of the FKN
theorem: large degree-$1$ Fourier mass forces the
quartic quantity $\sum_{i\in[d]}\langle f,\chi_i\rangle^4$
to be large; see, for example,~\cite{MR3443800-ODonnell14}. This
quartic conclusion serves as a polynomial proxy for the existence of a
single highly correlated dictator. We prove that this proxy admits a
degree-$4$ sum-of-squares derivation from the following system of
polynomial constraints:
\[\cA= \begin{cases}
f(x)^2 = 1\,, \forall x \in \{\pm 1\}^d\,, \\
\sum_{i \in [d]} \langle f,\chi_i\rangle^2\ge 1-\delta\,,\\
\delta \ge 0\, .
\end{cases}\]

\begin{theorem}[Sum-of-squares FKN theorem -- see \Cref{thm:sos-fkn}]
\label{thm:overview-sos-fkn}
There is a degree-$4$ sum-of-squares proof that $\cA$ implies $\sum_{i \in [d]} \langle f,\chi_i\rangle^4\ge 1-O(\delta)$.
\end{theorem}

Our argument follows the outline of the proof of the FKN theorem in O'Donnell's book~\cite{MR3443800-ODonnell14}, but we crucially alter it in order not to use the Paley-Zygmund anti-concentration inequality, which is not amenable to sum-of-squares proofs. Instead, we manipulate the variance of the linear part of $f$ through a series of degree-$2$ or degree-$4$ polynomial inequalities. Another difference from the proof of O'Donnell is that we do not need to employ the Bonami-Beckner lemma. While we also had a version that uses the sum-of-squares proof of Bonami-Beckner~\cite{MR2961513-Barak12}, it leads to a proof with degree higher than $4$, so we do not include it here. The details of our proof are deferred to \Cref{sec:sos-fkn}.

%% file: overview-sos-certificate.tex
\subsubsection{Patching the proof outline with sum-of-squares proofs}
\label{sec:overview-sos-certificate}

We can now patch the proof outline presented in \Cref{sec:overview-quartic-proof-outline} by replacing the non-sum-of-squares steps with their efficiently certifiable counterparts that we presented in \Cref{sec:overview-spectral-concentration} and \Cref{sec:overview-sos-fkn}. Specifically, to complete the proof of \Cref{lem:overview-quartic-proxy}, we combine
the spectral certificate with our sum-of-squares FKN theorem.
We multiply each coordinate-specific cut comparison by
$\langle f,\chi_i\rangle^2$ and sum over $i$, obtaining
a degree-$4$ inequality involving
$\sum_i\langle f,\chi_i\rangle^4$; see \Cref{lem:cut-comparison-quartic-proxy}.
The sum-of-squares FKN theorem then bounds the quartic
deficit by the excess cut value in the full hypercube.
Choosing $\epsilon$ sufficiently small allows us to absorb
this error and obtain the desired inverse-polynomial bound.
The full argument appears in the proofs of
\Cref{lem:sos-cut-approximation,claim:sos-quartic-proxy-certificate}.

%% file: overview-degree-2-lower-bound.tex
\subsection{Limitations of the degree-2 relaxation}
\label{sec:overview-degree-2-lower-bound}

One may wonder whether the Fourier concentration underlying our algorithm can already be certified by the degree-$2$ relaxation. We show that this degree-$2$ relaxation does not provide the Fourier concentration needed by our analysis.
More precisely, in the regime $p=\Theta(\log d/d)$, its feasible region
contains a pseudoexpectation whose level-$1$ Fourier mass is at most
$1-\Omega(1/\log d)$.

\begin{theorem}[Degree-$2$ lower bound -- see \Cref{thm:degree-2-lower-bound}]
\label{thm:overview-degree-2-lower-bound}
For a large enough constant $C$ and \smash{$p = \frac{C \log d}{d}$}, with probability $1-o_d(1)$ over \smash{$G^* \sim Q_d(p)$}, there exists a degree-$2$ pseudoexpectation operator $\tilde\E$ that satisfies the system $\cP^*$ in \Cref{eq:overview-cut-constraints}, but $\tilde\E \sum_{i \in [d]}\langle f, \chi_i \rangle^2 \leq 1-\Omega(1/\log d)$.
\end{theorem}

The proof of this fact is technical and we defer it to \Cref{sec:degree-2-lower-bound}. We present here a very high-level overview.

\begin{proof}[Proof sketch]
The first step is to construct an SDP solution $X \succeq 0$ so that it acts orthogonally to the constant direction, its diagonal is all-ones, and its cut value in $G^*$ is small, i.e., $\frac{1}{p}\langle X, L_{G^*}\rangle \le 2-\Omega(1/\log d)$. Then, such $X$ satisfies $\cP^*$. Now, we can take a convex combination of $X$ with a character of constant degree but larger than $1$, say $\chi_{\{1,2\}}$. Note that \smash{$\frac{1}{p}\langle \chi_{\{1,2\}}^{\otimes 2}, L_{G^*} \rangle \approx \langle \chi_{\{1,2\}}^{\otimes 2}, L_{Q_d} \rangle = 4 $}. Thus, the matrix $X' = (1-\alpha) \cdot X + \alpha \cdot \chi_{\{1,2\}}^{\otimes 2}$ with $\alpha = \Omega(1/\log d)$ satisfies $\cP^*$, because $\frac{1}{p}\langle X', L_{G^*} \rangle \le (1-\alpha)(2-\Omega(1/\log d)) + O(\alpha) \le 2$. It remains to note that $X'$ has $\alpha=\Omega(1/\log d)$ Fourier mass on a character of degree $2$, and hence has at most $1-\Omega(1/\log d)$ mass on the coordinate cuts $\chi_i$.
\end{proof}

This result is in contrast with the algorithmic result in \Cref{thm:intro-main-recovery}: \Cref{lem:overview-quartic-proxy} shows that a degree-$4$ pseudoexpectation operator satisfies \smash{$\tilde\E \sum_{i \in [d]}\langle f, \chi_i \rangle^4 \ge 1-1/d^{O(1)}$}, while \Cref{thm:overview-degree-2-lower-bound} says that there are degree-$2$ pseudoexpectations that do not even satisfy the weaker inequality \smash{$\tilde\E \sum_{i \in [d]}\langle f, \chi_i \rangle^2 \ge 1-1/d^{O(1)}$}. In particular, one cannot hope to apply our approach to the degree-$2$ relaxation and recover the coordinate cuts up to an error better than $\Omega(1/\log d)$. We recall that when the recovery error for coordinate cuts is larger than $1/d$, we cannot draw any conclusion about the number of correctly recovered vectors, since it is possible that the errors spread and affect every single vertex.

%% file: preliminaries.tex
\section{Preliminaries}
\label{sec:preliminaries}

\paragraph{Notation.}
For $f,g:\{\pm 1\}^d \to \R$, we define $\langle f, g\rangle = \frac{1}{2^d} \sum_{x \in \{\pm 1\}^d} f(x) g(x) = \E_{x} f(x) g(x)$.
Correspondingly, we define \smash{$\Norm{f}^2 = \langle f, f\rangle$}.
We denote by $f \circ g$ the function $(f \circ g)(x) = f(x) g(x)$.
For $S \subseteq [d]$, we denote by $\chi_S : \{\pm 1\}^d \to \mathbb{R}$ the character function $\chi_S(x) = \prod_{ i \in S} x_i$.
We also denote, for convenience, $\chi_{\{i\}}$ by $\chi_i$.
For $x \in \{\pm 1\}^d$ and $i \in [d]$, we denote by
\(x^{\sim i}\) the vector obtained from \(x\) by flipping its \(i\)-th coordinate, i.e., $(x^{\sim i})_i =- x_i$, and $(x^{\sim i})_j = x_j $ for $j \neq i$.
For $i \in [d]$, we define the projection matrices  \smash{$ P_{\le i} = \frac{1}{2^d}\sum_{|S| \leq i}\chi_S^{\otimes 2}$} and \smash{$P_{\geq i} = \frac{1}{2^d}\sum_{|S| \geq i}\chi_S^{\otimes 2}$}.
The projection matrices $P_{< i}$ and $P_{>i}$ are defined analogously.

\subsection{Sum-of-squares proofs}

In this section we include a brief introduction to sum-of-squares proofs, taken from the corresponding section in~\cite{buhai2023beyond}.
For a more detailed exposition, see the lecture notes~\cite{barak2016proofs} or the monograph~\cite{MR4059250-Fleming19}.

\begin{definition}[Sum-of-squares proofs]
\label{def:sos-proof}
Let $p(x)$ and $q_1(x), ..., q_m(x)$ be polynomials over $x \in \mathbb{R}^n$ and let $\mathcal{A} = \{q_1(x) \geq 0, ..., q_m(x) \geq 0\}$ be a system of polynomial inequalities.
A \textit{sum-of-squares proof of degree $t$} that $p(x) \geq 0$ under $\mathcal{A}$ is an identity of the form
\[p(x) = \sum_{S \subseteq [m]} \left(\sum_{i=1}^{m_S} r_{S,i}(x)^2\right) \prod_{j\in S} q_j(x)\]
for polynomials $r_{S, i}(x)$, such that $\max_{S, i} \operatorname{deg}(r_{S, i}(x)^2 \prod_{j \in S} q_j(x)) \leq t$.
\end{definition}

If there exists a sum-of-squares proof of degree $t$ that $p(x) \geq 0$ under $\mathcal{A}$, we write $\mathcal{A} \sststile{t}{x} p(x) \geq 0$.
We also use the notation $\mathcal{A} \sststile{t}{x} p(x) \geq q(x)$ if $\mathcal{A} \sststile{t}{x} p(x) - q(x) \geq 0$ and $\mathcal{A} \sststile{t}{x} p(x) \leq q(x)$ if $\mathcal{A} \sststile{t}{x} q(x) - p(x) \geq 0$.
If $\mathcal{A}=\emptyset$, we omit it altogether and write $\sststile{t}{x} p(x) \geq 0$.
We also sometimes omit $\mathcal{A}$ if it is clear from context what axioms are assumed.
We note that if $\mathcal{A} \sststile{t}{x} p(x) \geq q(x)$ and $\mathcal{A} \sststile{t}{x} q(x) \geq r(x)$, then $\mathcal{A} \sststile{t}{x} p(x) \geq r(x)$, which allows writing chains of inequalities of the form $\mathcal{A} \sststile{t}{x} p(x) \geq s(x) \geq r(x)$.

\subsubsection{Pseudo-distributions and pseudo-expectations}

\begin{definition}[Pseudo-distributions]
A \textit{pseudo-distribution} $D$ of degree $t$ is a function from $\mathbb{R}^n$ to $\mathbb{R}$ with finite support such that $\sum_{x \in \operatorname{supp}(D)} D(x) = 1$ and $\sum_{x \in \operatorname{supp}(D)} D(x) p(x)^2 \geq 0$ for all polynomials $p(x)$ with $\operatorname{deg}(p^2) \leq t$.
\end{definition}

\begin{definition}[Pseudo-expectations]
Given a pseudo-distribution $D$ of degree $t$, the associated \textit{pseudo-expectation} $\tilde{\mathbb{E}}_{D(x)}$ is defined by $\tilde{\mathbb{E}}_{D(x)} f(x) = \sum_{x \in \operatorname{supp}(D)} D(x) f(x)$ for a function $f(x)$.
\end{definition}

We now define the notion of a pseudo-distribution that satisfies a set of polynomial inequalities.

\begin{definition}[Constrained pseudo-distributions]
A pseudo-distribution $D$ of degree $t$ \textit{satisfies} the set of polynomial inequalities $\mathcal{A} = \{q_1(x) \geq 0, ..., q_m(x)\geq 0\}$ if, for all $S \subseteq [m]$, $\tilde{\mathbb{E}}_{D(x)} r(x)^2 \prod_{j \in S} q_j(x) \geq 0$ for all polynomials $r(x)$ such that $\operatorname{deg}(r^2 \prod_{j \in S}q_j) \leq t$.

The pseudo-distribution $D$ \textit{approximately satisfies} $\mathcal{A}$ up to error $\eta$ if, under the same conditions as in the previous case, $\tilde{\mathbb{E}}_{D(x)} r(x)^2 \prod_{j \in S} q_j(x) \geq -\eta \|r^2\| \prod_{j \in S} \|q_j\|$, where $\|p(x)\|$ denotes the $2$-norm of the vector of coefficients of the polynomial $p(x)$.
\end{definition}

The connection between pseudo-distributions and sum-of-squares proofs is made in \Cref{fact:pseudoexpectation-soundness}, which shows that if a pseudo-distribution satisfies a set of polynomial inequalities, then it also satisfies any other polynomial inequalities derived from this set through sum-of-squares proofs.

\begin{fact}
\label{fact:pseudoexpectation-soundness}
If $D$ is a pseudo-distribution of degree $t$ that satisfies $\mathcal{A}$ and if $\mathcal{A} \sststile{s}{x} p(x) \geq 0$, then ${\tilde{\mathbb{E}}_{D(x)} r(x)^2 p(x) \geq 0}$ for all polynomials $r(x)$ such that $\operatorname{deg}(r^2) + s \leq t$.
If $D$ approximately satisfies $\mathcal{A}$ up to error $\eta$ and $p(x) = \sum_{S \subseteq [m]} \sum_{i=1}^{m_S} a_{S,i}(x)^2 \prod_{j \in S} q_j(x)$ is a degree-$s$ sum-of-squares proof, then for all polynomials $r(x)$ such that $\operatorname{deg}(r^2) + s \leq t$,
\[\tilde{\mathbb{E}}_{D(x)} r(x)^2 p(x) \geq -\eta \sum_{S \subseteq [m]} \sum_{i=1}^{m_S} \Norm{(r a_{S,i})^2} \prod_{j \in S} \Norm{q_j}\,.\]
\end{fact}

Finally, \Cref{fact:pseudoexpectation-computation} shows that there exists an algorithm with time complexity $(n+m)^{O(t)}$ to compute a pseudo-distribution of degree $t$ that approximately satisfies $\mathcal{A}$ up to error $2^{-n^{\Theta(t)}}$.

\begin{fact}
\label{fact:pseudoexpectation-computation}
For $x \in \mathbb{R}^n$, if $\mathcal{A} = \{q_1(x) \geq 0, ..., q_m(x) \geq 0\}$ is feasible and explicitly bounded,\footnote{Explicit boundedness means that $\mathcal{A}$ contains a constraint of the form $x_1^2 + ... + x_n^2 \leq B$.
In our applications it is possible to add such a constraint with $B$ large enough such that the constraint is always satisfied by the intended solution.}
then there exists an algorithm that runs in time $(n+m)^{O(t)}$ and computes a pseudo-distribution of degree $t$ that approximately satisfies $\mathcal{A}$ up to error $2^{-n^{\Theta(t)}}$.\footnote{In our applications the sum-of-squares proofs have polynomial bit complexity, so the error incurred through \Cref{fact:pseudoexpectation-soundness} is negligible.}
\end{fact}

\subsubsection{Sum-of-squares proofs for basic inequalities}

We state now some widely used sum-of-squares proofs that we use in the current paper.

\begin{fact}[Sum-of-squares almost triangle inequality, see Lemma A.2 in~\cite{kothari2017outlier}]
\label[fact]{fact:sos-triangle}
Let $x, y \in \R$ be variables. For any even $k \geq 2$,
\[\proves{x,y}{k} \Set{(x+y)^k \leq 2^{k-1} x^k + 2^{k-1}y^k}\,.\]
\end{fact}

\begin{fact}[Sum-of-squares Cauchy-Schwarz, see Lemma A.1 in~\cite{MR3631006-Ma16}]
\label[fact]{fact:sos-cauchy-schwarz}
    Let $x, y \in \R^n$ be vectors of variables. Then
    \[ \proves{x, y}{4} \left\{ \langle x, y \rangle^2 \leq \norm{x}^2 \norm{y}^2 \right\} \,.\]
    Furthermore, if $y \in \R^n$ is a constant vector, then
    \[ \proves{x}{2} \left\{ \langle x, y \rangle^2 \leq \norm{x}^2 \norm{y}^2 \right\} \,.\]
\end{fact}

\begin{fact}[Sum-of-squares Jensen inequality for square function]
\label[fact]{fact:sos-jensen}
Let $x_1, \ldots, x_n$ be variables. Then
\[\proves{x_1, \ldots, x_n}{2} \Set{\Paren{\frac{1}{n} \sum_{i=1}^n x_i}^2 \leq \frac{1}{n} \sum_{i=1}^n x_i^2}\,.\]
\end{fact}
\begin{proof}
By~\Cref{fact:sos-cauchy-schwarz},
\begin{align*}
\proves{x_1, \ldots, x_n}{2}
\Paren{\frac{1}{n} \sum_{i=1}^n x_i}^2
&\leq \Paren{\frac{1}{n} \sum_{i=1}^n x_i^2} \Paren{\frac{1}{n} \sum_{i=1}^n 1} = \frac{1}{n} \sum_{i=1}^n x_i^2\,.
\end{align*}
\end{proof}

\begin{fact}
\label[fact]{fact:sos-boolean-bounded}
Let $x$ be a variable. We have
\[\Set{x^2 = 1} \proves{x}{2} \Set{-1 \leq x \leq 1}\,.\]
\end{fact}
\begin{proof}
We have
\[\Set{x^2 = 1} \proves{x}{2} 1 + x = \frac{1}{2}(1+x)^2 - \frac{1}{2} (x^2-1) \geq 0\]
and
\[\Set{x^2 = 1} \proves{x}{2} 1 - x = \frac{1}{2}(1-x)^2 - \frac{1}{2} (x^2-1) \geq 0\,.\]
\end{proof}

\begin{fact}
\label[fact]{fact:sos-bounded-squares}
Let $x$ be a variable. For any constant $a \geq 0$, we have
\[\Set{-a \leq x \leq a} \proves{x}{2} \Set{x^2 \leq a^2}\,.\]
\end{fact}
\begin{proof}
We have
\[\Set{-a \leq x \leq a} \proves{x}{2} a^2-x^2=(a-x)(a+x) \geq 0\,.\]
\end{proof}

\subsubsection{Sum-of-squares proofs for basic analysis of Boolean functions}

Below are some sum-of-squares versions of some basic Fourier-analytic results.

\begin{claim}[Sum-of-squares Fourier expansion]
\label[claim]{claim:sos-fourier-expansion}
Let $\Set{f(x)}_{x \in \Set{\pm 1}^d}$ be variables. For all $x \in \Set{\pm 1}^d$, we have
\[\proves{f}{2} \Set{f(x) = \sum_{S \subseteq [d]} \langle f, \chi_S\rangle \chi_S(x)}\,.\]
\end{claim}
\begin{proof}
We have
\begin{align*}
\proves{f}{2} \sum_{S \subseteq [d]} \langle f, \chi_S\rangle \chi_S(x)
&= \sum_{S \subseteq [d]} \Paren{\frac{1}{2^d} \sum_{x' \in \Set{\pm 1}^d} f(x') \chi_S(x')} \chi_S(x)\\
&= \frac{1}{2^d} \sum_{x' \in \Set{\pm 1}^d} f(x') \sum_{S \subseteq [d]} \chi_S(x') \chi_S(x)\\
&= \frac{1}{2^d}\sum_{x' \in \Set{\pm 1}^d} f(x') \cdot 2^d \cdot \bm{1}[x' = x]\\
&= f(x)\,.
\end{align*}
\end{proof}

\begin{claim}[Sum-of-squares Parseval]
\label[claim]{claim:sos-parseval}
Let $\Set{f(x)}_{x \in \Set{\pm 1}^d}$ be variables. We have
\[\proves{f}{2} \sum_{S \subseteq [d]} \langle f, \chi_S\rangle^2 = \E_{ x} f( x)^2\,.\]
\end{claim}
\begin{proof}
Using \Cref{claim:sos-fourier-expansion} and the fact that $\E_{ x} \chi_S( x) \chi_{S'}( x) = 0$ for $S \neq S' \subseteq [d]$, we have
\begin{align*}
\proves{f}{2}
\E_{ x} f( x)^2
&= \E_{ x} \Paren{\sum_{S \subseteq [d]} \langle f, \chi_S\rangle \chi_S( x)}^2
= \sum_{S \subseteq [d]} \langle f, \chi_S\rangle^2\,.
\end{align*}
\end{proof}

\begin{claim}[Sum-of-squares Parseval for Boolean-valued functions]
\label[claim]{claim:sos-boolean-parseval}
Let $\Set{f(x)}_{x \in \Set{\pm 1}^d}$ be variables. We have
\[\Set{f(x)^2 = 1, \forall x \in \{\pm 1\}^d} \proves{f}{2} \Set{\sum_{S \subseteq [d]} \langle f, \chi_S\rangle^2 = 1}\,.\]
\end{claim}
\begin{proof}
Using \Cref{claim:sos-parseval}, we have
\begin{align*}
\Set{f(x)^2 = 1, \forall x \in \{\pm 1\}^d}
&\proves{f}{2} \sum_{S \subseteq [d]} \langle f, \chi_S\rangle^2 = \E_{\bm x} f(\bm x)^2 = 1\,.
\end{align*}
\end{proof}

\subsection{Sum-of-squares rounding}

We will make use of the following rounding algorithm guarantee from Theorem 5.2 in \cite{MR3631006-Ma16}.

\begin{theorem}[See Theorem 5.2 in~\cite{MR3631006-Ma16}]
\label{thm:sos-rounding}
For every $\ell \in \N$, there exists an $(n+d+d')^{O(\ell)}$-time algorithm with the following property:
Let $\epsilon > 0$ be smaller than some constant.
Let $d, d' \in \N$ be numbers.
Let $P : \R^d \to \R^{d'}$ be a polynomial with degree $\deg P \leq \ell$.
Let $\Set{a_1, \ldots, a_n} \subseteq \R^{d}$ be a set of vectors such that $b_1 = P(a_1)$, $\ldots$, $b_n = P(a_n) \in \R^{d'}$ all have norm at least $1-\epsilon$ and $\Norm{\sum_{i=1}^n b_i b_i^\top} \leq 1+\epsilon$.
Let $\cA$ be a system of polynomial inequalities in variables $u=(u_1, \ldots, u_d)$ such that the vectors $a_1, \ldots, a_n$ satisfy $\cA$ and
\[\cA \proves{u}{\ell} \Set{\sum_{i=1}^n \langle b_i, P(u)\rangle^4 \geq (1-\epsilon) \Norm{P(u)}^4}\,.\]
Then, the algorithm on input $\cA$ and $P$ outputs with high probability a set of unit vectors $\Set{b_1', \ldots, b_n'} \subseteq \R^{d'}$ such that
\[d_H\Paren{\Set{b_1^{\otimes 2}, \ldots, b_n^{\otimes 2}}, \Set{(b'_1)^{\otimes 2}, \ldots, (b'_n)^{\otimes 2}}} \leq O(\sqrt{\epsilon})\,,\]
where the Hausdorff distance $d_H(X, Y)$ is the maximum of
\[
\max_{x \in X} \min_{y \in Y} \Norm{x-y} \qquad \text{and} \qquad \max_{y \in Y} \min_{x \in X} \Norm{x-y}.
\]
\end{theorem}

\subsection{Concentration inequalities}
We will make use of the following concentration inequalities.
\begin{lemma}[Bernstein's inequality, see e.g., Theorem 2.8.4 in \cite{Vershynin}]
\label{lem:bernstein}
Let $K>0$ and let $X_1, \dots, X_N$ be independent mean-zero random variables such that $|X_i| \le K$ for all $i \in [N]$. Then, for every $t \ge 0$, one has
\begin{equation*}
    \Pr\left[\left|\sum_{i=1}^N X_i\right| \ge t\right] \le 2 \exp\left(-\frac{t^2/2}{\sum_{i=1}^N \E[X_i^2] + Kt/3}\right) \, .
\end{equation*}
\end{lemma}

\begin{lemma}[Matrix Bernstein's inequality, see e.g., Theorem 1.6.2 in \cite{Tropp15}]
\label{lem:matrix-bernstein}
Let $K>0$ and let $X_1, \dots, X_N$ be independent mean-zero $n \times n$ symmetric random matrices such that $\|X_i\| \le K$ for all $i \in [N]$. Then, for every $t \ge 0$, one has
\begin{equation*}
    \Pr\left[\left\|\sum_{i=1}^N X_i\right\| \ge t\right] \le 2n \exp\left(-\frac{t^2/2}{\|\sum_{i=1}^N \E[X_i^2]\| + Kt/3}\right) \, .
\end{equation*}
\end{lemma}

\begin{lemma}[Bernstein's inequality with bounded dependence, Theorem 2.3 in \cite{Janson04}]
\label{lem:bernstein-bounded-dependence}
Let $K>0, \Delta \ge 1$ and let $X_1, \dots, X_N$ be mean-zero random variables such that their dependency graph has maximum degree at most $\Delta$ and $|X_i| \le K$ for all $i \in [N]$. Then, for every $t \ge 0$, one has
\begin{equation*}
    \Pr\left[\left|\sum_{i=1}^N X_i\right| \ge t\right] \le 2 \exp\left(-\frac{t^2/4}{(\Delta+1)\left(\sum_{i=1}^N \E[X_i^2] + Kt/3\right)}\right) \, .
\end{equation*}
\end{lemma}

%% file: sos-fkn.tex
\section{Sum-of-squares FKN theorem}
\label{sec:sos-fkn}

Below we present a degree-$4$ sum-of-squares version of the Friedgut-Kalai-Naor theorem~\cite{MR1942632-Friedgut02}.
Our proof follows the outline of the proof of FKN given in Chapter 9 of O'Donnell~\cite{MR3443800-ODonnell14}.

\begin{theorem}[Sum-of-squares FKN theorem]
\label{thm:sos-fkn}
Consider the following system of polynomial inequalities in variables $\{f(x)\}_{x \in \{\pm 1\}^d}$ and $\delta$:
\[
\cP = \begin{cases}
\sum_{i \in [d]} \langle f,\chi_i\rangle^2\ge 1-\delta\,, \\
f(x)^2 = 1\,, \forall x \in \{\pm 1\}^d\,, \\
\delta \geq 0\,.
\end{cases}
\]
Then $\cP \proves{f, \delta}{4} \Set{\sum_{i \in [d]} \langle f, \chi_i\rangle^4 \geq 1 - O(\delta)}$.
Furthermore, $\delta$ only appears in the sum-of-squares proof as part of degree-$1$ monomials.
\end{theorem}

\begin{proof}[Proof of \Cref{thm:sos-fkn}]
Define the function $\ell : \Set{\pm 1}^d \to \mathbb{R}$ as the projection of $f$ onto Fourier level $1$: $\ell(x) = \sum_{i \in [d]} \langle f, \chi_i\rangle \chi_i(x)$.
Because $\cP \proves{f,\delta}{2} \sum_{i \in [d]} \langle f, \chi_i\rangle^2 \geq 1-\delta$, it follows that $\ell$ is close to $f$: specifically, using \Cref{claim:sos-fourier-expansion} and \Cref{claim:sos-boolean-parseval}, we get
\begin{align*}
\cP
\proves{f, \delta}{2} \E_{ {x}} (f( {x})-\ell( {x}))^2 &= \E_{ {x}} \Paren{\sum_{\substack{S \subseteq [d] \\ |S| \neq 1}} \langle f, \chi_S\rangle \chi_S( {x})}^2
= \E_{  x} \sum_{\substack{S, S' \subseteq [d]\\ |S|, |S'| \neq 1}} \langle f, \chi_S \rangle \langle f, \chi_{S'}\rangle \chi_S(  x) \chi_{S'}(  x)\\
&= \sum_{\substack{S \subseteq [d] \\ |S| \neq 1}} \langle f, \chi_S\rangle^2
=
1 - \sum_{i \in [d]} \langle f, \chi_i\rangle^2 \leq \delta\,.
\end{align*}

Next, we prove that the function $\ell^2$ is close to the constant function $1$. Using that $\cP \proves{f,\delta}{2} f(x)^2 = 1$ for all $x \in \Set{\pm 1}^d$, we have for all $x \in \Set{\pm 1}^d$ that
\begin{align*}
\cP \proves{f,\delta}{2}
\ell(x)^2 - 1= \ell(x)^2 - f(x)^2 = (\ell(x) - f(x))(\ell(x)+f(x))\,.
\end{align*}
Let $\Pi_{\leq 2}$ be the linear operator that projects onto Fourier level at most $2$.
Note that, since $\ell$ has degree at most $1$, $\ell^2$ has degree at most $2$.
Then, using the sum-of-squares triangle inequality in \Cref{fact:sos-triangle},
\begin{align*}
\cP \proves{f,\delta}{4} \E_{  x}(\ell(  x)^2 - 1)^2
&= \E_{  x} ((\Pi_{\leq 2}((\ell - f)(\ell + f)))(  x))^2\\
&= \E_{  x} ((\Pi_{\leq 2}((\ell - f)\ell))(  x) + (\Pi_{\leq 2}((\ell - f)f))(  x))^2\\
&\leq 2\E_{  x} ((\Pi_{\leq 2}((\ell - f)\ell))(  x))^2 + 2\E_{  x} ((\Pi_{\leq 2}((\ell - f) f))(  x))^2\,.
\end{align*}
We now remove the projection in the second term; this produces a valid upper bound because, for $h = (\ell -f)f$, we have by Parseval (\Cref{claim:sos-parseval}) that
\[\cP \proves{f,\delta}{4} \E_{  x} ((\Pi_{\leq 2} h)(x))^2 = \sum_{S \subseteq [d], |S| \leq 2} \langle h, \chi_S\rangle^2 \leq \sum_{S \subseteq [d]} \langle h, \chi_S\rangle^2 = \E_{  x} (h(x))^2\,.\]
Hence,
\begin{align*}
\cP \proves{f,\delta}{4} \E_{  x}(\ell(  x)^2 - 1)^2
&\leq 2\E_{  x} ((\Pi_{\leq 2}((\ell - f)\ell))(  x))^2 + 2\E_{  x} (\ell(  x) - f(  x))^2f(  x)^2\\
&= 2\E_{  x} ((\Pi_{\leq 2}((\ell - f)\ell))(  x))^2 + 2\E_{  x} (\ell(  x) - f(  x))^2\\
&\leq 2\E_{  x} ((\Pi_{\leq 2}((\ell - f)\ell))(  x))^2 + 2\delta\,.
\end{align*}
Let us now bound the first term.
We have for all $x \in \Set{\pm 1}^d$ that
\begin{align*}
\cP \proves{f,\delta}{2}
&
(\Pi_{\leq 2}((\ell - f)\ell))(x)\\
&= -\sum_{\substack{S \subseteq [d]\\ |S| \neq 1}} \sum_{i \in [d]} \langle f, \chi_S\rangle \langle f, \chi_i\rangle \chi_{S \Delta \{i\}}(x) \ind[|S \Delta \{i\}| \leq 2]\\
&= -\sum_{i \in [d]} \Paren{\langle f, \chi_{\emptyset}\rangle \langle f, \chi_i\rangle + \sum_{\substack{j \in [d]\\ j \neq i}} \langle f, \chi_{\{i,j\}}\rangle \langle f, \chi_{j}\rangle} \chi_{i}(x)\\*
&\quad - \sum_{\substack{i, j \in [d]\\i < j}} \Paren{\sum_{\substack{k \in [d]\\ k \neq i,j}} \langle f, \chi_{\{i,j,k\}}\rangle \langle f, \chi_k\rangle} \chi_{\{i,j\}}(x)\,,
\end{align*}
so by sum-of-squares Parseval (\Cref{claim:sos-parseval}),
\begin{align*}
\cP \proves{f,\delta}{4}
&\E_{  x} ((\Pi_{\leq 2}((\ell - f)\ell))(  x))^2\\
&\leq \sum_{i \in [d]} \Paren{\langle f, \chi_{\emptyset}\rangle \langle f, \chi_i\rangle + \sum_{\substack{j \in [d]\\ j \neq i}} \langle f, \chi_{\{i,j\}}\rangle \langle f, \chi_{j}\rangle}^2 + \sum_{\substack{i, j \in [d]\\i < j}} \Paren{\sum_{\substack{k \in [d]\\ k \neq i,j}} \langle f, \chi_{\{i,j,k\}}\rangle \langle f, \chi_k\rangle}^2\,.
\end{align*}
Applying the sum-of-squares triangle inequality (\Cref{fact:sos-triangle}) and Cauchy-Schwarz (\Cref{fact:sos-cauchy-schwarz}),
\begin{align*}
\cP \proves{f,\delta}{4}
&\E_{  x} ((\Pi_{\leq 2}((\ell - f)\ell))(  x))^2\\
&\leq 2 \langle f, \chi_{\emptyset}\rangle^2 \sum_{i \in [d]} \langle f, \chi_i\rangle^2 + 2 \sum_{i \in [d]} \Paren{\sum_{\substack{j \in [d]\\j \neq i}}\langle f, \chi_{\{i,j\}}\rangle^2}\Paren{\sum_{\substack{j \in [d]\\j \neq i}} \langle f, \chi_j\rangle^2}\\
&\quad + \sum_{\substack{i,j \in [d]\\i < j}} \Paren{\sum_{\substack{k \in [d]\\k \neq i,j}}\langle f, \chi_{\{i,j,k\}}\rangle^2} \Paren{\sum_{\substack{k \in [d]\\k\neq i,j}} \langle f, \chi_k\rangle^2}\,.
\end{align*}
Finally, using that by \Cref{claim:sos-boolean-parseval} $\cP \proves{f,\delta}{2} \sum_{S \subseteq [d]} \langle f, \chi_S\rangle^2 = 1$,
\begin{align*}
\cP \proves{f,\delta}{4}
&\E_{  x} ((\Pi_{\leq 2}((\ell - f)\ell))(  x))^2\\
&\leq 2 \langle f, \chi_{\emptyset}\rangle^2 + 2 \sum_{i \in [d]} \Paren{\sum_{\substack{j \in [d]\\j \neq i}}\langle f, \chi_{\{i,j\}}\rangle^2}
+ \sum_{\substack{i,j \in [d]\\i < j}} \Paren{\sum_{\substack{k \in [d]\\k \neq i,j}}\langle f, \chi_{\{i,j,k\}}\rangle^2}\\
&\leq O\Paren{1 - \sum_{i \in [d]} \langle f, \chi_i\rangle^2}\\
&\leq O(\delta)\,.
\end{align*}
Hence, overall,
\[\cP \proves{f,\delta}{4} \E_{  x}(\ell(  x)^2 - 1)^2 \leq O(\delta)\,.\]

We have proved that the function $\ell^2$ is close to the constant function $1$.
Let us consider now the variance term $\E_{ {x}} (\ell( {x})^2-\E_{  x'} \ell(  x')^2 )^2$.
This term is upper bounded by the version with the mean replaced by $1$:
\begin{align*}
\cP \proves{f, \delta}{4} \E_{ {x}} (\ell( {x})^2-\E_{  x'} \ell(  x')^2 )^2  \leq \E_{ {x}} (\ell( {x})^2-1)^2\,,
\end{align*}
as it is equivalent after expanding all terms to $0 \leq (\E_{  x} \ell(  x)^2 - 1)^2$.

To close the proof, we derive the following relation on the variance term:
\begin{align*}
\cP
\proves{f, \delta}{4} \frac{1}{2} \E_{ {x}} (\ell( {x})^2 - \E_{ {x}'} \ell( {x}')^2)^2
&= \frac{1}{2} \E_{  x} \ell(  x)^4 - \frac{1}{2} \Paren{\E_{  x} \ell( {x})^2}^2\\
&= \frac{1}{2} \E_{  x} \Paren{\sum_{i \in [d]} \langle f, \chi_i\rangle \chi_i(  x)}^4 - \frac{1}{2} \Paren{\sum_{i \in [d]} \langle f, \chi_i\rangle^2}^2\\
&= \sum_{i \neq j \in [d]} \langle f, \chi_i\rangle^2 \langle f, \chi_j\rangle^2\\
&= \Paren{\sum_{i \in [d]} \langle f, \chi_i\rangle^2}^2 - \sum_{i \in [d]} \langle f, \chi_i\rangle^4\\
&= \Paren{1 - \Paren{ 1- \sum_{i \in [d]} \langle f, \chi_i\rangle^2}}^2 - \sum_{i \in [d]} \langle f, \chi_i\rangle^4\\
&\geq 1 - 2\Paren{ 1- \sum_{i \in [d]} \langle f, \chi_i\rangle^2}  - \sum_{i \in [d]} \langle f, \chi_i\rangle^4\\
&\geq 1-2\delta - \sum_{i \in [d]} \langle f, \chi_i\rangle^4\,.
\end{align*}

Using our upper bound on $\E_{ {x}} (\ell( {x})^2 - 1)^2$, we get that
\begin{align*}
\cP
\proves{f, \delta}{4} 1-2\delta -\sum_{i \in [d]} \langle f, \chi_i\rangle^4 \leq O(\delta)\,,
\end{align*}
and the conclusion follows by rearranging.
\end{proof}

%% file: spectral-concentration.tex
\newcommand{\Plow}{\ensuremath{P}}
\newcommand{\Phigh}{\ensuremath{Q}}

\section{Spectral concentration}\label{sec:spectral-concentration}
For ease of notation, here and in the subsequent sections we denote the unnormalized Laplacians of the cube and the sampled graph by $L$ and $L'$, respectively, corresponding to $L_{Q_d}$ and $L_{G^*}$ in the technical overview (\Cref{sec:technical-overview}).

In this section, we prove the following spectral concentration result for the percolated hypercube:

\begin{lemma}[Spectral concentration of the percolated hypercube]\label[lemma]{lem:spectral-concentration}
Fix any $\epsilon \in (0, 1/10)$, $C_1\ge 1$, and $C_2 \ge 20$.
Let $Q_d$ be the $d$-dimensional hypercube, and let $Q_d'$ be a random subgraph of $Q_d$ obtained by including each edge of $Q_d$ independently with probability $p$.
Also let $L$ be the unnormalized Laplacian of $Q_d$ and let $L'$ be the unnormalized Laplacian of $Q_d'$.
For $d \geq 1$ large enough and $\frac{C \log d}{d} \le p \le 1$ for some $C = C(\epsilon, C_1, C_2)$, the following holds with probability at least $1-1/d^{C_1}$:
There exists a set $R \subseteq \Set{\pm 1}^d$ such that $|R| \leq 2^d/d^{C_2}$ and
\begin{equation*}
    -\epsilon pL-d^2P_{>0}\Pi_R P_{>0}
    \preceq
    L'-pL
    \preceq
    \epsilon pL+d^2 P_{>0}\Pi_R P_{>0} \, ,
\end{equation*}
where $P_{>0} = I-\frac{1}{2^d} \chi_\emptyset^{\otimes 2}$ and $\Pi_R$ is the diagonal matrix with $1$ on the entries corresponding to vertices in $R$.
\end{lemma}

The proof is organized as follows. Our set $R$ will be the set of vertices in $Q_d'$ whose degree deviates significantly from their mean $pd$.
First, in \Cref{lem:size-bad-degree-set}, we show that with high probability the claimed bound on the size of $R$ holds. Next, in \Cref{lem:good-degree-spectral-concentration}, we show that after projecting onto the complement of $R$ (i.e., onto the set of vertices whose degree is well concentrated around $pd$), the Laplacians $L'$ and $pL$ are close in spectral norm. After that, we split the space according to Fourier degree. In \Cref{lem:low-degree-spectral-approx}, we handle the contribution from the low-degree Fourier space directly by applying concentration inequalities. Then, in \Cref{lem:high-degree-spectral-lb}, we handle the contribution from the high Fourier degree space.
Finally, we  use a Schur-complement argument in the final step to handle the contribution from the cross terms.

We start by setting up notation and parameters.
Throughout this section, let $V_d=\{\pm1\}^d$ and $n=|V_d|=2^d$, and write $Q_d=(V_d,H_d)$ and $Q_d'=(V_d,E)$.
Here $Q_d'$ denotes the sampled subgraph in its original cube coordinates, before any vertex renaming.
For an edge set $F$, let $F(W)$ denote the edges with both endpoints in $W$, and let $F(S,T)$ denote the edges with one endpoint in $S$ and the other in $T$, counting each edge once even when $S$ and $T$ overlap.

Let $\epsilon$, $C_1$, $C_2$ be as in the lemma statement. Suppose that $d$ is sufficiently large, so that $d \ge \max\{C_1^2, 200\}$.
Let $\gamma = \epsilon/1000$, and let $\eta =\gamma^2/C^2_\bl$, where $C_\bl \geq 1$ denotes the constant from \Cref{lem:converse-expander-mixing}.
Let $C = \max\{240(C_1+15)/\eta^4 , 4(C_1+C_2+11)/\gamma^{2}\}$  be a constant depending on $\epsilon,C_1$ and $C_2$.
Let $\frac{C\log d}{d} \le p \le 1$ (in particular, $d$ has to be large enough to satisfy this inequality).

Given a subgraph $G = (V_d,E)$ of $Q_d$, we define $R_G = \{x \in V_d: |\deg_G(x)-pd|>\gamma pd\}$ and $U_G = V_d\setminus R_G$.
Let $R$ be the random set
$$ R \coloneqq R_{Q_d'} = \{x \in V_d: |\deg_{Q_d'}(x)-pd|>\gamma pd\}\,.$$
Additionally, write $U \coloneqq U_{Q_d'}$. We start by bounding the size of the trimmed set $R$.

\begin{lemma}[Small bad-degree set]
\label{lem:size-bad-degree-set}
    With probability at least $1-1/d^{C_1+10}$ over the sampling of $Q_d'$, one has $|R| \le n/d^{C_2}$.
\end{lemma}
\begin{proof}
A vertex $u \in V_d$ is in $R$ if its degree deviates from its mean $pd$ by more than a relative factor of $\gamma$.  By a Chernoff bound, for a fixed $u$ this happens with probability at most
\begin{equation*}
    2 \exp(-\gamma^2 pd/3) \le 2 \exp(-\gamma^2 C/3 \cdot \log d) \le 2 d^{-(C_1+C_2+11)} \le 1/d^{C_1+C_2+10} \, ,
\end{equation*}
where the second transition follows by the setting of $C$.
Thus, the expected size of $R$ is at most $n/d^{C_1+C_2+10}$. By Markov's inequality, $\Pr[|R| > n/d^{C_2} ] \leq  1/d^{C_1+10}$.
\end{proof}

Next, in \Cref{lem:good-degree-spectral-concentration} we will show that after projecting onto $U$, the matrix $L'-pL$ has small spectral norm. We start by showing in \Cref{lem:good-degree-cut-concentration}, that with high probability, for every pair of disjoint sets $S,T \subseteq U$, the bilinear form of $L'-pL$ against $\ind_S$ and $\ind_T$ is small. The proof uses the standard edge-isoperimetric inequality for the hypercube, which we recall below.

 \begin{lemma}[Hypercube subgraph isoperimetry, see \cite{Bollobas86}]
\label{lem:cube-subgraph-edge-bound}
    For any $\emptyset \neq W \subseteq V_d$, one has $|H_d(W)| \le \frac{1}{2}|W| \log_2 |W|$.
\end{lemma}

\begin{lemma}[Concentration of good-degree cuts]
\label{lem:good-degree-cut-concentration}
    With probability at least $1-1/d^{C_1+10}$ over the sampling of $Q_d'$, for every disjoint $S,T \subseteq U$ one  has
    \begin{equation*}
        \left|\langle \ind_S,  L' \ind _T \rangle - p\langle \ind_S,  L \ind_T \rangle\right| \le \eta pd \cdot \|\ind_S\|_2 \|\ind_T\|_2 \, .
    \end{equation*}
\end{lemma}
\begin{proof}
The condition $ \left|\langle \ind_S,  L' \ind _T \rangle - p\langle \ind_S,  L \ind_T \rangle\right| \le \eta pd \cdot \|\ind_S\|_2 \|\ind_T\|_2 $ can be equivalently written as
\begin{equation}\label{eq:good-degree-cut-discrepancy}
    \left||E(S,T)|-p |H_d(S,T)|\right| \leq \eta \cdot pd \sqrt{|S||T|}\,.
\end{equation}
The proof is organized as follows. First, in \Cref{claim:imbalanced-set-discrepancy}, we show that the condition \eqref{eq:good-degree-cut-discrepancy} holds deterministically for all imbalanced pairs of sets. Next, in \Cref{claim:small-set-discrepancy}, we show that it also holds deterministically for all small balanced pairs of sets. We then handle the remaining case in \Cref{claim:balanced-set-discrepancy-tail}, where we bound the probability that condition \eqref{eq:good-degree-cut-discrepancy} fails for each large balanced pair. Finally, we complete the proof by carefully applying a union bound.
\begin{claim}
\label{claim:imbalanced-set-discrepancy}
    For every edge subgraph $G=(V_d,E)$ of $Q_d$, and for every pair of disjoint $\emptyset \neq A,B \subseteq U_G$ with $|A|/|B| \notin [ \eta^2/16, 16/\eta^2]$, one has $||E(A,B)|-p |H_d(A,B)|| \le \eta \cdot pd \sqrt{|A||B|}$.
\end{claim}
\begin{proof}
    Since $A,B \subseteq U_G$, each $u \in A\cup B$ has degree at most $(1+\gamma)pd \le 2 pd$. Hence,
    \begin{equation*}
        \left ||E(A,B)|-p |H_d(A,B)|\right| \le |E(A,B)| + p |H_d(A,B)| \leq 3pd \min\{|A|,|B|\} \, .
    \end{equation*}
    Because we assume $|A|/|B| \notin [ \eta^2/16, 16/\eta^2]$, we have
    \begin{equation*}
        \sqrt{|A||B|} = \sqrt{\max\{|A|,|B|\} \min \{|A|,|B|\}}\ge \frac{4}{\eta} \cdot \min\{|A|,|B|\} \, .
    \end{equation*}
    Thus we conclude
    \begin{equation*}
               \left ||E(A,B)|-p |H_d(A,B)|\right| \le 3pd \min\{|A|,|B|\} \le \eta pd \sqrt{|A||B|}\, .
    \end{equation*}
\end{proof}

Next, we handle pairs of small balanced sets.
\begin{claim}
\label{claim:small-set-discrepancy}
    For every edge subgraph $G=(V_d,E)$ of $Q_d$, and for every pair of  $\emptyset \neq A,B \subseteq V_d$ with $|A|/|B| \in [ \eta^2/16, 16/\eta^2]$ and $|A\cup B| \le d$, one has $||E(A,B)|-p |H_d(A,B)|| \le \eta \cdot pd \sqrt{|A||B|}$.
\end{claim}
\begin{proof}
    By the assumption that $|A|/|B| \in [ \eta^2/16, 16/\eta^2]$, we have
    \begin{equation*}
    \sqrt{|A||B|} =  \sqrt{\max\{|A|,|B|\} \min \{|A|,|B|\}}\ge \frac{\eta}{4} \cdot \max \{|A|,|B|\} \ge  \frac{\eta}{8}\cdot  |A\cup B|\,.
    \end{equation*}
   By the assumption that $p \geq C \log d/d$, this gives
\begin{equation*}
    \eta pd \cdot \sqrt{|A||B|} \ge \frac{\eta^2 }{8} C |A\cup B| \log d \, .
\end{equation*}
On the other hand, we have
\begin{align*}
    \left||E(A,B)| - p |H_d(A,B)|\right| & \le |E(A,B)| + p|H_d(A,B)|\,, \\
    & \le 2|H_d(A,B)|\,, \\
    & \le |A\cup B| \log_2 |A\cup B| && \text{by \Cref{lem:cube-subgraph-edge-bound} }\,, \\
    & \le |A\cup B| \log_2 d && \text{by the assumption that $|A \cup B| \leq d$ }\,.
\end{align*}
One can check that our assumption on $C$ satisfies
\begin{equation*}
    \log_2 d \le \frac{\eta^2 }{8} C \frac{\log_2 d}{\log_2 \e} \, ,
\end{equation*}
so we conclude $||E(A,B)| - p |H_d(A,B)|| \le \eta pd \cdot \sqrt{|A||B|}$.
\end{proof}

Finally, it remains to handle large balanced pairs of sets. The following claim gives an upper bound, in terms of $|A \cup B|$, on the failure probability for any fixed pair of balanced sets $A$ and $B$.
\begin{claim}
\label{claim:balanced-set-discrepancy-tail}
    Let $\emptyset \neq A,B \subseteq V_d$ be subsets such that $|A|/|B| \in [ \eta^2/16, 16/\eta^2]$. Then
    \begin{equation*}
        \Pr\left[\left||E(A,B)| - p |H_d(A,B)|\right| > \eta pd \cdot \sqrt{|A||B|} \right]  \le \frac{1}{d^{(C_1+14)\cdot |A\cup B|}} \, .
    \end{equation*}
\end{claim}
\begin{proof}
    We apply \Cref{lem:bernstein} to the random variables $X_e = \ind(e \in E) - p$, for $e \in H_d(A,B)$.
    To do so, we take $K=1$ and $t= \eta p d \sqrt{|A||B|}$, and note that for every $e \in H_d$ we have $\E[X_e^2] = \Var [X_e] = p(1-p) \leq p$. Then, we get
    \begin{align*}
         \Pr\left[\left||E(A,B)| - p |H_d(A,B)|\right| > \eta pd \cdot \sqrt{|A||B|} \right]
        & \le 2 \exp\left(-\frac{\eta^2 p^2d^2 \cdot |A||B|}{2p|H_d(A,B)|+\eta pd \cdot \sqrt{|A||B|}}\right) \\
        & \le 2 \exp\left(-\frac{\eta^2 p^2d^2 \cdot |A||B|}{2pd |A\cup B|+\eta pd \cdot |A\cup B|}\right) \, .
    \end{align*}
    Again by the assumption that $|A|/|B| \in [ \eta^2/16, 16/\eta^2]$, we have
    \begin{equation*}
        |A||B| \ge \frac{\eta^2}{16} \max\{ |A|, |B|\} ^2 \geq \frac{\eta^2}{64} |A\cup B|^2 \, ,
    \end{equation*}
    so continuing from above we get
    \begin{align*}
         \Pr\left[\left||E(A,B)| - p |H_d(A,B)|\right| > \eta pd \cdot \sqrt{|A||B|} \right]
        & \le 2 \exp\left(-\frac{\eta^2 p^2 d^2 \cdot \frac{\eta^2}{64} |A\cup B|^2}{2pd |A\cup B|+\eta pd \cdot |A\cup B|}\right ) \\
        & \le 2 \exp \left(-\frac{1}{192} \eta^4 pd |A\cup B|\right) \\
        & \le 2 \exp \left(-\frac{1}{192} \eta^4 C |A\cup B| \log d \right)\, .
    \end{align*}
    We conclude the claim by our assumption on $C$.
\end{proof}
We are now ready to union bound over pairs of disjoint sets $S,T \subseteq U$. Given disjoint $S,T \subseteq V_d$, let us denote by $k_{S,T}$ the number of connected components of $Q_d[S\cup T]$. Furthermore, we denote these connected components by \smash{$\{C_{S,T}^{(1)}, \dots, C_{S,T}^{(k_{S,T})}\}$} and let \smash{$S_i = C_{S,T}^{(i)} \cap S$} and \smash{$T_i = C_{S,T}^{(i)} \cap T$} for each $i \in [k_{S,T}]$. With this notation, we have
\begin{align*}
    & \Pr\left[\exists \, S,T \subseteq U \text{ s.t. } S\cap T =\emptyset \text{ and } \left|\langle \ind_S,  L' \ind _T \rangle - p\langle \ind_S,  L \ind_T \rangle\right| > \eta pd \cdot \|\ind_S\|_2 \|\ind_T\|_2 \right] \\
    = & \Pr\left[\exists \, S,T \subseteq U \text{ s.t. } S\cap T =\emptyset \text{ and } ||E(S,T)|-p |H_d(S,T)|| > \eta \cdot pd \sqrt{|S||T|} \right] \\
    \le & \Pr\left[\exists \, S,T \subseteq U \text{ s.t. } S\cap T =\emptyset \text{ and } \sum_{i=1}^{k_{S,T}} ||E(S_i,T_i)|-p |H_d(S_i,T_i)|| > \eta \cdot pd \sum_{i=1}^{k_{S,T}}\sqrt{|S_i||T_i|} \right] \\
    \le & \Pr\left[\exists \, S,T \subseteq U \text{ s.t. } S\cap T =\emptyset \text{ and } \exists \, i \in [k_{S,T}] \text{ s.t. } ||E(S_i,T_i)|-p |H_d(S_i,T_i)|| > \eta \cdot pd \sqrt{|S_i||T_i|} \right] \, .
\end{align*}
Here the first inequality used that
$$||E(S,T)|-p |H_d(S,T)|| =  \left| \sum_{i=1}^{k_{S,T}} (|E(S_i,T_i)|-p |H_d(S_i,T_i)|) \right| \leq  \sum_{i=1}^{k_{S,T}} ||E(S_i,T_i)|-p |H_d(S_i,T_i)||,$$ since the $S_i \cup T_i$ are the connected components of $Q_d[S \cup T]$, and that by Cauchy-Schwarz, $$    \sum_{i=1}^{k_{S,T}} \sqrt{|S_i||T_i|}
    \le
    \sqrt{\left(\sum_{i}|S_i|\right)
          \left(\sum_{i} |T_i|\right)}
=
    \sqrt{|S||T|}\,.$$

By \Cref{claim:imbalanced-set-discrepancy}, the last probability above equals
\begin{equation*}
\Pr\left[
\begin{aligned}
&\exists\, S,T \subseteq U \text{ such that } S\cap T=\emptyset
\text{ and } \exists\, i \in [k_{S,T}] \text{ such that} \\
&\qquad S_i,T_i \text{ are non-empty}, \\
&\qquad \frac{|S_i|}{|T_i|}
    \in \left[\frac{\eta^2}{16}, \frac{16}{\eta^2}\right], \\
&\qquad \left||E(S_i,T_i)|-p |H_d(S_i,T_i)|\right|
    > \eta pd \sqrt{|S_i||T_i|}
\end{aligned}
\right].
\end{equation*}
This can in turn be upper-bounded by
\begin{equation*}
\Pr\left[
\begin{aligned}
&\exists\, W \subseteq V_d \text{ such that } Q_d[W] \text{ is connected}\\
&\text{ and } \exists\, \text{ disjoint non-empty } A,B \subseteq W \text{ such that} \\
&\qquad A \cup B = W, \\
&\qquad \frac{|A|}{|B|}\in \left[\frac{\eta^2}{16}, \frac{16}{\eta^2}\right], \\
&\qquad \left||E(A,B)| - p |H_d(A,B)|\right|
    > \eta pd \cdot \sqrt{|A||B|}
\end{aligned}
\right].
\end{equation*}
By \Cref{claim:small-set-discrepancy}, we can further upper-bound this quantity as
\begin{equation}\label{eq:large-connected-set-failure-probability}
\Pr\left[
\begin{aligned}
&\exists\, W \subseteq V_d \text{ such that } |W| \ge d, \, Q_d[W] \text{ is connected,}\\
&\text{ and } \exists\, \text{ disjoint non-empty } A,B \subseteq W \text{ such that} \\
&\qquad A \cup B = W, \\
&\qquad \frac{|A|}{|B|}\in \left[\frac{\eta^2}{16}, \frac{16}{\eta^2}\right], \\
&\qquad \left||E(A,B)| - p |H_d(A,B)|\right|
    > \eta pd \cdot \sqrt{|A||B|}
\end{aligned}
\right].
\end{equation}

In order to union bound over all valid choices of $W,A,B$, we first bound the number of such sets $W$.
\begin{claim}\label{claim:connected-set-count}
   For every positive integer $t\leq n$, the number of sets $W \subseteq V_d$ such that $Q_d[W]$ is connected with $|W|=t$ is at most $n\binom{dt}{t} \le n(\e d)^t$.
\end{claim}
\begin{proof} We prove the claim by counting the number of trees on $t$ vertices in $Q_d$.  For each vertex $x \in V_d$ and $i \in [d]$, let us denote by  $N_i(x) \in V_d$ the neighbor of $x$  in $Q_d$ with the $i$-th bit flipped.
    Let us now fix a root vertex $x_0 \in V_d$, and consider any $x_0$-rooted tree $(W,F)$, where $F \subseteq H_d$, $W \subseteq V_d$ with $x_0 \in W$ and $|W|=t$.
    We label every vertex $z$ in $W$ with a tuple of at most $t$ entries from $[d]$, as follows:
    consider the unique path from $x_0$ to $z$;
    denote this path $(\pi_0, \pi_1,\dots,\pi_{\ell(z)})$ where $\ell(z)$ is the distance of $z$ from $x_0$ in the tree $(W,F)$, $\pi_0=x_0$, and $\pi_{\ell(z)} =z$;
    we label $z$ with $\text{label}(z)=(i_1,\dots,i_{\ell(z)})$ where, for each $j \in [\ell(z)]$, the entry $i_j\in [d]$ is such that $\pi_j = N_{i_j}(\pi_{j-1})$.

    Now, we order the vertices in $W$ as follows: for two distinct vertices $z,z'$, we write $z < z'$ if and only if $\ell(z) < \ell(z')$ or $\ell(z) = \ell(z')$ and $\text{label}(z)$ comes before $\text{label}(z')$ lexicographically. Since no two vertices can have the same label, this gives a total order. We then write the elements of $W$ as $w_1,\dots, w_t$.
    Next, we define a matrix $a \in \{0,1\}^{t \times d}$ as
    \begin{equation*}
        a_{j,i} = \begin{cases}
            1 & \text{if } \{w_j, N_i(w_j)\} \in F \text{ and } \text{label}(N_i(w_j)) > \text{label}(w_j)\,, \\
            0 & \text{otherwise}\,.
        \end{cases}
    \end{equation*}
    Note that $a$ has $t-1$ non-zero entries. We have thus given an encoding of $x_0$-rooted trees with $t$ vertices into $t \times d$ boolean matrices with $t-1$ non-zero entries. One can check that this encoding is injective. Hence, for a fixed $x_0$, we have at most $\binom{dt}{t}$ choices for where to put the non-zero entries in the matrix. The claim then follows, since there are $n$ choices for the vertex $x_0$.
\end{proof}

For each set $W$, there are at most $2^{|W|}$ ways to partition $W$ into subsets $A,B$. Combining with \Cref{claim:connected-set-count}, we get that for each fixed $t$, there are at most $n (2 \e d)^t $ valid choices for the sets $A,B$ and $W$ with $|W| =t$.

In combination with \Cref{claim:balanced-set-discrepancy-tail}, we can finally upper-bound the probability in \Cref{eq:large-connected-set-failure-probability} by
\begin{align*}
    n \sum_{t=d}^n  (2 \e d)^t \frac{1}{d^{(C_1+14)t}} \le n \sum_{t=d}^n \frac{1}{d^{(C_1+12)t}} \le n^2 \frac{1}{d^{(C_1+12)d}} \le \frac{1}{d^{C_1+10}} \, .
\end{align*}

\end{proof}

We now show that after projecting onto the vertices in $U$, the Laplacians $L'$ and $pL$ are close in spectral norm. The proof uses the following lemma.

\begin{lemma}[Converse expander mixing lemma, Lemma 3.6 in \cite{BL04}]
\label{lem:converse-expander-mixing}
There exists a constant $C_\bl \ge 1$ such that the following holds. Let $B \ge \alpha >0$ and let $M \in \R^{N \times N}$ be a symmetric matrix such that $M_{ii} =0$ and $\sum_{j=1}^N |M_{ij}| \le B$ for all $i \in [N]$, and for every disjoint $S,T \subseteq [N]$ one has $|\langle \ind_S, M \ind_T\rangle| \le \alpha \cdot \|\ind_S\|_2 \|\ind_T\|_2$. Then
\begin{equation*}
    \|M\| \le C_\bl \alpha \cdot \left(1+\log\frac{B}{\alpha}\right) \, .
\end{equation*}
\end{lemma}

\begin{lemma}
\label{lem:good-degree-spectral-concentration}
    With probability at least $1-1/d^{C_1+10}$ over the sampling of $Q_d'$, one has
    \begin{equation*}
        \|\Pi_{U} (L'-pL) \Pi_{U} \| \le 2\gamma pd \, .
    \end{equation*}
\end{lemma}
\begin{proof}
    By \Cref{lem:good-degree-cut-concentration}, we have $\left|\langle \ind_S,  L' \ind _T \rangle - p\langle \ind_S,  L \ind_T \rangle\right| \le \eta pd \cdot \|\ind_S\|_2 \|\ind_T\|_2$ with probability at least $1-d^{-(C_1+10)}$ simultaneously for every disjoint $S,T \subseteq U$. Let $A'$ and $A$ be the adjacency matrices of $(V_d,E(U))$ and $(V_d,H_d(U))$, respectively. Note that $A'$ and $p A$ have an all-zero diagonal, and their rows have $\ell_1$-norm at most $2pd$. Then, by \Cref{lem:converse-expander-mixing} applied to the matrix $A'-p A$, we have
    \begin{equation*}
        \|A'-p A\| \le C_\bl \eta pd \left(1+\log \frac{4pd}{\eta pd} \right) =  C_\bl \eta pd \left(1+\log \frac{4}{\eta } \right)  \, ,
    \end{equation*}

Using the equation $x \left(1  + \log \frac{4}{x}\right) \leq \sqrt{x} $ for all $x \in (0,0.01)$, and the setting $\eta = \gamma^2/C^2_\bl < 0.01$, we get
    \begin{equation*}
   \|A'-p A\|  \le C_\bl \sqrt{\eta}  pd =  \gamma pd\,.
    \end{equation*}
    Now we rewrite
    \begin{equation*}
        \Pi_{U} (L' -pL)\Pi_{U}  = \diag((\deg_{Q_d'}(u))_{u \in U})- pd\cdot  \diag(\ind_U)-(A'-pA) \, .
    \end{equation*}
    By definition of $U$,
      \begin{equation*}
        \|  \diag((\deg_{Q_d'}(u))_{u \in U})- pd\cdot  \diag(\ind_U)\| \le \gamma pd \, ,
    \end{equation*}
so we conclude
    \begin{equation*}
        \|\Pi_{U} (L'-pL) \Pi_{U} \| \le \|  \diag((\deg_{Q_d'}(u))_{u \in U})- pd\cdot  \diag(\ind_U)\|  + \| (A'-pA) \| \le 2\gamma pd \,.
    \end{equation*}
\end{proof}

We will now split the analysis into the high Fourier degree part and the low degree part. To this end, let
$$ r= \left \lfloor \frac{d}{100} \right\rfloor, $$
and  $k = \sum_{t = 0}^r \binom{d}{t}$ be the dimension of $\spn\{ \chi_S : |S| \leq r\}$.
\begin{remark}\label{rem:low-degree-subspace-dimension}
    Note that $k \ll n$. Indeed, we have
\begin{equation*}
\begin{aligned}
    k & \le (r+1)\binom{d}{r}
       \leq (r+1) \left(  \frac{\e d}{r}\right)^{r}
       \leq (r+1) \left(  \frac{\e d}{r}\right)^{d/100}
       \leq (r+1) \left(  \frac{100 \e \left((d/100-1)+1\right)}{d/100-1}\right)^{d/100}  \\
     & = (r+1) \left(100 \e \left(1+ \frac{1}{d/100-1}\right)\right)^{d/100}
      \leq (r+1)\left(200 \e\right)^{d/100}
     \le 2^{d/10} \\
     &= n^{1/10}\,.
\end{aligned}
\end{equation*}
\end{remark}

We start by showing that on the low Fourier degree subspace, $L'$ is well-approximated by $pL$. Since this space has small dimension, we can prove this directly using matrix concentration.
\begin{lemma}
    \label{lem:low-degree-spectral-approx}
 With probability at least $1-1/d^{C_1+10}$ over the sampling of $Q_d'$, one has
    \begin{equation*}
        - \gamma p P_{\le r} L P_{\le r} \preceq P_{\le r}(L' - pL)P_{\le r} \preceq \gamma p P_{\le r} L P_{\le r} \, .
    \end{equation*}
\end{lemma}
\begin{proof}
We prove the lemma using Matrix Bernstein’s inequality (\Cref{lem:matrix-bernstein}). We start by rewriting the condition from the lemma conclusion as a spectral concentration statement for a sum of independent random matrices.
The condition
    \begin{equation*}
    	  -\gamma p \cdot P_{\le r} L P_{\le r} \preceq P_{\le r}(L' - pL)P_{\le r} \preceq \gamma p \cdot P_{\le r} L P_{\le r}
    \end{equation*}
    is equivalent to
    \begin{equation*}
        \forall \, x \in \R^{V_d}, \quad \left| \langle x, P_{\le r} L' P_{\le r}x\rangle  -  p\langle x, P_{\le r} L P_{\le r}x\rangle\right| \le \gamma p \cdot \langle x, P_{\le r} L P_{\le r} x \rangle\, .
    \end{equation*}
    Let $M=P_{\le r} L P_{\le r}$. We denote by $M^\dagger$ its pseudoinverse, and denote by $M^{\dagger/2}$ the square root of $M^\dagger$. Also let $P = P_{\le r}P_{\ge 1}$. Since $\langle \chi_{\emptyset}, P_{\le r} L' P_{\le r}\chi_{\emptyset} \rangle =   p\langle \chi_{\emptyset}, P_{\le r} L P_{\le r}\chi_{\emptyset}\rangle =  \gamma p \cdot \langle \chi_{\emptyset}, P_{\le r} L P_{\le r} \chi_{\emptyset} \rangle = 0$, it is enough to prove the above inequality for $x \in \image(P_{\geq 1})$.

    Applying the change of variables $x \mapsto M^{\dagger/2}x$, we can then rewrite the condition as
    \begin{equation*}
        \forall \, x \in \R^{V_d}, \quad \left| \langle x, M^{\dagger/2} P_{\le r} \left(L'- p L\right) P_{\le r} M^{\dagger/2} x\rangle \right| \le \gamma p \cdot \langle x, P x \rangle\, .
    \end{equation*}
    For each $e \in H_d$, let for convenience $a_e = \ind(e \in E)$ be the random variable indicating whether the edge $e$ has been sampled into $G=(V_d,E)$ or not, and let $A_e = \left(a_e - p\right) \cdot M^{\dagger/2} P_{\le r} L_e P_{\le r} M^{\dagger/2}$, where $L_e$ denotes the Laplacian of $e$. Then, the above bound is equivalent to
    \begin{equation*}
    	\left\|\sum_{e \in H_d}A_e  \right\| \le \gamma p \, .
    \end{equation*}
   To apply Matrix Bernstein's inequality to the matrices $A_e$, we first bound $\|A_e\|$ and the matrix second moment $\E[A_e^2]$ for each $e \in H_d$. The following claim achieves this.
    \begin{claim}
        \label{claim:projected-edge-matrix-moments}
        For each $e \in H_d$, one has $\|A_e\| \le\frac{2k}{n}$ and $\E [A_e^2] \preceq p(1-p)\frac{2k}{n} \cdot M^{\dagger/2} P_{\le r} L_e P_{\le r} M^{\dagger/2}$.
    \end{claim}
    \begin{proof}
Fix an arbitrary orientation of \(e = (u,v)\), and let \(b_e\in \mathbb R^{V_d}\) be the vector with \(b_e(u)=1\),  \(b_e(v)=-1\), and $0$ everywhere else. Then
   \begin{align*}
   \|A_e\| & \le \| M^{\dagger/2} P_{\le r} L_e P_{\le r} M^{\dagger/2}\| && \text{since $|a_e - p| \leq 1$} \\
& = 2^d \langle b_e , (P_{\le r} M^{\dagger} P_{\le r}) b_e\rangle && \text{since $L_e = b_e^{\otimes 2}$}  \\
& \leq  \frac{1}{2} 2^d \sum_{S \subseteq [d]:  1 \le |S| \le r } \langle \chi_S, b_e\rangle^2  && \text{since $\| M^{\dagger}\| = \frac{1}{2}$. }
   \end{align*}

Since $|(\chi_S^\top b_e)| \leq 2 $ for every $ e \in H_d, S \subseteq [d]$ we get $\langle \chi_S, b_e\rangle^2 \leq4\cdot 2^{-2d}$.  Since there are $\sum_{t =1}^r{\binom{d}{t}} \leq k$ terms in the above sum, we get
       \begin{equation*}
            \|A_e\| \le 2^{d-1} \sum_{S \subseteq [d]:  1 \le |S| \le r } \langle \chi_S, b_e\rangle^2 \le \frac{2k}{n} \, .
        \end{equation*}

        Similarly,  we also get that
        \begin{equation*}
            \E [A_e^2] = \E[(a_e - p)^2](M^{\dagger/2} P_{\le r} L_e P_{\le r} M^{\dagger/2})^2 \preceq p(1-p) \cdot \frac{2k}{n} \cdot M^{\dagger/2} P_{\le r} L_e P_{\le r} M^{\dagger/2} \, .
        \end{equation*}
    \end{proof}
    Since the matrices $\E[A_e^2]$ are positive semidefinite, \Cref{claim:projected-edge-matrix-moments} gives
    \begin{align*}
        \left\| \sum_{e \in H_d} \E[A_e^2]\right\| & \le p(1-p) \cdot \frac{2k}{n} \cdot \left\| \sum_{e \in H_d} M^{\dagger/2} P_{\le r} L_e P_{\le r} M^{\dagger/2}\right\| \\
        & = p(1-p) \cdot \frac{2k}{n} \cdot \left\| M^{\dagger/2} P_{\le r} L P_{\le r} M^{\dagger/2}\right\| \\
        & =  p(1-p) \cdot \frac{2k}{n} \cdot \left\| P \right\| \\
        & = p(1-p) \cdot \frac{2k}{n}\,.
    \end{align*}
    Now, we finally apply \Cref{lem:matrix-bernstein} to the matrices $(A_e)_{e \in H_d}$ with $K=2k/n$ and $t= \gamma p$, and get
    \begin{align*}
        \Pr\left[\|\sum_{e \in H_d} A_e\| > \gamma p\right] & \le 2n \exp\left(-\frac{\gamma^2 p^2/2}{p(1-p) \cdot \frac{2k}{n}+2k/(3n)}\right) \\
        & \le 2n \exp\left(-\frac{\gamma^2p^2 n}{8k}\right) \\
        & \le 2n \exp(-\gamma^2 p^2 n^{0.9}/8) \,,
    \end{align*}
    where the last inequality used that $k \leq n^{1/10}$, by \Cref{rem:low-degree-subspace-dimension}.
    The lemma now follows since
    \begin{align*}
          2n \exp(-\gamma^2 p^2 n^{0.9}/8)  & =  2^{d+1} \exp(-\gamma^2 p^2 2^{0.9d}/8)  \leq 2^{-d} \le 2^{-(C_1+10) \log_2 d}
    \end{align*}
        for sufficiently large $d$.

\end{proof}
Next, we handle the high Fourier degree subspace.
\begin{lemma}
    \label{lem:high-degree-spectral-lb}
    With probability at least $1-1/d^{C_1+10}$ over the sampling of $Q_d'$, for both signs $\sigma \in \{\pm 1\}$ one has
    \begin{equation*}
        P_{>r} \left(\sigma\cdot \left(L'-pL\right) + \epsilon p L + d^2\cdot\Pi_R\right) P_{>r} \succeq \gamma pd P_{> r} \, .
    \end{equation*}
\end{lemma}
\begin{proof}
By \Cref{lem:good-degree-spectral-concentration}, we have that with probability at least $1-1/d^{C_1+10}$ one has
\begin{equation*}
        \|\Pi_{U} (L'-pL) \Pi_{U} \| \le 2\gamma pd \, .
\end{equation*}
Hereafter fix $f \in \image(P_{> r})$ and $\sigma \in \{\pm 1\}$. We want to show that the above spectral bound implies
    \begin{equation*}
        \langle f, \left(\sigma\cdot \left(L'-pL\right) + \epsilon p L + d^2\cdot\Pi_R\right) f\rangle \ge \gamma pd \|f\|_2^2 \, .
    \end{equation*}

Let $f_R = \Pi_R f$ and $f_U = \Pi_U f$. Then
\begin{align*}
    \langle f, \left(\sigma\cdot \left(L'-pL\right) + \epsilon p L + d^2\cdot\Pi_R\right) f\rangle =
    & \hphantom{+} \langle f_R, \sigma  \left(L'-pL\right)  f_R \rangle  \\
    & + 2 \, \langle f_R, \sigma\left(L'-pL\right)  f_U \rangle \\
        & + \langle f_U, \sigma \left(L'-pL\right) f_U \rangle  \\
    & + d^2 \| f_R\|_2^2 + \langle f, \epsilon p L f\rangle  \,.
\end{align*}
Note that the matrix $\sigma \left(L'-pL\right) $ has spectral norm at most $4d$, since $L' \preceq L \preceq 2dI$. We can then bound the first two terms by Cauchy-Schwarz, obtaining
\begin{gather*}
     \langle f, \left(\sigma\cdot \left(L'-pL\right) + \epsilon p L + d^2\Pi_R\right) f\rangle \\
     \ge  (-4d+d^2) \|f_R\|_2^2 - 8d\|f_R\|_2 \|f_U\|_2 + \langle f_U, \sigma \left(L'-pL\right) f_U \rangle  + \langle f, \epsilon pL f\rangle \\
     \geq   \frac{d^2}{2} \|f_R\|_2^2 - 8d\|f_R\|_2 \|f_U\|_2+ \langle f_U, \sigma \left(L'-pL\right) f_U \rangle  + \langle f, \epsilon pL f\rangle \, .
\end{gather*}
If $\|\Pi_{U} (L'-pL) \Pi_{U} \| \le 2\gamma pd$, then we can lower-bound the third term as follows:
\begin{equation*}
    \langle f_U, \sigma \left(L'-pL\right)   f_U \rangle \ge - 2\gamma pd \|f_U\|_2^2 \, .
\end{equation*}
Since $f \in \image(P_{> r})$,  we also know that $\langle f, L f \rangle \ge   2r \|f\|_2^2 \ge \frac{d}{100} \|f\|_2^2$. To conclude the proof, it is enough to show that
\begin{equation}\label{eq:high-degree-positivity}
    \epsilon \frac{p d}{100} \|f\|_2^2 - 2\gamma pd \|f_U\|_2^2 +\frac{d^2}{2}\|f_R\|_2^2 - 8d\|f_R\|_2 \|f_U\|_2 \ge \gamma pd \|f\|_2^2 \, .
\end{equation}
To do this, we use that $\frac{d^2}{2}(\|f_R\|_2 - \frac{8}{d}\|f_U\|_2)^2 \geq 0$, and thus
\begin{equation*}
    8d \|f_R\|_2 \|f_U\|_2 \le  \frac{d^2}{2} \|f_R\|_2^2 + 32 \|f_U\|^2_2  \leq  \frac{d^2}{2} \|f_R\|_2^2 + \gamma pd \|f_U\|^2_2\, .
\end{equation*}
Here, we used that $C \geq 32/\gamma$ by the parameter assumption, and thus $32 \leq C\gamma \leq \gamma pd$, by the setting of $p \ge C \log d/d$. Thus, recalling the parameter setting $\gamma =\frac{\epsilon}{1000}$, we get that
\begin{gather*}
     \epsilon \frac{p d}{100} \|f\|_2^2 - 2\gamma pd \|f_U\|_2^2 +\frac{d^2}{2}\|f_R\|_2^2 - 8d\|f_R\|_2 \|f_U\|_2 - \gamma pd \|f\|^2_2 \\
     \geq    \epsilon \frac{p d}{100} \|f\|_2^2  -  3\gamma pd \|f_U\|_2^2 -  \gamma pd \|f\|_2^2   \\
  \geq   pd\|f\|^2_2 \left( \frac{\epsilon}{100}- 4 \gamma \right) \\
  \geq   0,
\end{gather*}
which implies \Cref{eq:high-degree-positivity}, as required.

\end{proof}

The final proof also needs to handle the cross terms between low and high Fourier degrees. We will need the following lemma.
\begin{lemma}
    \label{lem:low-degree-squared-error}
    With probability at least $1-1/d^{C_1+10}$ over $G \sim Q_d(p)$, one has
    \begin{equation*}
        P_{\le r} (L'-pL)^2 P_{\le r} \preceq 4 p \cdot P_{\le r} L P_{\le r} \, .
    \end{equation*}
\end{lemma}
\begin{proof}
We prove the lemma using Matrix Bernstein’s inequality. We start by rewriting the condition from the lemma conclusion as a spectral concentration statement for a sum of independent random matrices.

We denote $M=P_{\le r} L P_{\le r}$ for convenience. For each $x \in V_d$, let
    \begin{align*}
        a_x & = M^{\dagger/2} P_{\le r} (L'-pL) \ind_x = M^{\dagger/2} P_{\le r} \sum_{i =1}^d \left(\ind(\{x,x^{\sim i}\} \in E) - p\right)\left(\ind_x-\ind_{x^{\sim i}}\right)\, \\
        & = \frac{1}{n}\sum_{1 \leq |S| \leq r}\frac{1}{\sqrt{2|S|}} \chi_S^{\otimes 2} \sum_{i =1}^d \left(\ind(\{x,x^{\sim i}\} \in E) - p\right)\left(\ind_x-\ind_{x^{\sim i}}\right)\,.
    \end{align*}
    The condition $P_{\le r} (L'-pL)^2 P_{\le r} \preceq 4 p \cdot P_{\le r} L P_{\le r}$ is equivalent to
    \begin{equation}\label{eq:projected-error-norm-bound}
        \left\|\sum_{x \in \{\pm 1\}^d} a_x^{\otimes 2}\right\| \le 4 p \, .
    \end{equation}
    For each $\sigma \in \{\pm 1\}$, we let $V_{\sigma} = \{x \in \{\pm 1\}^d: \, \chi_{[d]}(x)=\sigma\}$. We split the above sum according to the bipartition $V_d=V_+\cup V_-$ so that, within each $V_\sigma$, the matrices $(a_x^{\otimes 2})_{x\in V_\sigma}$ depend on disjoint edge indicators and are therefore independent.

    We now fix $\sigma = \pm 1$, and observe that the random matrices $(a_x^{\otimes 2})_{x \in V_\sigma}$ are indeed independent. This is because for every $x \in V_d$ and every $i \in [d]$, exactly one of $x$ and $x^{\sim i}$ belongs to $V_{\sigma}$.   For $x \in V_{\sigma}$, let $A_x = a_x^{\otimes 2} - \E a_x^{\otimes 2}$.  We will apply Matrix Bernstein's inequality (\Cref{lem:matrix-bernstein}) to the matrices $(A_x)_{x \in V_{\sigma}}$.

    We start by bounding  $\|a_x^{\otimes 2} \|$ for $x \in V_{\sigma}$.
    For every $x \in V_\sigma$, we have
    \begin{equation*}
        \left\|a_x^{\otimes 2}\right\| = a_x^\top a_x =  n  \sum_{S \subseteq [d]: 1\le |S| \le r} \frac{1}{2|S|}\left(\sum_{i=1}^d \left(\ind(\{x,x^{\sim i}\} \in E) - p\right)\left\langle \ind_x-\ind_{x^{\sim i}},\chi_S\right\rangle\right)^2 \,.
    \end{equation*}

For every $S$, we have that $\langle \ind_x-\ind_{x^{\sim i}},\chi_S \rangle$ is equal to $0$ if $i \notin S$, and equal to $\pm 2/n$ if $i \in S$. Hence, the inner sum $\sum_{i=1}^d \left(\ind(\{x,x^{\sim i}\} \in E) - p\right)\left\langle \ind_x-\ind_{x^{\sim i}},\chi_S\right\rangle$ has absolute value at most $2|S|/n$ for every $S$. Therefore, we get \
\begin{equation*}
        \left\|a_x^{\otimes 2}\right\| \leq n  \sum_{1 \leq |S| \leq r}\frac{4|S|^2}{2|S|n^2} \leq \frac{2rk}{n} < \frac{dk}{n}\,.
\end{equation*}
Since $a_x^{\otimes 2}$ is always PSD, we also have $0 \preceq \E [a_x^{\otimes 2}] \preceq \frac{dk}{n}I$. Thus,
\begin{equation*}
    \|A_x\| \le \frac{dk}{n}\,.
\end{equation*}

Next, we bound the spectral norm of $\sum_{x \in V_\sigma}\E[A_x^{2}]$.
We have $ \left(a_x^{\otimes 2}\right)^2 = \| a_x^{\otimes 2}\| a_x ^{\otimes 2}$, so
    \begin{equation*}
         \sum_{x \in V_\sigma}\E\left[\left(A_x\right)^2\right] = \sum_{x \in V_\sigma}\left(\E\left[\left(a_x^{\otimes 2}\right)^2\right]- \E\left[a_x^{\otimes 2}\right]^2\right) \preceq \sum_{x \in V_\sigma}\E\left[\left(a_x^{\otimes 2}\right)^2\right] \preceq \frac{dk}{n} \sum_{x \in V_\sigma}\E\left[a_x^{\otimes 2}\right] \, .
    \end{equation*}
    It then suffices to bound the spectral norm of $\sum_{x \in V_\sigma}\E[a_x^{\otimes 2}]$. To do so, let $f \in \R^{V_d}$ and $h= (M^{\dagger/2} P_{\le r})f$. Then from the definition of $a_x$, we get
    \begin{align*}
        \left\langle f, \left(\sum_{x \in V_\sigma}\E\left[a_x^{\otimes 2}\right]\right) f \right\rangle & =  n \sum_{x \in V_\sigma}\E\left[\left\langle h, (L'-pL) \ind_x \right\rangle^2\right] \, .
    \end{align*}
    Now note that for a fixed $x \in V_\sigma$ we have
    \begin{equation*}
        \left\langle h, (L'-pL) \ind_x \right\rangle = \frac{1}{n} \sum_{i =1}^d \left(\ind(\{x,x^{\sim i}\} \in E) - p\right) \left(h(x)-h(x^{\sim i})\right) \, .
    \end{equation*}
    Since these terms are independent mean-zero random variables, we have
    \begin{align*}
        \left\langle f, \left(\sum_{x \in V_\sigma}\E\left[a_x^{\otimes 2}\right]\right) f \right\rangle & \le \frac{1}{n} p  \sum_{x \in V_\sigma} \sum_{i =1}^d \left(h(x)-h(x^{\sim i})\right)^2 \\
        & = p \langle h, L h\rangle \\
        & = p \left\langle f, \left(M^{\dagger/2} P_{\le r} L P_{\le r} M^{\dagger/2}\right) f \right\rangle \\
        & = p \|P_{\geq 1} P_{\le r} f\|_2^2 \\
        & \le p \| f\|_2^2 \, .
    \end{align*}
    Thus, we have $$ \left\| \sum_{x \in V_\sigma}\E\left[A_x^2\right] \right\| \leq pdk/n $$
    Finally, by \Cref{lem:matrix-bernstein} applied with $t = p$, $K = dk/n$, we get
    \begin{align*}
        \Pr\left[\left\|\sum_{x \in V_\sigma} A_x\right\| >  p\right] \le 2n \exp\left(-\frac{p^2/2}{pdk/n+p dk/(3n)}\right) \leq 2n \exp\left(-\frac{3p n}{8dk} \right)\le \frac{1}{2d^{C_1+10}} \, ,
    \end{align*}
    where the last inequality used that $k \leq n^{1/10}$, by \Cref{rem:low-degree-subspace-dimension}.
    Taking a union bound over $\sigma = \pm 1$ gives that

    \begin{equation*}
        \left\|\sum_{x \in V_d} a_x^{\otimes 2} \right\| \le \sum_{\sigma \in \{\pm 1\}}\left\|\sum_{x \in V_\sigma} a_x^{\otimes 2} \right\| \le \sum_{\sigma \in \{\pm 1\}}\left(\left\| \sum_{x \in V_\sigma}A_x  \right\|+\left\|\sum_{x \in V_\sigma}\E\left[a_x^{\otimes 2}\right]\right\|\right)  \le 2(p+p) = 4p \, .
    \end{equation*}
\end{proof}
\begin{lemma}
    \label{lem:bad-degree-projection-bound}
    With probability at least $1-1/d^{C_1+9}$ over the sampling of $Q_d'$, one has $\|P_{\le r} \Pi_R P_{\le r}\| \le 1/d^{10}$.
\end{lemma}
\begin{proof}
 For every $S, T \subseteq [d]$, we have that
 $$\chi_S^\top \Pi_R\chi_T= \sum_{x \in R} \chi_S(x) \chi_T(x) = \sum_{x \in R} \chi_{S \triangle T}(x) = n \widehat {\ind_R}(S \triangle T). $$
 Therefore, we can write
    \begin{align*}
        P_{\le r} \Pi_R P_{\le r} &= \frac{1}{n} \sum_{S,T \subseteq [d]: \, |S|,|T| \le r}\chi_S \chi_T^\top \widehat{\ind_R}(S\triangle T) \\
        & = \sum_{Z \subseteq [d]: |Z| \le 2r}  \widehat{\ind_R}(Z) \cdot \frac{1}{n} \sum_{\substack{S,T \subseteq [d]:\\|S|,|T| \le r,\\ S \triangle T = Z}}\chi_S \chi_T^\top \, .
    \end{align*}
    For each $Z \subseteq [d]$ such that $|Z| \le 2r$, note that
    \begin{align*}
        P_{\le r} \diag(\chi_Z) P_{\le r} & = \frac{1}{n^2} \sum_{S,T \subseteq [d]: \, |S|,|T| \le r} \left(\sum_{x \in V_d}\chi_S(x) \chi_T(x) \chi_{Z}(x)\right)\chi_S\chi_T^\top\\
        & = \frac{1}{n} \sum_{S,T \subseteq [d]: \, |S|,|T| \le r} \langle \chi_{S \triangle T}, \chi_Z \rangle \chi_S \chi_T^\top \\
        & = \frac{1}{n} \sum_{\substack{S,T \subseteq [d]:\\|S|,|T| \le r,\\ S \triangle T = Z}} \chi_S \chi_T^\top \, .
    \end{align*}
    Therefore, we have
    \begin{align*}
        P_{\le r} \Pi_R P_{\le r} & = \sum_{Z \subseteq [d]: |Z| \le 2r}  \widehat{\ind_R}(Z) \cdot P_{\le r} \diag(\chi_Z) P_{\le r} \, .
    \end{align*}
    Since $\|P_{\le r} \diag(\chi_Z) P_{\le r}\| \le 1$ for every summand, it is enough to bound $|\widehat{\ind_R}(Z)|$.

First, observe that for $Z=\emptyset$ we have $\widehat{\ind_R}(Z)=|R|/n$.

   To bound $|\widehat{\ind_R}(Z)|$ for $1 \leq |Z| \leq 2r$,  we will apply Bernstein's inequality with bounded dependence. We start by rewriting $|\widehat{\ind_R}(Z)|$ as a sum of random variables satisfying the required dependence condition in \Cref{lem:bernstein-bounded-dependence}.

    Let \(q\in[0,1]\) denote the common probability that a vertex belongs to $R$. That is, for every
\(x\in\{\pm1\}^d\), $q =  \Pr\bigl(|\deg_{Q_d'}(x)-pd|>\gamma pd\bigr)
   = \Pr(x\in R).$

    Fix $Z \subseteq [d]$ such that $1 \le |Z| \le 2r$. We have
    \begin{equation*}
        \widehat{\ind_R}(Z) = \frac{1}{n}\sum_{x \in \{\pm 1\}^d} \ind_R(x) \chi_Z(x) = \frac{1}{n}\sum_{x \in \{\pm 1\}^d} (\ind_R(x)-q) \chi_Z(x) \, .
    \end{equation*}
    Each summand is bounded by $1$ in absolute value. Additionally, each of the random variables $(\ind_R(x)-q) \chi_Z(x)$  is dependent only on the random variables corresponding to neighbors of \(x\) in \(Q_d\),  so the dependency graph of the random variables $((\ind_R(x)-q) \chi_Z(x))_{x \in \{\pm 1\}^d}$ has maximum degree $d$.

    Since $\ind_R(x) \sim \mathrm{Bernoulli}(q)$ for every $x \in \{\pm 1\}^d$, the sum of second moments is
    \begin{equation*}
        \sum_{x \in \{\pm 1\}^d} \E\left[\left((\ind_R(x)-q) \chi_Z(x)\right)^2\right] \le q n \, .
    \end{equation*}
    So by \Cref{lem:bernstein-bounded-dependence} applied with $\Delta  = d$, $t = n^{3/4}$, and $K = 1$,  we have
    \begin{equation*}
        \Pr\left[\left|\widehat{\ind_R}(Z)\right| > \frac{1}{n^{1/4}} \right] \le 2 \exp\left(-\frac{n^{6/4}}{4(d+1)\cdot\left(qn+n^{3/4}\right)}\right) \le \frac{1}{n^{1/2}} \le  \frac{1}{n^{0.2}} \frac{1}{d^{C_1+15}} \, .
    \end{equation*}
    The number of sets $Z \subseteq [d]$ such that $1 \leq |Z | \leq 2r$ is
    \begin{align*}
        \sum_{t = 1}^{2r}{\binom{d}{t}} \leq 2r{\binom{d}{2r}} \leq 2r \left(  \frac{ed}{2r}\right)^{2r} \leq 2r\left(50e\right)^{d/50} \leq n^{0.2}\,.
    \end{align*}

    So by a union bound, with probability at least $1-1/d^{C_1 +15}$,  $\left|\widehat{\ind_R}(Z)\right| \leq \frac{1}{n^{1/4}}$ for all $Z \subseteq [d]$ with $1 \leq |Z| \leq 2r$.

    Combining this with \Cref{lem:size-bad-degree-set}, we have with probability at least $1-1/d^{C_1+9}$ that
    \begin{align*}
        \|P_{\le r} \Pi_R P_{\le r}\| &\le \frac{|R|}{n} + \sum_{Z \subseteq [d]: 1\le |Z| \le 2r} \left|\widehat{\ind_R}(Z)\right| \\
        & \le \frac{1}{d^{C_2}}+\frac{n^{0.2}}{n^{1/4}} \\
        & = \frac{1}{d^{C_2}}+\frac{1}{n^{1/20}}  \le \frac{1}{d^{10}} \, .
    \end{align*}
\end{proof}

\begin{proof}[Proof of \Cref{lem:spectral-concentration}]
We use $R_{Q_d'}$ as the set $R$ promised in the statement. Then, for $\sigma = \pm 1$, let
$$M_\sigma = \sigma\cdot (L'-pL)+\epsilon p L +d^2 P_{\geq 1} \Pi_R P_{\geq 1}\,.$$
We will show that with probability at least $1-1/d^{C_1}$ we have $|R| \le n/d^{C_2}$ and for both $\sigma =\pm 1$ one has $M_\sigma \succeq 0$.

By \Cref{lem:size-bad-degree-set}, \Cref{lem:low-degree-spectral-approx}, \Cref{lem:high-degree-spectral-lb}, \Cref{lem:low-degree-squared-error}, and \Cref{lem:bad-degree-projection-bound}, we have that with probability at least $1-5/d^{C_1+9}$ the following hold simultaneously:
\begin{align}
    |R| & \le \frac{n}{d^{C_2}}  && \text{by \Cref{lem:size-bad-degree-set}} \label{eq:spectral-bad-degree-set-size}, \\
     - \gamma p P_{\le r} L P_{\le r}  \preceq P_{\le r}(L'&  -  pL)P_{\le r} \preceq \gamma p P_{\le r} L P_{\le r} && \text{by \Cref{lem:low-degree-spectral-approx}} \label{eq:spectral-low-degree-approximation} , \\
      \gamma pd P_{> r}&  \preceq P_{>r} M_\sigma P_{>r}   \quad \text{ for both } \sigma=\pm 1\,  &&  \text{by \Cref{lem:high-degree-spectral-lb}}, \label{eq:spectral-high-degree-lower-bound}\\
    P_{\le r} (L'-pL)^2 P_{\le r}&  \preceq 4 p \cdot P_{\le r} L P_{\le r} && \text{by \Cref{lem:low-degree-squared-error}}, \label{eq:spectral-low-degree-squared-error}\\
    \|P_{\le r} \Pi_R P_{\le r}\|&  \le \frac{1}{d^{10}} && \text{by \Cref{lem:bad-degree-projection-bound}}\label{eq:spectral-bad-degree-projection-bound}.
\end{align}

We will now show that Equations \eqref{eq:spectral-low-degree-approximation}, \eqref{eq:spectral-high-degree-lower-bound}, \eqref{eq:spectral-low-degree-squared-error} and \eqref{eq:spectral-bad-degree-projection-bound} together imply that  $M_\sigma \succeq 0$ for $\sigma \in \{\pm1\}$. We achieve this through a Schur complement argument.

Fix $\sigma =\pm 1$. For convenience, we denote
\[
\Delta = L' - p L,\qquad \Plow=P_{\le r},\qquad \Phigh=P_{>r},\qquad X = \Plow M_\sigma \Phigh  (\Phigh M_\sigma \Phigh )^\dagger \Phigh M_\sigma \Plow.
\]

Since $M_\sigma \chi_\emptyset = 0$, it is enough to verify positive semidefiniteness on $\image(P_{>0})$.  Let $f \in \image(P_{>0})$. We want to prove $\langle f, M_\sigma f\rangle \ge 0$. We can write
\begin{equation}\label{eq:schur-complement-decomposition}
\begin{aligned}
    \langle f, M_\sigma f\rangle & = \langle \Plow f, M_\sigma (\Plow f)\rangle + \langle \Phigh f, M_\sigma (\Phigh f)\rangle + 2 \langle \Plow f, M_\sigma (\Phigh f)\rangle \\
    & = \langle f, (\Plow M_\sigma \Plow  - X) f\rangle + \langle f, X f\rangle + \langle \Phigh f, M_\sigma (\Phigh f)\rangle + 2 \langle \Plow f, M_\sigma (\Phigh f)\rangle\, .
\end{aligned}
\end{equation}
Since $\Phigh M_\sigma \Phigh  \succeq \gamma pd \Phigh  \succeq 0$ by \Cref{eq:spectral-high-degree-lower-bound}, the last three terms can be written as the quadratic form of a PSD matrix: if we take
\begin{equation*}
    g = \Phigh f + (\Phigh  M_\sigma \Phigh )^{\dagger}\Phigh M_\sigma \Plow  f\,,
\end{equation*}
then
\begin{align*}
    0 & \le \langle g, (\Phigh M_\sigma \Phigh ) g \rangle \\
    & = \langle \Phigh  f, (\Phigh  M_\sigma \Phigh ) \Phigh f\rangle + \langle (\Phigh  M_\sigma \Phigh )^{\dagger}\Phigh M_\sigma \Plow  f, \Phigh M_\sigma \Plow  f\rangle + 2\langle \Phigh  f, \Phigh M_\sigma \Plow  f\rangle \\
    & = \langle \Phigh f, M_\sigma (\Phigh f)\rangle + \langle f, X f\rangle + 2\langle \Phigh  f, \Phigh M_\sigma \Plow  f\rangle \, .
\end{align*}
Plugging this back into \Cref{eq:schur-complement-decomposition} gives
$$    \langle f, M_\sigma f\rangle  \geq\langle f, (\Plow M_\sigma \Plow  - X) f\rangle\,.  $$

Therefore, to conclude it suffices to show that $ \langle f, (\Plow M_\sigma \Plow  - X) f\rangle \ge 0$.

To do this, we first note that
\begin{equation}\label{eq:high-degree-pseudoinverse-bound}
    \left\|\left(\Phigh M_\sigma \Phigh \right)^\dagger \right\| \le \frac{1}{\gamma pd} \, ,
\end{equation}
since the smallest non-zero eigenvalue of $\Phigh M_\sigma \Phigh $ is at least $\gamma pd$ by \Cref{eq:spectral-high-degree-lower-bound}.

Using this, we bound the quantity $ \langle f, X f\rangle$. Observe that $\Plow L\Phigh  = 0$  and $\Plow (P_{>0} \Pi_R P_{>0})\Phigh  = P_{>0}\Plow  \Pi_R \Phigh $.
Therefore, we have $\Plow  M_{\sigma}\Phigh  = P_{>0}\Plow (\sigma \Delta + d^2 \Pi_R)\Phigh  $, so
\begin{align*}
    \langle f, X f\rangle & = \left\langle f, \left(\Plow  ( \sigma \Delta+d^2\Pi_R) \Phigh (\Phigh M_\sigma \Phigh )^\dagger \Phigh   (\sigma \Delta+d^2\Pi_R) \Plow \right) f\right\rangle \\
    & \le 2 \left\langle f, \left(\Plow  \Delta \Phigh (\Phigh M_\sigma \Phigh )^\dagger \Phigh   \Delta \Plow \right)f\right\rangle + 2d^4 \left\langle f, \left(\Plow  \Pi_R \Phigh (\Phigh M_\sigma \Phigh )^\dagger \Phigh   \Pi_R \Plow \right) f\right\rangle \\
& \le \frac{2}{\gamma pd} \left(\langle f, (\Plow  \Delta \Phigh ^2\Delta \Plow ) f\rangle + d^4\langle f, (\Plow  \Pi_R\Phigh ^2 \Pi_R \Plow ) f \rangle \right)  \quad   \text{by \Cref{eq:high-degree-pseudoinverse-bound}}\\
    & \le \frac{2}{\gamma pd} \left(\langle f, (\Plow  \Delta^2 \Plow ) f\rangle + d^4\langle f, (\Plow  \Pi_R \Plow ) f \rangle \right) \quad \qquad \qquad  \ \text{since $\|\Phigh \| \leq 1$} \\
    & \le  \frac{2}{\gamma pd} \left( 4p \cdot \langle f, (\Plow  L \Plow ) f\rangle + \frac{d^4}{d^{10}}\|\Plow  f\|_2^2 \right) \quad \qquad \qquad \text{by \Cref{eq:spectral-low-degree-squared-error} and \Cref{eq:spectral-bad-degree-projection-bound}}\,.
\end{align*}
On the other hand, by \Cref{eq:spectral-low-degree-approximation}, we have
\begin{align*}
 \langle f, \Plow M_\sigma \Plow  f\rangle  =   \langle f, \Plow \left(  \sigma (L' -pL)  + \epsilon p L + d^2 \Pi_R\right) \Plow f\rangle  \ge ( \epsilon -\gamma) p \langle f , \Plow L \Plow  f\rangle \, .
\end{align*}
Therefore
\begin{align*}
    \langle f, (\Plow M_\sigma \Plow  - X) f\rangle & \ge ( \epsilon -\gamma) p \langle f , \Plow L \Plow  f\rangle -\frac{2}{\gamma pd} \left( 4p \cdot \langle f, (\Plow  L \Plow ) f\rangle + \frac{1}{d^{6}}\|\Plow  f\|_2^2 \right) \\
    & \ge ( \epsilon -\gamma) p \langle f , \Plow L \Plow  f\rangle -\frac{2}{\gamma pd} \left( 4p \cdot \langle f, (\Plow  L \Plow ) f\rangle + \frac{1}{d^{6}}\cdot \langle f, \Plow  L \Plow  f\rangle  \right) \\
    & \ge p\left(\epsilon - \gamma -\frac{8}{p\gamma d}-\frac{1}{d^{4}}\right) \langle f, \Plow  L \Plow  f\rangle  \\
    & \ge 0 \, ,
\end{align*}
where the last inequality uses $p \ge C\log d/d$, the setting $\gamma = \epsilon/1000$, and the assumption that $d$ is sufficiently large.
\end{proof}

%% file: coordinate-recovery.tex
\section{Coordinate cut recovery with sum-of-squares}
\label{sec:coordinate-recovery}

Below we state our main result about recovering the coordinate cuts of the percolated hypercube.

\begin{theorem}[Main result]
\label{thm:coordinate-cut-recovery}
Let $n = 2^d$.
Let $G^*$ be sampled according to $Q_d(p)$, and let $G=(V,E)$ be obtained from $G^*$ by arbitrarily renaming its vertices. Let $z : V \to \Set{\pm 1}^d$ be the bijective embedding assigning each vertex its latent Boolean vector in $G^*$.
For $d$ large enough and $p \geq \frac{C \log d}{d}$ for a large enough absolute constant $C$, there exists an algorithm that takes as input $ G$, runs in time $n^{O(1)}$, and with probability $1-o_d(1)$ outputs functions $ f_1, \ldots,  f_d : V \to \Set{\pm 1}$ such that, up to a permutation of the functions, there exists some $ \sigma \in \{\pm 1\}^d$ such that for each $i \in [d]$
\[\E_{v \sim \mathrm{Unif}(V)} (f_i(v) - \sigma_i z_i(v))^2 \leq d^{-10}\,.\]
\end{theorem}

\begin{remark}[From coordinate cuts to latent vectors]
\label{rem:latent-vector-recovery}
Combining \Cref{thm:coordinate-cut-recovery} with \Cref{lem:overview-latent-vector-recovery} gives \Cref{thm:intro-main-recovery}.
Indeed, since the cuts are $\{\pm1\}$-valued, the squared-error bound in \Cref{thm:coordinate-cut-recovery} implies agreement at least $1-\frac14 d^{-10}$ with each coordinate cut, up to a permutation.
The algorithm outputs $\widehat z(v)=(f_1(v),\ldots,f_d(v))$; by \Cref{lem:overview-latent-vector-recovery}, after a common coordinate permutation and sign flips, these vectors agree exactly with the hidden vectors on all but at most a $\frac14 d^{-9}$ fraction of vertices.
This postprocessing takes $O(nd)$ time and preserves the success probability $1-o_d(1)$.
\end{remark}

The key component of the proof is \Cref{lem:sos-cut-approximation}, in which we give a degree-$4$ sum-of-squares proof that balanced cuts with small value in the subsampled hypercube also have small value in the original hypercube.

\begin{lemma}[Cut approximation]
\label[lemma]{lem:sos-cut-approximation}
In the same setting as \Cref{thm:coordinate-cut-recovery}, let $L$ be the unnormalized Laplacian of $Q_d$ and let $ L'$ be the unnormalized Laplacian of $ G^*$.
Consider the following system of polynomial inequalities in variables $\{f(x)\}_{x \in \{\pm 1\}^d}$ for any $\delta \geq 0$ (possibly depending on $d$):
\[
\cP = \begin{cases}
\frac{1}{p} \langle f,  L' f \rangle \leq 2+\delta\,, \\
\sum_{x \in \Set{\pm 1}^d} f(x) = 0\,, \\
f(x)^2 = 1\,, \forall x \in \{\pm 1\}^d\,. \\
\end{cases}
\]
For $d$ large enough and $p \geq \frac{C \log d}{d}$ for some absolute constant $C > 0$, with probability $1-o_d(1)$
\[\cP \proves{f}{4} \Set{\langle f, L f\rangle \leq 2+O(\delta) + d^{-20}}\,.\]
\end{lemma}

Note that in the original hypercube we have $\langle f, L f\rangle \geq 2$ for every balanced cut $f$, so the conclusion of~\Cref{lem:sos-cut-approximation} for small $\delta \geq 0$ is that $\langle f, L f\rangle$ is close to minimum.
Furthermore, by some standard Fourier-analytic machinery and by the FKN theorem, $\langle f, L f\rangle$ being close to minimal implies that $f$ is close to a coordinate cut.

\paragraph{Concentration of coordinate cuts}

We first state a simple fact that we will use at multiple points in the analysis: With high probability, coordinate cuts have value in the subsampled hypercube close to their expected value.

\begin{lemma}[Coordinate cuts are approximately preserved]
\label{lem:coordinate-cut-concentration}
In the same setting as \Cref{lem:sos-cut-approximation},
with probability $1-d^{-100}$ it holds simultaneously for all $i \in [d]$ that
\[\Abs{\frac{1}{p} \langle \chi_i,  L' \chi_i \rangle - \langle \chi_i, L \chi_i \rangle} \leq 2^{-d/3}\,.\]
\end{lemma}
\begin{proof}
We have that $\frac{1}{p} \langle \chi_i,  L' \chi_i \rangle \sim \frac{4}{p 2^d}\mathrm{Bin}(2^{d-1}, p)$ and $\langle \chi_i, L \chi_i \rangle = 2$.
By a Chernoff bound, for $ X \sim \mathrm{Bin}(2^{d-1}, p)$ and $t \leq p 2^{d-1}$,
\[\mathbb{P}\Paren{\Abs{ X - p 2^{d-1}} \geq t} \leq 2\exp\Paren{-\frac{\Paren{\frac{t}{p2^{d-1}}}^2 p2^{d-1}}{3}} = 2\exp\Paren{-\frac{t^2}{3p2^{d-1}}}\,.\]
Let $t = 1000000 \sqrt{p 2^{d-1} \log d} \leq p2^{d-1}$. Then
\[\mathbb{P}\Paren{\Abs{ X - p 2^{d-1}} \geq 1000000 \sqrt{p 2^{d-1} \log d}} \leq d^{-101}\,,\]
and hence also
\[\mathbb{P}\Paren{\Abs{\frac{1}{p} \langle \chi_i,  L' \chi_i \rangle - \langle \chi_i, L \chi_i \rangle} \geq \frac{4}{p2^d} t} \leq d^{-101}\,.\]
Note that $\frac{4}{p2^d} t = O(\sqrt{\frac{\log d}{p 2^{d}}}) \leq 2^{-d/3}$.
The result follows by a union bound over all $i \in [d]$.
\end{proof}

\subsection{Cut approximation (proof of \texorpdfstring{\Cref{lem:sos-cut-approximation}}{Lemma \ref*{lem:sos-cut-approximation}})}

In this section we develop the proof of \Cref{lem:sos-cut-approximation}.
First, let us relate the value of a cut, the value of a coordinate cut, and the cut obtained by the symmetric difference of the two.

\begin{lemma}
\label[lemma]{lem:cut-symmetric-difference-identity}
In the same setting as \Cref{lem:sos-cut-approximation}, for any $i \in [d]$,
consider the subgraph of $Q_d$ formed of all of its edges $\{x, y\}$ with $x(i) = y(i)$ and let $L_{(i)}$ be its unnormalized Laplacian.
Similarly, consider the subgraph of $ G^*$ formed of all of its edges $\{x, y\}$ with $x(i) = y(i)$, and let $ L'_{(i)}$ be its unnormalized Laplacian.
Thus, $L_{(i)}$ and $L'_{(i)}$ correspond to $L_{Q_d}^{(i)}$ and $L_{G^*}^{(i)}$ in the technical overview (\Cref{sec:technical-overview}).
We have
\[\cP \proves{f}{2}\langle f, L f\rangle - \langle \chi_i, L \chi_i \rangle = \langle f \circ \chi_i, (2 L_{(i)} - L) (f \circ \chi_i)\rangle\]
and
\[\cP \proves{f}{2}\langle f,  L' f\rangle - \langle \chi_i,  L' \chi_i \rangle = \langle f \circ \chi_i, (2  L'_{(i)} -  L') (f \circ \chi_i)\rangle\,.\]
\end{lemma}
\begin{proof}
For an edge $\{x, y\}$ in $Q_d$, denote by $L_{\Set{x,y}}$ the Laplacian of the edge $\{x, y\}$, i.e., $L_{\Set{x,y}}(x,x)=L_{\Set{x,y}}(y,y)=1$, $L_{\Set{x,y}}(x,y)=L_{\Set{x,y}}(y,x)=-1$, and $L_{\Set{x, y}}$ is zero everywhere else.
Then, for any edge $\{x, y\}$ in $Q_d$ we have
\begin{align*}
\cP \proves{f}{2}
\langle f, L_{\Set{x,y}} f\rangle - \langle \chi_i, L_{\Set{x,y}} \chi_i\rangle
&= \frac{1}{2^d}\Paren{- 2 f(x) f(y) + 2 \chi_i(x) \chi_i(y)}\,.
\end{align*}
In particular, if $x(i) = y(i)$, we have $\chi_i(x) \chi_i(y) = 1$, and hence
\begin{align*}
\cP \proves{f}{2} \langle f, L_{\Set{x,y}} f\rangle - \langle \chi_i, L_{\Set{x,y}} \chi_i\rangle
&= \frac{1}{2^d}\Paren{-2f(x)\chi_i(x)f(y)\chi_i(y) + 2}\\
&= \langle f \circ \chi_i, L_{\Set{x,y}} (f \circ \chi_i)\rangle\\
&= \langle f \circ \chi_i, (2 \cdot L_{\Set{x,y}} - L_{\Set{x,y}}) (f \circ \chi_i)\rangle\,.
\end{align*}
Else, if $x(i) \neq y(i)$, we have $\chi_i(x) \chi_i(y) = -1$, and hence
\begin{align*}
\cP \proves{f}{2} \langle f, L_{\Set{x,y}} f\rangle - \langle \chi_i, L_{\Set{x,y}} \chi_i\rangle
&= \frac{1}{2^d}\Paren{2f(x)\chi_i(x)f(y)\chi_i(y) - 2}\\
&= -\langle f \circ \chi_i, L_{\Set{x,y}} (f \circ \chi_i)\rangle\\
&= \langle f \circ \chi_i, (2 \cdot 0 -L_{\Set{x,y}}) (f \circ \chi_i)\rangle\,.
\end{align*}
Given these equalities, the first claim follows because
\[L = \sum_{\substack{\{x, y\} \text{ edge in } Q_d}} L_{\Set{x,y}}\,,\quad\quad L_{(i)} = \sum_{\substack{\{x, y\} \text{ edge in } Q_d\\ x(i) = y(i)}} L_{\Set{x,y}}\,.\]
The second claim follows because
\[ L' = \sum_{\substack{\{x, y\} \text{ edge in }  G^*}} L_{\Set{x,y}}\,,\quad\quad  L'_{(i)} = \sum_{\substack{\{x, y\} \text{ edge in }  G^*\\ x(i) = y(i)}} L_{\Set{x,y}}\,.\]
\end{proof}

Next, using \Cref{lem:cut-symmetric-difference-identity} and \Cref{lem:spectral-concentration}, we prove an upper bound on a cut in the original hypercube in terms of the value of the cut in the subsampled hypercube and the cut obtained by the symmetric difference between it and a coordinate cut.

\begin{lemma}
\label[lemma]{lem:cut-comparison-symmetric-difference}
Fix any $\epsilon > 0$ small enough.
In the same setting as \Cref{lem:sos-cut-approximation} and with $C$ large enough as a function of $\epsilon$, it holds with probability $1-o_d(1)$ simultaneously for all $i \in [d]$ that
\[\cP \proves{f}{2} \Set{\langle f, L f\rangle - \langle \chi_i, L \chi_i \rangle \leq \frac{1}{p} \Paren{\langle f,  L' f\rangle - \langle \chi_i,  L' \chi_i \rangle} + \epsilon \cdot \langle f \circ \chi_i, L (f \circ \chi_i) \rangle + d^{-80}}\,.\]
\end{lemma}
\begin{proof}
By \Cref{lem:cut-symmetric-difference-identity}, we have
\begin{align*}
&\cP \proves{f}{2} \langle f, Lf\rangle - \langle \chi_i, L \chi_i\rangle \\
&\quad = \langle f \circ \chi_i, (2 L_{(i)} - L) (f \circ \chi_i)\rangle\\
&\quad= \langle f \circ \chi_i, \frac{1}{p} (2  L'_{(i)} -  L') (f \circ \chi_i)\rangle + \Iprod{f \circ \chi_i, \Paren{(2 L_{(i)} - L) - \frac{1}{p} (2  L'_{(i)} -  L')} (f \circ \chi_i)}\\
&\quad= \frac{1}{p}\Paren{\langle f,  L' f\rangle - \langle \chi_i,  L' \chi_i\rangle } + \Iprod{f \circ \chi_i, \Paren{(2 L_{(i)} - L) - \frac{1}{p} (2  L'_{(i)} -  L')} (f \circ \chi_i)}\,.
\end{align*}

It remains to upper bound the error term.

\begin{claim}
    \label{claim:sos-cut-value-concentration}
    With probability $1-o_d(1)$, for all $i \in [d]$ we have
    \begin{equation*}
        \Iprod{f \circ \chi_i, \Paren{2 L_{(i)} - L}  (f \circ \chi_i)} \le \frac{1}{p}\Iprod{f \circ \chi_i, \Paren{2  L'_{(i)} -  L'}(f \circ \chi_i)}+5\epsilon \cdot \Iprod{f \circ \chi_i, L (f \circ \chi_i)} + 5 \frac{1}{d^{160}} \,.
    \end{equation*}
\end{claim}
\begin{proof}
Write
\begin{align*}
&\cP \proves{f}{2} \Iprod{f \circ \chi_i, \Paren{(2 L_{(i)} - L) - \frac{1}{p} (2  L'_{(i)} -  L')} (f \circ \chi_i)}\\
&\quad= 2 \Iprod{ f \circ \chi_i, \Paren{L_{(i)} - \frac{1}{p}  L'_{(i)}} (f \circ \chi_i)} - \Iprod{ f \circ \chi_i, \Paren{L - \frac{1}{p}  L'}(f \circ \chi_i)}\,.
\end{align*}

Let us start with the term involving $L - \frac{1}{p} L'$.
Apply \Cref{lem:spectral-concentration} with $C_1, C_2$ large enough such that $|R| \leq 2^d/d^{320}$ and such that the probability of success is at least $1-d^{-100}$.
Then we have with probability $1-d^{-100}$ that
\[L - \frac{1}{p}  L' \succeq -\epsilon \cdot L - \frac{d^2}{p} P_{>0} \Pi_R P_{>0}\,,\]
so
\begin{align*}
\cP \proves{f}{2} \Iprod{ f \circ \chi_i, \Paren{L - \frac{1}{p}  L'}(f \circ \chi_i)}
&\geq -\epsilon \cdot \Iprod{f \circ \chi_i, L (f \circ \chi_i)} - \frac{d^2}{p} \Iprod{f \circ \chi_i, P_{>0} \Pi_R P_{>0} (f \circ \chi_i)}\,.
\end{align*}
For the second term, we have
\begin{align*}
\cP \proves{f}{2}\Iprod{f \circ \chi_i, P_{>0} \Pi_R P_{>0} (f \circ \chi_i)}
&= \Iprod{P_{>0} (f \circ \chi_i), \Pi_R P_{>0} (f \circ \chi_i)}\\
&= \frac{1}{2^d} \sum_{x \in R} ((P_{>0} (f \circ \chi_i))(x))^2\\
&= \frac{1}{2^d} \sum_{x \in R} \Paren{\Paren{\Paren{I - \frac{1}{2^d} \chi_{\emptyset}^{\otimes 2}} (f \circ \chi_i)}(x)}^2\\
&= \frac{1}{2^d} \sum_{x \in R} \Paren{ (f \circ \chi_i)(x) - \E_{ x'} (f \circ \chi_i)( x') }^2\,.
\end{align*}
By \Cref{fact:sos-jensen} and the Booleanity constraints, we have
\begin{align*}
\cP \proves{f}{2} ((f \circ \chi_i)(x) - \E_{ x'} (f \circ \chi_i)( x'))^2
&\leq 2 (f \circ \chi_i)(x)^2 + 2 (\E_{ x'} (f \circ \chi_i)( x'))^2\\
&\leq 2 (f \circ \chi_i)(x)^2 + 2 \E_{x'} (f \circ \chi_i)(x')^2\\
&\leq 4\,.
\end{align*}
Hence
\begin{align*}
\cP \proves{f}{2}\Iprod{f \circ \chi_i, P_{>0} \Pi_R P_{>0} (f \circ \chi_i)}\leq \frac{1}{2^d} 4|R| \leq O(d^{-320})\,,
\end{align*}
so overall
\begin{align*}
\cP \proves{f}{2} \Iprod{ f \circ \chi_i, \Paren{L - \frac{1}{p}  L'}(f \circ \chi_i)} \geq -\epsilon \cdot \Iprod{f \circ \chi_i, L (f \circ \chi_i)} - d^{-160}\,.
\end{align*}

Let us continue with the term involving $L_{(i)} - \frac{1}{p}  L'_{(i)}$.
We note that $L_{(i)}$ corresponds to a graph formed of two connected components, each a $(d-1)$-dimensional hypercube. Then we can further write $L_{(i)} = L_{(i)}^{(1)} + L_{(i)}^{(2)}$ and $ L'_{(i)} =  L_{(i)}^{'(1)} +  L_{(i)}^{'(2)}$, corresponding to the decomposition into the two $(d-1)$-dimensional hypercubes.
We can then apply the argument that we applied above to $L - \frac{1}{p}  L'$ analogously to $L_{(i)}^{(1)} - \frac{1}{p} L_{(i)}^{'(1)}$ and $L_{(i)}^{(2)} - \frac{1}{p} L_{(i)}^{'(2)}$ for all $i \in [d]$.
Specifically, we apply \Cref{lem:spectral-concentration} to the two subcubes obtained by fixing the value of the $i$-th coordinate. Taking a union bound over the probability of success for the two subcubes, and using that $L_{(i)} \preceq L$, we get with probability $1-d^{-10}$ that
\begin{align*}
\cP \proves{f}{2} \Iprod{ f \circ \chi_i, \Paren{L_{(i)} - \frac{1}{p}  L_{(i)}'}(f \circ \chi_i)} \leq 2\epsilon \cdot \Iprod{f \circ \chi_i, L (f \circ \chi_i)} + 2d^{-160}\,.
\end{align*}
Adding the bounds for $L$ and $L_{(i)}$ and taking a union bound over $i \in [d]$ gives the claim.
\end{proof}
We finally plug in the result of \Cref{claim:sos-cut-value-concentration} into our initial rewriting, so we get
\begin{align*}
&\cP \proves{f}{2} \langle f, Lf\rangle - \langle \chi_i, L \chi_i\rangle \\
&\quad\leq \frac{1}{p}\Paren{\langle f,  L' f\rangle - \langle \chi_i,  L' \chi_i\rangle } + O(\epsilon) \cdot \Iprod{f \circ \chi_i, L (f \circ \chi_i)} + O(d^{-160})\,.
\end{align*}
The final result follows, up to a reparameterization of $\epsilon$.
\end{proof}

Let us bound the term in the symmetric difference of $f$ and $\chi_i$ that appears in \Cref{lem:cut-comparison-symmetric-difference}.

\begin{lemma}
\label[lemma]{lem:symmetric-difference-energy-bound}
In the same setting as \Cref{lem:sos-cut-approximation}, it holds for all $i \in [d]$ that
\[\cP \proves{f}{2} \Set{\langle f \circ \chi_i, L (f \circ \chi_i) \rangle \leq (\langle f, Lf\rangle - \langle \chi_i, L \chi_i\rangle) + 4(1 - \langle f, \chi_i\rangle^2)}\,.\]
\end{lemma}
\begin{proof}
We have
\begin{align*}
\cP \proves{f}{2}
\langle f \circ \chi_i, L (f \circ \chi_i) \rangle
&= \langle f \circ \chi_i, (2 L_{(i)} - L) (f \circ \chi_i) \rangle + 2 \langle f \circ \chi_i, (L - L_{(i)}) (f \circ \chi_i) \rangle\\
&\stackrel{(*)}{=} (\langle f, Lf\rangle - \langle \chi_i, L \chi_i\rangle) + 2 \langle f \circ \chi_i, (L - L_{(i)}) (f \circ \chi_i) \rangle\,,
\end{align*}
where (*) follows by \Cref{lem:cut-symmetric-difference-identity}.

Let us upper bound the second term.
For any $x \in \Set{\pm 1}^d$, denote by $x^{\sim i} \in \Set{\pm 1}^d$ the vector obtained from $x$ by flipping the $i$-th coordinate.
Then, noting that the only edges corresponding to $L - L_{(i)}$ are those in which the $i$-th coordinate flips between the two endpoints, we have
\begin{align*}
\cP \proves{f}{2} \langle f \circ \chi_i, (L - L_{(i)}) (f \circ \chi_i) \rangle
&= \frac{1}{2^{d+1}} \sum_{x \in \Set{\pm 1}^d} ((f\circ \chi_i)(x) - (f \circ \chi_i)(x^{\sim i}))^2\,.
\end{align*}
We will use now the fact that
\[\cP \proves{f}{2} \sum_{x \in \Set{\pm 1}^d} ((f\circ \chi_i)(x) - (f \circ \chi_i)(x^{\sim i}))^2 + \sum_{x \in \Set{\pm 1}^d} ((f\circ \chi_i)(x) + (f \circ \chi_i)(x^{\sim i}))^2 = 2^{d+2}\]
to rewrite the equality above as
\begin{align*}
\cP \proves{f}{2} \langle f \circ \chi_i, (L - L_{(i)}) (f \circ \chi_i) \rangle
&= \frac{1}{2^{d+1}} \Paren{2^{d+2} - \sum_{x \in \Set{\pm 1}^d} ((f\circ \chi_i)(x) + (f \circ \chi_i)(x^{\sim i}))^2}\,.
\end{align*}
Using a sum-of-squares version of Jensen's inequality (\Cref{fact:sos-jensen}),
\begin{align*}
\cP \proves{f}{2} \langle f \circ \chi_i, (L - L_{(i)}) (f \circ \chi_i) \rangle
&\leq \frac{1}{2^{d+1}} \Paren{2^{d+2} - \frac{1}{2^d}\Paren{\sum_{x \in \Set{\pm 1}^d} ((f\circ \chi_i)(x) + (f \circ \chi_i)(x^{\sim i}))}^2}\\
&= \frac{1}{2^{d+1}} \Paren{2^{d+2} - \frac{1}{2^{d-2}} \Paren{\sum_{x \in \Set{\pm 1}^d}(f\circ \chi_i)(x) }^2}\\
&= 2 \Paren{1 - \Paren{\frac{1}{2^d} \sum_{x \in \Set{\pm 1}^d} (f\circ \chi_i)(x)}^2}\\
&= 2 \Paren{1 - \langle f, \chi_i\rangle^2}\,.
\end{align*}
Plugging this error bound into our main derivation, we get the desired result.
\end{proof}

At this point, let us put together \Cref{lem:cut-comparison-symmetric-difference}, \Cref{lem:symmetric-difference-energy-bound}, and also apply \Cref{lem:coordinate-cut-concentration}.

\begin{lemma}
\label[lemma]{lem:cut-comparison-coordinate-correlation}
Fix any $\epsilon > 0$ small enough.
In the same setting as \Cref{lem:sos-cut-approximation} and with $C$ large enough as a function of $\epsilon$, it holds with probability $1-o_d(1)$ simultaneously for all $i \in [d]$ that
\[\cP \proves{f}{2} \Set{(1-\epsilon) \langle f, Lf\rangle \leq \frac{1}{p} \langle f,  L' f\rangle - 2\epsilon + 4\epsilon (1 - \langle f, \chi_i\rangle^2) + d^{-40}}\,.\]
\end{lemma}
\begin{proof}
By \Cref{lem:cut-comparison-symmetric-difference}, it holds with probability $1-o_d(1)$ simultaneously for all $i \in [d]$ that
\[\cP \proves{f}{2} \langle f, L f\rangle - \langle \chi_i, L \chi_i \rangle \leq \frac{1}{p} \Paren{\langle f,  L' f\rangle - \langle \chi_i,  L' \chi_i \rangle} + \epsilon \cdot \langle f \circ \chi_i, L (f \circ \chi_i) \rangle + d^{-80}\,.\]
Plugging in \Cref{lem:symmetric-difference-energy-bound}, we further get
\[\cP \proves{f}{2} \langle f, L f\rangle - \langle \chi_i, L \chi_i \rangle \leq \frac{1}{p} \Paren{\langle f,  L' f\rangle - \langle \chi_i,  L' \chi_i \rangle} + \epsilon \cdot \Paren{(\langle f, Lf\rangle - \langle \chi_i, L \chi_i\rangle) + 4(1 - \langle f, \chi_i\rangle^2)} + d^{-80}\,,\]
so, rearranging,
\[\cP \proves{f}{2} (1-\epsilon)\Paren{\langle f, L f\rangle - \langle \chi_i, L \chi_i \rangle} \leq \frac{1}{p} \Paren{\langle f,  L' f\rangle - \langle \chi_i,  L' \chi_i \rangle} + 4\epsilon (1 - \langle f, \chi_i\rangle^2) + d^{-80}\,.\]
By \Cref{lem:coordinate-cut-concentration}, it holds with probability $1-o_d(1)$ again simultaneously for all $i \in [d]$ that $|\langle \chi_i, L \chi_i\rangle - \frac{1}{p} \langle \chi_i,  L' \chi_i\rangle| \leq 2^{-d/3}$. Also note that $\langle \chi_i, L \chi_i\rangle = 2$. Then
\[\cP \proves{f}{2} (1-\epsilon)\langle f, L f\rangle \leq \frac{1}{p} \langle f,  L' f\rangle - 2\epsilon + 4\epsilon (1 - \langle f, \chi_i\rangle^2) + d^{-40}\,.\]
\end{proof}

The result of \Cref{lem:cut-comparison-coordinate-correlation} is for a fixed coordinate cut $\chi_i$.
In the next step, we put together the results for all coordinate cuts and obtain a new result in terms of the quantity $1 - \sum_{i=1}^d \langle f, \chi_i\rangle^4$, which will later interface nicely with our sum-of-squares FKN theorem in~\Cref{thm:sos-fkn}.

\begin{lemma}
\label[lemma]{lem:cut-comparison-quartic-proxy}
Fix any $\epsilon > 0$ small enough.
In the same setting as \Cref{lem:sos-cut-approximation} and with $C$ large enough as a function of $\epsilon$, it holds with probability $1-o_d(1)$ that
\[\cP \proves{f}{4} \Set{(1-\epsilon) \langle f, Lf\rangle \leq \frac{1}{p} \langle f,  L' f\rangle - 2\epsilon + 4\epsilon \Paren{1 - \sum_{i=1}^d \langle f, \chi_i\rangle^4} + d^{-40}}\,.\]
\end{lemma}
\begin{proof}
By rearranging the conclusion of \Cref{lem:cut-comparison-coordinate-correlation}, we get with probability $1-o_d(1)$ simultaneously for all $i \in [d]$ that
\[\cP \proves{f}{2} \frac{(1-\epsilon) \langle f, Lf\rangle - \frac{1}{p} \langle f,  L' f\rangle + 2\epsilon - d^{-40}}{4\epsilon} \leq 1 - \langle f, \chi_i\rangle^2\,.\]
For each $i \in [d]$, multiply both sides of the inequality by $\langle f, \chi_i\rangle^2$. Then
\[\cP \proves{f}{4} \frac{(1-\epsilon) \langle f, Lf\rangle - \frac{1}{p} \langle f,  L' f\rangle + 2\epsilon - d^{-40}}{4\epsilon} \langle f, \chi_i\rangle^2 \leq \langle f, \chi_i\rangle^2 - \langle f, \chi_i\rangle^4\,.\]
Summing up the inequalities over $i \in [d]$,
\[\cP \proves{f}{4} \frac{(1-\epsilon) \langle f, Lf\rangle - \frac{1}{p} \langle f,  L' f\rangle + 2\epsilon - d^{-40}}{4\epsilon} \sum_{i=1}^d \langle f, \chi_i\rangle^2 \leq \sum_{i=1}^d\langle f, \chi_i\rangle^2 - \sum_{i=1}^d\langle f, \chi_i\rangle^4\,.\]
Rearranging,
\begin{align*}
\cP \proves{f}{4} \sum_{i=1}^d\langle f, \chi_i\rangle^4
&\leq \Paren{1 - \frac{(1-\epsilon) \langle f, Lf\rangle - \frac{1}{p} \langle f,  L' f\rangle + 2\epsilon - d^{-40}}{4\epsilon}} \sum_{i=1}^d \langle f, \chi_i\rangle^2\,.
\end{align*}
Recall that we have by the conclusion of \Cref{lem:cut-comparison-coordinate-correlation} a sum-of-squares proof that the term in the parentheses on the right-hand side is non-negative.
Then, using also by \Cref{claim:sos-parseval} and \Cref{claim:sos-boolean-parseval} that $\cP \proves{f}{2} \sum_{i=1}^d \langle f, \chi_i\rangle^2 \leq 1$, it follows that
\begin{align*}
\cP \proves{f}{4} \sum_{i=1}^d\langle f, \chi_i\rangle^4
&\leq 1 - \frac{(1-\epsilon) \langle f, Lf\rangle - \frac{1}{p} \langle f,  L' f\rangle + 2\epsilon - d^{-40}}{4\epsilon}\,.
\end{align*}
Rearranging gives the desired conclusion.
\end{proof}

Let us also prove in sum-of-squares the simple fact that, if $\langle f, L f\rangle$ is close to the minimum value $2$, then most of the Fourier mass of $f$ lies on level $1$.

\begin{lemma}
\label[lemma]{lem:cut-energy-fourier-concentration}
In the same setting as \Cref{lem:sos-cut-approximation}, we have
\[\cP \proves{f}{2} \sum_{i \in [d]} \langle f, \chi_i\rangle^2 \geq 1 - \frac{1}{2}(\langle f, L f\rangle-2)\,.\]
\end{lemma}
\begin{proof}
We use that the eigenvectors of $L$ are $\{\chi_S \mid S \subseteq [d]\}$, with the eigenvalue corresponding to $\chi_S$ being $2|S|$.
Hence, using by \Cref{claim:sos-fourier-expansion} that $\cP \proves{f}{O(1)} f(x) = \sum_{S \subseteq [d]} \langle f, \chi_S\rangle \chi_S(x)$,
\begin{align*}
\cP \proves{f}{2} \langle f, L f \rangle
&= 2 \sum_{i \in [d]} \langle f, \chi_i\rangle^2 + \sum_{\substack{S \subseteq [d]\\ |S| \geq 2}} 2|S| \langle f, \chi_S\rangle^2\\
&\geq 2 \sum_{i \in [d]} \langle f, \chi_i\rangle^2 + 4 \sum_{\substack{S \subseteq [d]\\ |S| \geq 2}} \langle f, \chi_S\rangle^2\\
&\stackrel{(*)}{=} 2 + 2 \sum_{\substack{S \subseteq [d]\\ |S| \geq 2}} \langle f, \chi_S\rangle^2\,,
\end{align*}
where in (*) we used \Cref{claim:sos-parseval} and the fact that $\cP \proves{f}{2} \langle f, \chi_{\emptyset}\rangle = \frac{1}{2^d} \sum_{x \in \Set{\pm 1}^d} f(x) = 0$.
Hence,
\begin{align*}
\cP \proves{ f}{2} \sum_{\substack{S \subseteq [d]\\ |S| \geq 2}} \langle f, \chi_S\rangle^2 \leq \frac{1}{2}(\langle f, Lf\rangle - 2)\,,
\end{align*}
so
\begin{align*}
\cP \proves{ f}{2} \sum_{i \in [d]} \langle f, \chi_i\rangle^2 \geq 1 - \frac{1}{2}(\langle f, Lf\rangle - 2)\,.
\end{align*}
\end{proof}

We are ready to prove \Cref{lem:sos-cut-approximation}.
\begin{proof}[Proof of \Cref{lem:sos-cut-approximation}]
We have by assumption that $\cP \proves{f}{2} \frac{1}{p}\langle f,  L' f\rangle \leq 2 + \delta$.
Substituting this into the conclusion of \Cref{lem:cut-comparison-quartic-proxy} with $\epsilon$ small enough, we get
\begin{align*}
\cP \proves{f}{4} (1-\epsilon) \langle f, Lf\rangle
&\leq 2+\delta - 2\epsilon + 4\epsilon \Paren{1 - \sum_{i=1}^d \langle f, \chi_i\rangle^4} + d^{-40}\\
&= (1-\epsilon)2 + \delta + 4\epsilon \Paren{1 - \sum_{i=1}^d \langle f, \chi_i\rangle^4} + d^{-40}\,.
\end{align*}
Rearranging and dividing by $1-\epsilon$,
\[\cP \proves{f}{4} \langle f, Lf\rangle - 2 \leq \frac{1}{1-\epsilon} \Paren{\delta + 4\epsilon \Paren{1 - \sum_{i=1}^d \langle f, \chi_i\rangle^4} + d^{-40}}\,.\]

We apply now the sum-of-squares FKN theorem in \Cref{thm:sos-fkn} with error variable ($\delta$ in the statement of \Cref{thm:sos-fkn}) set to $\frac{1}{2}(\langle f, L f\rangle - 2)$, which satisfies the needed hypotheses by \Cref{lem:cut-energy-fourier-concentration} and Parseval (\Cref{claim:sos-boolean-parseval}).
Then, applying \Cref{thm:sos-fkn}, we get
\[\cP \proves{f}{4} 1 - \sum_{i=1}^d \langle f, \chi_i\rangle^4 \leq O\Paren{\langle f, L f\rangle - 2}\,,\]
and substituting in our main derivation,
\[\cP \proves{f}{4} \langle f, Lf\rangle - 2 \leq \frac{1}{1-\epsilon} \Paren{\delta + O(\epsilon) \Paren{\langle f, L f\rangle - 2} + d^{-40}}\,,\]
so
\[\cP \proves{f}{4} \langle f, Lf\rangle - 2 \leq O(\delta) + d^{-20}\,.\]
This gives the desired result.
\end{proof}

\subsection{Proof of \texorpdfstring{\Cref{thm:coordinate-cut-recovery}}{Theorem \ref*{thm:coordinate-cut-recovery}}}

In this section we complete the proof of \Cref{thm:coordinate-cut-recovery} using \Cref{lem:sos-cut-approximation}, the sum-of-squares FKN theorem in \Cref{thm:sos-fkn}, and the rounding algorithm in Theorem 5.2 of~\cite{MR3631006-Ma16}.

\begin{proof}[Proof of \Cref{thm:coordinate-cut-recovery}]

The algorithm first computes $\hat{p} = \frac{|E(G)|}{d 2^{d-1}}$, where $E(G)$ is the set of edges of $G$.
By standard concentration bounds, we have that $|\hat{p} - p| \leq 2^{-\Omega(d)}$ with probability $1-o_d(1)$.
From now on we condition on this event.

Let $\hat{ L}'$ be the unnormalized Laplacian of $ G$.
Consider the following system of polynomial inequalities in variables $\Set{\hat{f}(v)}_{v \in V}$:
\[
\hat\cP = \begin{cases}
\frac{1}{\hat{p}} \langle \hat{f}, {\hat{ L}}' \hat f \rangle \leq 2+d^{-12}\,, \\
\sum_{v \in V} \hat f(v) = 0\,, \\
\hat f(v)^2 = 1\,, \forall v \in V\,. \\
\end{cases}
\]

The system $\hat\cP$ is feasible with high probability for each $\hat{f}(v) = \chi_i(z(v)) = z_i(v)$ for all $v \in V$, $i \in [d]$. This is because, for such a choice $\hat{f}$, \Cref{lem:coordinate-cut-concentration} gives with probability $1-o_d(1)$ that $\frac{1}{p}\langle \hat{f}, \hat{ L}' \hat{f}\rangle \leq 2 + 2^{-d/3}$, so also $\frac{1}{\hat{p}}\langle \hat{f}, \hat{ L}' \hat{f}\rangle \leq 2 + d^{-12}$.
The other two constraints of $\hat\cP$ are trivially satisfied.  We inspect now the implications of the system $\hat\cP$.

\begin{claim}
\label{claim:sos-quartic-proxy-certificate}
    Let $\hat \chi_i(v) = \chi_i(z(v)) = z_i(v)$ for all $i \in [d]$ and $v \in V$. Then, with probability $1-o_d(1)$ we have
    \begin{equation}
\hat \cP \proves{\hat f}{4} \sum_{i \in [d]} \langle \hat f, \hat \chi_i\rangle^4 \geq 1 - O\left(\frac{1}{d^{11}}\right)\,.
\end{equation}
\end{claim}
\begin{proof}
Define $\Set{f(x)}_{x \in \{\pm 1\}^d}$ such that $f(x) = \hat{f}(z^{-1}(x))$.
It is trivial that $\hat\cP \proves{\hat{f}}{2} \sum_{x \in \{\pm 1\}^d} f(x) = 0$ and $\hat\cP \proves{\hat{f}}{2} f(x)^2 = 1$ for all $x \in \Set{\pm 1}^d$.
Also let $L'$ be the unnormalized Laplacian of $G^*$.
Note that $ L'(x, y) = \hat{ L}'(z^{-1}(x), z^{-1}(y))$ for all $x, y \in \Set{\pm 1}^d$, and hence ${\hat\cP \proves{\hat{f}}{2} \frac{1}{\hat{p}}\langle f,  L' f \rangle = \frac{1}{\hat{p}} \langle \hat f, \hat{ L}' \hat f \rangle \leq 2 + d^{-12}}$.
Furthermore, because we are assuming that $|\hat{p} - p| \leq 2^{-\Omega(d)}$, it also follows easily that $\hat\cP \proves{\hat{f}}{2} \frac{1}{p} \langle f, L' f \rangle \leq 2+d^{-11}$.
Therefore $\hat{\cP}$ implies the system of inequalities
\[\cP = \begin{cases}
\frac{1}{p} \langle f,  L' f \rangle \leq 2+d^{-11}\,, \\
\sum_{x \in \Set{\pm 1}^d} f(x) = 0\,, \\
f(x)^2 = 1\,, \forall x \in \{\pm 1\}^d\,. \\
\end{cases}\]
By \Cref{lem:sos-cut-approximation}, we get with probability $1-o_d(1)$ that
\[\cP \proves{f}{4} \langle f, L f \rangle \leq 2+O(d^{-11})\,.\]

We apply now the sum-of-squares FKN theorem in \Cref{thm:sos-fkn} with error variable $\delta = \frac{1}{2}(\langle f, L f\rangle - 2)$, which satisfies the needed hypotheses by \Cref{lem:cut-energy-fourier-concentration} and Parseval's identity (\Cref{claim:sos-boolean-parseval}).
Hence,
\[\cP \proves{f}{4} \sum_{i\in[d]}\langle f,\chi_i\rangle^4 \geq 1-O(\langle f,Lf\rangle-2)\,.\]
Combining this with the upper bound on $\langle f, L f\rangle$ gives
\[\cP \proves{f}{4} \sum_{i\in[d]}\langle f,\chi_i\rangle^4 \geq 1-O(d^{-11})\,,\]
and, by reverting the change of variables, we conclude
\begin{equation*}
\hat \cP \proves{\hat f}{4} \sum_{i \in [d]} \langle \hat f, \hat \chi_i\rangle^4 \geq 1 -O(d^{-11})\,.
\end{equation*}
\end{proof}

By the conclusion of \Cref{claim:sos-quartic-proxy-certificate}, we can apply the rounding theorem \Cref{thm:sos-rounding} with variable $u = \hat{f}$, $P(u) = 2^{-d/2} u$, and feasible vectors $a_i = \hat\chi_i$, so that $b_i = P(a_i) = 2^{-d/2} \hat \chi_i$.
Note that $\hat \cP \proves{\hat f}{4} \Norm{P(u)}^4 = 2^{-2d} \Norm{u}^4 = 1$, where the norm is the standard Euclidean norm, as used in \Cref{thm:sos-rounding}.
As argued above, with probability $1-o_d(1)$ each of $\hat \chi_1, \ldots, \hat \chi_d$ is feasible for $\hat\cP$.
Then, by \Cref{thm:sos-rounding}, there exists an algorithm that runs in time $(2^d)^{O(1)}$ and outputs with high probability a set of unit vectors $ v_1, \ldots,  v_d \in \R^{2^d}$ such that, up to a permutation of the vectors, $\Norm{ v_i^{\otimes 2} - \Paren{\frac{1}{\sqrt{2^d}} \hat \chi_i}^{\otimes 2}}_{\mathrm{op}} \leq O(\sqrt{d^{-11}})$.

\begin{claim}
If $v,u$ are unit vectors with $\Norm{v^{\otimes 2} - u^{\otimes 2}}_{\mathrm{op}} \leq \epsilon$ for some $\epsilon \leq 1$, then either $\Norm{v - u}^2 \leq O(\epsilon^2)$ or $\Norm{v + u}^2 \leq O(\epsilon^2)$.
\end{claim}
\begin{proof}
We have $\Norm{v^{\otimes 2} - u^{\otimes 2}}_{\mathrm{op}}^2 = 1 - \langle v, u\rangle^2$, so the assumption gives that $\langle v, u\rangle^2 \geq 1-\epsilon^2$, so either $\langle v, u\rangle \geq 1 - O(\epsilon^2)$ or $\langle v, u\rangle \leq -1+O(\epsilon^2)$.
In the former case we get $\Norm{v-u}^2 \leq O(\epsilon^2)$, and in the latter case we get $\Norm{v+u}^2 \leq O(\epsilon^2)$.
\end{proof}

Hence, for each $i \in [d]$, either $\Norm{ v_i - \frac{1}{\sqrt{2^d}} \hat \chi_i}^2 \leq O(d^{-11})$ or $\Norm{ v_i + \frac{1}{\sqrt{2^d}} \hat \chi_i}^2 \leq O(d^{-11})$.
Let $ \sigma_i \in \{\pm 1\}$ such that $\Norm{ v_i -  \sigma_i \frac{1}{\sqrt{2^d}} \hat \chi_i}^2 \leq O(d^{-11})$.
Then also $\Norm{\sqrt{2^d}  v_i -  \sigma_i \hat \chi_i}^2 \leq O(d^{-11})2^d$.
Let $ f_i = \mathrm{sign}(\sqrt{2^d} v_i) \in \{\pm 1\}^{V}$ for each $i \in [d]$, with entry $0$ rounded arbitrarily.

\begin{claim}
If $v, u$ are vectors and $u$ has all entries in $\Set{\pm 1}$, then $\Norm{\mathrm{sign}(v) - u}^2 \leq 4\Norm{v-u}^2$.
\end{claim}
\begin{proof}
We have $\Norm{\mathrm{sign}(v) - u}^2 = \sum_i (\mathrm{sign}(v_i) - u_i)^2$. It is easy to verify that $(\mathrm{sign}(v_i) - u_i)^2 \leq 4 (v_i - u_i)^2$ for all $v_i \in \R$ and $u_i \in \{\pm 1\}$. Then we also have $\sum_i (\mathrm{sign}(v_i) - u_i)^2 \leq 4 \sum_i (v_i - u_i)^2 = 4\Norm{v-u}^2$, completing the proof.
\end{proof}

Therefore $\Norm{ f_i -  \sigma_i \hat \chi_i}^2 \leq O(d^{-11})2^d \leq d^{-10} 2^d$.
Then we output $ f_1, \ldots,  f_d$.
Up to a permutation of the vectors, we have argued that these satisfy $\Norm{ f_i -  \sigma_i \hat \chi_i}^2 \leq d^{-10} 2^d$, so $\E_{v \sim \mathrm{Unif}(V)} (f_i(v) - \sigma_i \hat \chi_i(v))^2 \leq d^{-10}$. Since $\hat \chi_i(v)=z_i(v)$, this is the claimed bound.

\end{proof}

%% file: degree-2-lower-bound.tex
\section{Degree-\texorpdfstring{$2$}{2} lower bound}
\label{sec:degree-2-lower-bound}

In this section we prove that, in the regime $p = \Theta(\frac{\log d}{d})$, the degree-$2$ version of our relaxation can certify at most $\sum_{i \in [d]} \langle f, \chi_i\rangle^2 \geq 1 - 1/\polylog(d)$ instead of $1-1/\poly(d)$.
Hence, the degree-$2$ relaxation does not provide the Fourier
concentration needed by our FKN-based recovery argument.
Indeed, if the coordinate cuts are recovered with agreement $1-\varepsilon$, one can only deduce from this a feature vector recovery on a $1-d \varepsilon$ fraction of the vertices (see \Cref{lem:overview-latent-vector-recovery}). We remark that the following lower bound is not a formal proof that no algorithm based on degree-$2$ sum-of-squares can recover the latent feature vectors, but it means that our approach does need higher degree certificates.

\begin{theorem}[Degree-$2$ lower bound]
\label[theorem]{thm:degree-2-lower-bound}
In the same setting as \Cref{thm:coordinate-cut-recovery}, let $L'$ be the unnormalized Laplacian of $G^*$.
Consider the following system of polynomial inequalities in variables $\{f(x)\}_{x \in \{\pm 1\}^d}$:
\[
\cP = \begin{cases}
\frac{1}{p} \langle f, L' f \rangle \leq 2\,, \\
\sum_{x \in \Set{\pm 1}^d} f(x) = 0\,, \\
f(x)^2 = 1\,, \forall x \in \{\pm 1\}^d\,. \\
\end{cases}
\]
For $\frac{C \log d}{d} \leq p \leq 1/4$ for a large enough absolute constant $C$, with probability $1-o_d(1)$ there exists a degree-$2$ pseudoexpectation $\tilde\E$ that satisfies $\cP$ but $\tilde\E \sum_{i \in [d]}\langle f, \chi_i \rangle^2 \leq 1-\Omega(\frac{1}{pd})$.
\end{theorem}

\begin{remark}
In the sparse regime $p = \Theta(\frac{\log d}{d})$, the conclusion of \Cref{thm:degree-2-lower-bound} implies that\linebreak $\tilde\E \sum_{i \in [d]}\langle f, \chi_i \rangle^2 \leq 1-\Omega(\frac{1}{\log d})$.
\end{remark}

\begin{lemma}
\label[lemma]{lem:orthogonal-perturbation-exists}
In the same setting as \Cref{thm:degree-2-lower-bound}, with probability $1-o_d(1)$ there exists some $q : \Set{\pm 1}^d \to \R^d$ such that $q(x) = -q(-x)$, $\langle q(x), x\rangle =0$, and $\Norm{q(x)}^2=1$ for all $x \in \Set{\pm 1}^d$, and in addition
\[\Iprod{L'\frac{x}{\sqrt{d}},q} \leq -\frac{\sqrt{p}}{5}\,,\]
where $x: \Set{\pm 1}^d \to \R^d$ denotes the identity function (in the sense that each $x \in \{\pm 1\}^d$ is mapped to $x$ as a real-valued vector).
\end{lemma}
\begin{proof}
Let $E$ be the edge set of $G^*$.
For each $x \in \{\pm 1\}^d$ and $i \in [d]$, let
\begin{equation*}
    \eta_i(x) = \ind(\{x,x^{\sim i}\} \in E)+\ind(\{-x,-x^{\sim i}\} \in E) \, .
\end{equation*}
Then, we define vectors $h(x) \in \R^d$ as having $(h(x))_i = \eta_i(x) x_i$ for each $i \in [d]$.
We also define $g(x)$ to be the projection of $h(x)$ onto the subspace orthogonal to $x$, and
\begin{equation*}
    q(x) = \begin{cases}
        -\frac{g(x)}{\|g(x)\|}, & \text{if } g(x) \neq 0\,,\\
        u(x), \quad & \text{if } g(x)=0\,, \\
    \end{cases}
\end{equation*}
where $u(x)$ is any odd unit vector orthogonal to $x$, e.g., $u(x) = \frac{x_2 e_1 - x_1 e_2}{\sqrt{2}}$.
Then, it is clear that $q(x) = -q(-x)$, $\langle q(x), x\rangle =0$, and $\Norm{q(x)}^2=1$ for all $x \in \Set{\pm 1}^d$. Now we calculate
\begin{align*}
    \|g(x)\|^2 & = \sum_{i =1}^d \eta_i(x)^2 - \frac{1}{d}\left(\sum_{i=1}^d \eta_i(x)\right)^2 \\
    & = \frac{1}{2d}\sum_{i,j =1}^d \left(\eta_i(x)-\eta_j(x)\right)^2 \\
    & = \frac{1}{d}\sum_{1\le i<j \le d} \left(\eta_i(x)-\eta_j(x)\right)^2\,.
\end{align*}
For a fixed $x$, the expected number of coordinates $i \in [d]$ such that $\eta_i(x) \ge 1$ is $(2p(1-p)+p^2)d \ge pd$. Since the random variables $(\ind(\eta_i(x) \ge 1))_{i=1}^d$ are independent, we have at least $pd/2$ such coordinates with probability $1-1/d^{10}$. Similarly, the expected number of coordinates $i \in [d]$ such that $\eta_i(x) =0$ is $(1-p)^2 d \ge d/2$. Since the random variables $(\ind(\eta_i(x) = 0))_{i=1}^d$ are independent, we have at least $d/3$ such coordinates with probability $1-1/d^{10}$.

When both these events occur, we have $\|g(x)\|^2 \ge pd/6$. Hence, the probability that $\|g(x)\| < \sqrt{pd/6}$ is at most $1/d^9$. By Markov's inequality we then conclude that $\|g(x)\| \ge \sqrt{pd/6}$ for at least half of the vertices $x$ with probability at least $1-1/d^8$. In this case, we get that
\begin{equation*}
    \frac{1}{2^d}\sum_{x \in \{\pm 1\}^d} \|g(x)\| \ge \sqrt{pd/24}\,.
\end{equation*}
Finally, we conclude by observing
\begin{align*}
    (L' \chi_i)(x) & =  \sum_{j =1}^d \ind(\{x,x^{\sim j}\} \in E)\cdot \chi_i(x) - \sum_{j =1}^d \ind(\{x,x^{\sim j}\} \in E) \cdot \chi_i(x^{\sim j})  \\
    & = \sum_{j =1}^d \ind(\{x,x^{\sim j}\} \in E) \left(\chi_i(x)- \chi_i(x^{\sim j})\right) \\
    & = 2 \cdot \ind(\{x,x^{\sim i}\} \in E) \cdot \chi_i(x) \, ,
\end{align*}
so $L' x$ is the vector $v(x)$ whose $i$-th coordinate is $\ind(\{x,x^{\sim i}\} \in E) \cdot 2x_i$.
Then, we have
\begin{align*}
    \Iprod{ \frac{x}{\sqrt{d}}, L' q} & = \Iprod{ q, L' \frac{x}{\sqrt{d}}} \\
    & = \frac{1}{\sqrt{d}}\frac{1}{2^d} \sum_x \langle q(x),v(x) \rangle \\
    & = \frac{1}{\sqrt{d}}\frac{1}{2^{d+1}} \left( \sum_x \langle q(x),v(x)\rangle - \sum_x \langle q(-x),v(x) \rangle\right) \\
    & = \frac{1}{\sqrt{d}}\frac{1}{2^{d+1}} \left( \sum_x \langle q(x),v(x) \rangle - \sum_x \langle q(x),v(-x)\rangle\right) \\
    & = \frac{1}{\sqrt{d}} \frac{1}{2^{d+1}} \sum_x \langle q(x),v(x)-v(-x) \rangle   \, .
\end{align*}
Note that $\frac{1}{2}(v(x)-v(-x))=h(x)$. Moreover, $q(x)$ is already in the plane orthogonal to $x$, so $\langle q(x), h(x)\rangle = \langle q(x), g(x)\rangle$. Hence, continuing from above we conclude
\begin{align*}
    \Iprod{ \frac{x}{\sqrt{d}}, L' q} =  \frac{1}{\sqrt{d}} \frac{1}{2^{d}} \sum_x \langle q(x), g(x)  \rangle = -\frac{1}{\sqrt{d}} \frac{1}{2^{d}} \sum_x \|g(x)\| \le -\sqrt{p/24} \le -\sqrt{p}/5 \, .
\end{align*}
\end{proof}

\begin{proof}[Proof of \Cref{thm:degree-2-lower-bound}]
The degree-$2$ sum-of-squares program in \Cref{thm:degree-2-lower-bound} is equivalent to the following SDP in variable $X \in \R^{2^d \times 2^d}$:
\[\begin{cases}
\frac{1}{p} \langle X, L'\rangle \leq 2\,,\\
1^\top X 1 = 0\,,\\
\diag(X) = 1\,,\\
X \succeq 0\,,
\end{cases}\]
where the inner product is normalized by $\frac{1}{2^d}$.
Let $P_{=1} = \frac{1}{2^d} \sum_{i \in [d]} \chi_i \chi_i^\top$.
Then the conclusion of \Cref{thm:degree-2-lower-bound} is equivalent to saying that there exists some $X$ feasible for this program with $\langle X, P_{=1} \rangle \leq 1-\Omega(\frac{1}{pd})$.
Let us now construct such an $X$.

We will first construct some $F : \Set{\pm 1}^d \to \R^m$ such that $X \in \R^{2^d \times 2^d}$ defined with $X_{x,y} = \langle F(x), F(y)\rangle$ is feasible for the SDP with $\frac{1}{p} \langle X, L'\rangle \leq 2 - \Omega\Paren{\frac{1}{pd}}$.
Then, by taking a convex combination between $X$ and $\chi_{\Set{1,2}}\chi_{\Set{1,2}}^\top$, the result will follow.
We note that, for an $X$ defined as above, we have $\langle X, L'\rangle = \frac{1}{2^d} \sum_{x \sim_{G^*} y} \Norm{F(x)-F(y)}^2$.

Let $q$ be as in \Cref{lem:orthogonal-perturbation-exists}.
Let $S \subseteq [d]$ be any set with odd cardinality such that $|S| \geq d- 1$.
Then, for $r = \frac{1}{16d\sqrt{p}} \le 1/2$, we define
\[F'(x) = \frac{1}{\sqrt{d}} x + \frac{r}{\sqrt{2}}q(x) \in \R^d\]
and
\begin{equation*}
F(x)=
\frac{1}{\sqrt{1+r^2}}\begin{pmatrix}
F'(x)\\
\frac{r}{\sqrt{2}} \cdot \chi_S(x)
\end{pmatrix}  \in \R^{d+1}\, .
\end{equation*}

Let us verify the properties of $F$ and the resulting $X$.
First,
\begin{align*}
1^\top X 1 = \sum_{x, y \in \Set{\pm 1}^d} \langle F(x), F(y)\rangle = \Norm{\sum_{x \in \Set{\pm 1}^d} F(x)}^2 = 0\,,
\end{align*}
because $F(x) = -F(-x)$. Second, it is easy to verify by definition that
\begin{align*}
X_{x,x} = \Norm{F(x)}^2 = 1\,,
\end{align*}
since $q(x)$ is a unit-norm vector that is orthogonal to $x$.
It remains to understand $\frac{1}{p}\langle X, L'\rangle$.
We have
\begin{align*}
\langle X, L'\rangle
&= \frac{1}{2^d} \sum_{x \sim_{G^*} y} \Norm{F(x)-F(y)}^2\\
&\le  \frac{1}{2^d}\frac{1}{1+r^2} \sum_{x \sim_{G^*} y} \Norm{F'(x)-F'(y)}^2 + \frac{1}{2^d} \frac{2r^2}{1+r^2}|E|\,.
\end{align*}
Now we rewrite
\begin{align*}
    \frac{1}{2^d} \sum_{x \sim_{G^*} y} \Norm{F'(x)-F'(y)}^2 &  = \langle F', L' F' \rangle \\
    & = \frac{1}{d} \langle x, L' x \rangle + \frac{r^2}{2}\langle q, L' q \rangle + \frac{\sqrt{2} \cdot r} {\sqrt{d}} \langle x, L' q \rangle \\
    & \le \frac{1}{d} \langle x, L' x \rangle + {2r^2} \frac{|E|}{2^d} + \frac{\sqrt{2} \cdot r} {\sqrt{d}} \langle x, L' q \rangle\,.
\end{align*}
By the guarantee in \Cref{lem:orthogonal-perturbation-exists}, we have $\Iprod{ \frac{x}{\sqrt{d}}, L' q} \le -\frac{\sqrt{p}}{5}$ and also $\langle x, L' x \rangle = \frac{4|E|}{2^d}$.
Putting everything together, we get
\begin{align*}
    \langle X, L'\rangle \le \left(\frac{1}{d(1+r^2)}+\frac{r^2}{1+r^2}\right) \frac{4|E|}{2^d} - \frac{\sqrt{2}r\sqrt{p}}{5(1+r^2)}\,.
\end{align*}
Since $|E| \le p(1+r^2)d 2^d /2$ with high probability, we have
\begin{align*}
    \langle X, L'\rangle \le 2p+ 2dr^2p-\frac{r\sqrt{p}}{5} = 2p + \frac{1}{128 d}-\frac{1}{80d} \le 2p - 1/(1000d) \, .
\end{align*}

We define now the matrix that will give the desired conclusion: $X^* = (1-\epsilon)X + \epsilon \chi_{\Set{1,2}} \chi_{\Set{1,2}}^\top$ for some ${\epsilon \in (0, 1)}$ to be selected.
Standard concentration bounds give with probability $1-o_d(1)$ that\linebreak $\frac{1}{p}\langle \chi_{\Set{1,2}}, L' \chi_{\Set{1,2}}\rangle \leq O(1)$.
Then, for $\epsilon = \Omega\Paren{\frac{1}{pd}}$ small enough, $X^*$ is also feasible for the SDP, because $\frac{1}{p}\langle X^*, L'\rangle \leq (1-\epsilon) \Paren{2 - 1/(1000pd)} + O(\epsilon) \leq 2$ and $1^\top (\chi_{\Set{1,2}} \chi_{\Set{1,2}}^\top) 1 = 0$, $\diag(\chi_{\Set{1,2}} \chi_{\Set{1,2}}^\top) = 1$, and $\chi_{\Set{1,2}} \chi_{\Set{1,2}}^\top \succeq 0$.
Next, we claim that $\langle X, P_{=1}\rangle \leq 1$.
This follows because
\begin{align*}
\langle X, P_{=1}\rangle
&= \frac{1}{2^{2d}} \sum_{i \in [d]} \chi_i^\top X \chi_i\\
&= \frac{1}{2^{2d}} \sum_{i \in [d]} \sum_{x, y \in \Set{\pm 1}^d} \langle F(x) \chi_i(x), F(y) \chi_i(y)\rangle\\
&= \frac{1}{2^{2d}} \sum_{i \in [d]} \Norm{\sum_{x \in \Set{\pm 1}^d} F(x) \chi_i(x)}^2\\
&= \frac{1}{2^{2d}} \sum_{i \in [d]} \sum_{a \in [d+1]} \Paren{\sum_{x \in \Set{\pm 1}^d} F_a(x) \chi_i(x)}^2\\
&= \sum_{i \in [d]} \sum_{a \in [d+1]} \langle F_a, \chi_i\rangle^2\,,
\end{align*}
where $F_a = (F_a(x))_{x \in \Set{\pm 1}^d}$.
By Parseval, $\sum_{S \subseteq [d]} \langle F_a, \chi_S\rangle^2 = \E_x F_a(x)^2$, and also $\sum_{a \in [d+1]} \E_x F_a(x)^2 = \E_x \Norm{F(x)}^2 = 1$, so it also follows that $\langle X, P_{=1}\rangle \leq 1$.
Then, finally, using also that $\langle \chi_{\Set{1,2}} \chi_{\Set{1,2}}^\top, P_{=1}\rangle = 0$, we get that $\langle X^*, P_{=1} \rangle \leq 1-\epsilon \leq 1 - \Omega\Paren{\frac{1}{pd}}$.

\end{proof}